%% file: sn_main_arxiv.tex
\documentclass[pdflatex,sn-mathphys]{sn-jnl}% Math and Physical Sciences Reference Style
\jyear{2023}%

\theoremstyle{thmstyleone}%
\newtheorem{theorem}{Theorem}%  meant for continuous numbers
\theoremstyle{thmstyletwo}%
\newtheorem{example}{Example}%
\theoremstyle{thmstylethree}%
\newtheorem{Assumption}{\bf Assumption}

\usepackage{color}
\usepackage{amsmath,amssymb,amsfonts,amstext,amsthm}  % math
\usepackage{cancel}
\usepackage{xr}
\usepackage{anyfontsize}  % fix font size warnings

\makeatletter
\newcommand*{\addFileDependency}[1]{
  \typeout{(#1)}
  \@addtofilelist{#1}
  \IfFileExists{#1}{}{\typeout{No file #1.}}
}
\makeatother

\begin{document}

\title[Constructing Custom Thermodynamics Using Deep Learning]{Constructing Custom Thermodynamics Using Deep Learning}

\author[1,2]{\fnm{Xiaoli} \sur{Chen}}\email{xlchen@nus.edu.sg}
\equalcont{These authors contributed equally to this work.}
\author[3]{\fnm{Beatrice W.} \sur{Soh}}\email{beatrice\_soh@imre.a-star.edu.sg}
\equalcont{These authors contributed equally to this work.}
\author[3]{\fnm{Zi-En} \sur{Ooi}}\email{ooize@imre.a-star.edu.sg}
\author[4]{\fnm{Eleonore} \sur{Vissol-Gaudin}}\email{eleonore.vg@ntu.edu.sg}
\author[5,6]{\fnm{Haijun} \sur{Yu}}\email{hyu@lsec.cc.ac.cn}
\author*[2]{\fnm{Kostya S.} \sur{Novoselov}}\email{kostya@nus.edu.sg}
\author*[2,3,4]{\fnm{Kedar} \sur{Hippalgaonkar}}\email{kedar@ntu.edu.sg}
\author*[1,2]{\fnm{Qianxiao} \sur{Li}}\email{qianxiao@nus.edu.sg}
\affil[1]{\orgdiv{Department of Mathematics}, \orgname{National University of Singapore}, \orgaddress{\postcode{119077}, \country{Singapore}}}
\affil[2]{\orgdiv{Institute for Functional Intelligent Materials}, \orgname{National University of Singapore}, \orgaddress{\postcode{117544}, \country{Singapore}}}
\affil[3]{\orgdiv{Institute of Materials Research and Engineering}, \orgname{A*STAR (Agency for Science)}, \orgaddress{\street{2 Fusionopolis Way}, \postcode{138634}, \country{Singapore}}}
\affil[4]{\orgdiv{School of Materials Science and Engineering}, \orgname{Nanyang Technological University}, \orgaddress{\postcode{639798},\country{Singapore}}}
\affil[5]{\orgdiv{LSEC \& ICMSEC, Academy of Mathematics and Systems Science}, \orgname{Chinese Academy of Sciences}, \orgaddress{\city{Beijing}, \postcode{100190},\country{China}}}
\affil[6]{\orgdiv{School of Mathematical Sciences}, \orgname{University of Chinese Academy of Sciences}, \orgaddress{\city{Beijing}, \postcode{100049},\country{China}}}

\abstract{
One of the most exciting applications of artificial intelligence (AI) is automated scientific discovery based on previously amassed data, coupled with restrictions provided by known physical principles, including symmetries and conservation laws. Such automated hypothesis creation and verification can assist scientists in studying complex phenomena, where traditional physical intuition may fail. Here we develop a platform based on a generalized Onsager principle to learn macroscopic dynamical descriptions of arbitrary stochastic dissipative systems directly from observations of their microscopic trajectories. Our method simultaneously constructs reduced thermodynamic coordinates and interprets the dynamics on these coordinates. We demonstrate its effectiveness by studying theoretically and validating experimentally the stretching of long polymer chains in an externally applied field. Specifically, we learn three interpretable thermodynamic coordinates and build a dynamical landscape of polymer stretching, including the identification of stable and transition states and the control of the stretching rate. Our general methodology can be used to address a wide range of scientific and technological applications.
}


\maketitle

The modern scientific method adopts a universal approach that ensures stable and
non-conflicting progression of our understanding of nature: new theories need
to be hypothesized and tested on previously amassed data, be compatible with the
basic scientific principles and be verifiable by experiments.
Unfortunately, there is no general algorithmic
recipe to do so in complex systems to facilitate discovery.
Hence, up to now, only the most basic physical phenomena
-- often static, in equilibrium -- are
described by an intuitive set of equations.
Many dynamic, non-equilibrium
phenomena, which determine functionality in biology, soft-condensed matter, and
chemistry, are instead described via very approximate, empirical laws.
The advancement of artificial intelligence and
machine learning gives rise to the possibility of a data-driven solution to this
challenge~\cite{noh2020machine,hippalgaonkar2023knowledge}.

In this paper, we develop Stochastic OnsagerNet, an AI platform which can
discover an interpretable and closed thermodynamic description of an arbitrary
stochastic dissipative dynamical system directly from observations of
microscopic trajectories.
There are essentially two types of approaches
to understand and predict the behavior
of dynamical processes from data
-- unstructured and structured.
Unstructured approaches parametrizes dynamical equations
by a generic set of building blocks, such as
fixed polynomials~\cite{brunton2016discovering},
trainable feature maps or kernels~\cite{hamzi2021learning,dietrich2021learning},
and determines the associated parameters that best fit the observations.
Physical insights can be incorporated as regularizers in the fitting process~\cite{raissi2019physics}.
Their generality comes at a cost of long-time predictive accuracy,
stability~\cite{yu2020onsagernet}, and more importantly, interpretability.
This is addressed by the class of structured approaches,
where physical insights directly guide the design of model architectures.
Our approach belongs to this latter
category. Previous work in this direction include models based on
Hamiltonian or symplectic
dynamics~\cite{hesthaven2022rank,valperga2022learning},
Poisson systems~\cite{jin2022learning} and quasi-potentials~\cite{lin2022data}.
However, to date there lacks a general structured approach
to model dissipative, non-equilibrium and noisy dynamics
that often arise in soft matter, biophysics and other applications.
Our methodology based on the classical
Onsager principle~\cite{onsager1931re1,onsager1931re2}
is tailored to such problems.

Macroscopic thermodynamic descriptions
of physical systems are highly sought after
for the insights they provide.
A prototypical example is the ideal gas law
as a macroscopic description of non-interacting gas systems.
These guide the design of verification experiments and
provides principled ways to manipulate macroscopic behavior.
For a general complex dynamical system,
however, constructing an intuitive thermodynamic description that enables
subsequent analysis and control is a daunting task.
Our approach addresses this challenge as follows.
For a given microscopic dynamics, we learn a macroscopic thermodynamic description
via the simultaneous construction of low-dimensional closure coordinates
-- ensured to be partially interpretable -- and a time evolution law on these coordinates.
Unlike general AI approaches, our platform intrinsically limits
the search to physically relevant evolution laws.
In particular, we ensure compatibility with existing scientific knowledge
by constructing our neural network architecture based on a generalized Onsager principle.

We demonstrate our method by learning the
stretching dynamics of polymer chains containing up to 900
degrees of freedom,
condensing it into a thermodynamic description involving
only three macroscopic coordinates
that governs polymer stretching dynamics
in both computational and experimental data.
We build an energy landscape of the macroscopic evolution, revealing
the presence of stable and transition states. This can be viewed as a dynamic
equation of state. Mastering such an equation allows the design of verification
computational experiments, including the interpretation of the thermodynamic coordinates and
the control of the stretching rate of the polymers.
We extend this further to conduct
single-molecule DNA stretching experiments and show that our thermodynamic description
can be used to distinguish fast and slow stretching polymers,
much beyond current human-labelling capabilities.
Furthermore, the predicted fluctuation correlations derived from the
free energy landscape agree with experimental data.

Constructing low-dimensional physical
models from high-dimensional dynamical data is an active area of research.
Data-driven modeling of dynamical processes based on the Onsager's principle
was proposed in Yu et al.~\cite{yu2020onsagernet} to study
Rayleigh-B\'{e}nard convection.
Chen et al.~\cite{chen2022.AutomatedDiscoveryFundamental}
combine encoder-decoders and manifold learning
to construct latent dynamical models directly from video data,
including those of reaction-diffusion processes and pendulum motions.
Here, we make several advancements in terms of methodology
and applications.
First, unlike the deterministic models considered previously,
we explicitly capture stochastic fluctuations --
an important element of non-equilibrium
processes at finite temperatures.
In fact, the non-trivial heterogeneity of polymer stretching dynamics
studied in this paper is directly caused by thermal fluctuations.
Developing our method in the stochastic setting
requires extensions of model reduction theory
and training algorithms (see theoretical results and
model implementation in Methods).
Second and more importantly,
we go beyond dimensionality reduction~\cite{yu2020onsagernet,noe2020machine,chen2022.AutomatedDiscoveryFundamental}
and solve a closure problem:
given \textit{a priori} fixed macroscopic variables of interest
(e.g. polymer extension),
we construct both the closure coordinates
sufficient to govern the evolution of these macroscopic variables,
and the dynamics that describes this evolution.
Compared with more flexible parametrizations of
reduced dynamics~\cite{chen2022.AutomatedDiscoveryFundamental},
our approach inherently limit the search space
to those satisfying a generalized Onsager principle,
which sacrifices complete generality
but facilitates physical interpretation of the closure coordinates
and the dynamical landscape.

\section*{Results}

We now describe our approach.
The most complete description of a complex, multi-component system is the
coordinates of all the components as a function of time $t$ ($X(t)$). For the ideal
gas it would be the positions and momenta of all molecules and for a magnetic
system -- the spin state of each atom. An alternative to this expensive
microscopic modelling approach is a thermodynamic one, where the full
description is replaced by some macroscopic coordinates ($Z^{*}(t)$, with
dimensionality much smaller than $X(t)$). This can be the pressure of an ideal
gas, or the magnetization of a magnetic system. The thermodynamic approach
links these macroscopic coordinates to other macroscopic coordinates, or
closure variables ($\hat{Z}(t)$),
and external parameters (volume and temperature for ideal gas, magnetic field and temperature for magnetic systems)
via an equation of state.

We propose a generic approach of building such custom thermodynamics for an arbitrary stochastic,
dissipative dynamical system from data.
We are given macroscopic coordinates $Z^*$ whose dynamics we wish to model.
For polymer dynamics this can be a single variable -- the extension of the
polymer chain, see Fig.~\ref{fig:workflow}.
Then, we learn a set of closure variables $\hat{Z}(t)$
and simultaneously an evolution law on the combined thermodynamic
coordinates $Z(t)=(Z^*(t),\hat{Z}(t))$
that enables scientific understanding, experimental verification, and control.
The evolution equation is a generalization of
the classical Onsager principle~\cite{onsager1931re1,onsager1931re2}
that has been successfully applied to model a variety of non-equilibrium phenomena,
including phase separation kinetics, gel dynamics and
molecular modeling~\cite{doi2011onsager,doi2015onsager}.
It posits a time evolution law
\begin{align}
  \dot{Z}(t) &= -M \nabla V(Z(t)),
 \end{align}
for a given set of coordinates $Z(t)$, where the symmetric positive
semi-definite matrix $M$ models energy dissipation and $V$ is a generalized
potential.
A limitation of the Onsager principle is its inability to capture dynamics
far from equilibrium, or with substantial stochastic behavior. To this end, we
propose an extension in the form of a \textit{generalized stochastic Onsager
principle}
\begin{equation}\label{eq:sto_onsa}
   \dot{Z}(t) =
   -[M(Z(t))+W(Z(t))]\nabla V(Z(t)) + \sigma(Z(t)) \dot{B}(t),
\end{equation}
where $M(\cdot), W(\cdot)$ are now functions of the reduced coordinates $Z$
that output $d\times d$ matrices.
$M(\cdot)$ is symmetric positive semi-definite
to conform to stability requirements and Onsager's reciprocal relations~\cite{onsager1931re1},
while $W(\cdot)$ is anti-symmetric and models conservative forces. The diffusion
matrix $\sigma(\cdot)$ together with the white noise process $\dot{B}(t)$
models the thermal fluctuations in the system. Eq.~\eqref{eq:sto_onsa} forms
the basis of our dynamical model in reduced coordinates.
We note that alternative generalizations of the Onsager principle
have been proposed using large deviations theory~\cite{mielke2014.RelationGradientFlows},
but their forms are more complex and hence
less amenable to computations.
It can be shown that our model has long-time stability through
energy dissipation up to the order of thermal fluctuations
(Thm.~\ref{thm:appen_stability})
and the flexibility to represent many physical stochastic processes,
including Langevin and generalized Poisson dynamics
(see theoretical results in Methods).
Our method departs from classical modelling paradigms, where the unknown equation
parameters are few and can be fitted from few experiments. Instead, the unknowns
here are functions $M,W,V,\sigma$. We leverage
machine learning and represent these functions as trainable deep neural networks, while
preserving the required physical constraints (e.g. symmetric positive
definiteness of $M$). Simultaneously, we generate a set of the
closure coordinates by another deep neural network, which combines approximation
flexibility of residual networks and approximate feature orthogonality
through principal component analysis.
This is to be contrasted with generic coarse-graining methods
based on volume averages~\cite{tanaka2019multi},
in that we seek a very small set of closure coordinates that
are sufficient to describe the motion of the macroscopic states of interest.
Our learning-based approach to discover hidden coordinates
shares some similarity with recently proposed
machine learning-based coarse-graining methods in molecular
simulations~\cite{noe2020machine,friederich2021machine}
but here we work with a closure problem,
thus we may end up with macroscopic dynamics
of substantially lower dimensions.
We perform end-to-end training of the combined architecture on large-scale microscopic trajectory data
to simultaneously learn the reduced coordinates and their dynamics.
The overall workflow for creating custom thermodynamics, which we call \textit{Stochastic
OnsagerNet} (S-OnsagerNet), is summarized in Fig.~\ref{fig:workflow}. Detailed
network architectures and training algorithms are found in Methods.

\paragraph*{Training and prediction of polymer stretching dynamics}

We first demonstrate our approach by modelling the temporal evolution of
polymer extension under elongational forces, which has long been of interest to
the polymer physics
community~\cite{Bird1987,Doi1988,Larson1999structure,McKinley2003}. Hallmark
experimental~\cite{Perkins1997,Smith1998} and computational
studies~\cite{Larson1999,Jendrejack2002,Hsieh2003} in elongational rheology of
dilute polymers have examined the deformation of single DNA molecules in planar
elongational flows and revealed the highly heterogeneous stretching dynamics
among identical polymer chains. Due to the complex interactions within and
stochastic nature of the system, it is challenging to identify macroscopic
descriptors of the polymer chain (closure coordinates) and governing equations
on these descriptors that are sufficient to determine the outcome of the
stretching dynamics. Yet, such a thermodynamic description is essential for
understanding the origins of unfolding heterogeneity and paves the way to make
desired modifications to the unfolding dynamics. Thus, our data-driven method
offers a promising alternative to achieve this goal.

We simulate polymer chain stretching in a planar elongational
flow. The polymer chain consists of 300 coarse-grained beads
connected by rigid rods, resulting in 900 degrees of freedom if we ignore inertial effects
(Fig.~\ref{fig:unfolding_times}(A,B)). Snapshots of the shape of the chains
under stretching conditions are shown in Fig.~\ref{fig:unfolding_times}(C),
revealing highly heterogeneous dynamics of the
chain extension (Fig.~\ref{fig:unfolding_times}(C,D)), defined as the projected
chain length along the elongational axis of the flow. This is our
macroscopic coordinate of interest $Z^*(t)$.
Our aim is to model its stochastic evolution and
understand the origin of its heterogeneity.
We train the S-OnsagerNet on this dataset following the workflow in Fig.~\ref{fig:workflow}.
The network architecture selection and training procedures are found in Methods.
Our approach constructs 2 closure coordinates in addition to the chain extension $Z^* (t)$,
leading to a 3 dimensional dynamical system -- following Eq.~\eqref{eq:sto_onsa}
with learned functions $M,W,V$ and $\sigma$ -- that governs the dynamics of
stretching. We have empirically chosen the number of macroscopic coordinates:
using more than 3 did not substantially improve predictive accuracy, whereas a
2-variable system has modelling limitations due to physical symmetry. The
detailed selection procedure of the reduced coordinate dimension is discussed
under polymer dynamics analysis in methods.

In Fig.~\ref{fig:unfolding_times}(E-M), we test the trained S-OnsagerNet on
three unseen, different and representative initial polymer configurations. The
selected chains start with similar extension lengths, but subsequently stretch
at vastly different rates. Fig.~\ref{fig:unfolding_times}(E-G,I-K,M-P) show
that the true statistics (black) can be accurately predicted (red). Moreover,
the distributions of the time taken to reach a reference extension length are
successfully captured (Fig.~\ref{fig:unfolding_times}(H,I,M)).

\paragraph{Interpreting learned closure coordinates}
Having shown that only two closure coordinates $\hat{Z}=(Z_2,Z_3)$ are required
to characterize the stochastic evolution of the extension length $Z^*=Z_1$, it
is natural to probe the meaning of these discovered coordinates.
Here, we utilize an intrinsic property of neural networks -- it represents the non-linear
reduction functions $X\mapsto Z$ as differentiable maps, as if we have learned
their analytical forms. We compute via automatic differentiation the
perturbations on a generic microscopic configuration $X$ in the directions of
$\pm \partial Z_2 / \partial X$ and $\pm \partial Z_3 / \partial X$,
corresponding to conformations with steepest changes in $Z_2$ and $Z_3$
respectively. The resulting conformational changes suggest physical
interpretation of these coordinates. For example, from
Fig.~\ref{fig:Z2_Z3_meaning}(A), we observe that perturbations in the direction
of $\partial Z_2 / \partial X$ tend to change the end-to-end distance in the
elongational axis, or distance between the first and the last bead in the
polymer chain along the elongational axis. We confirm this hypothesis by visualizing
the correlation of the end-to-end distance and the magnitude of $Z_2$ in
Fig.~\ref{fig:Z2_Z3_meaning}(B,C). A similar analysis reveals the correlation
between $Z_3$ and a degree of foldedness of the chain in the elongational axis
of the flow (Fig.~\ref{fig:Z2_Z3_meaning}(D-F)).

\paragraph{Free energy landscape}
The constructed potential $V$ can be interpreted as a
generalized free energy, allowing us to gain important insights into the dynamical landscape.
The local minima of $V$ represent stable or metastable states, while the saddle
points correspond to transition states.
The differentiable representation of $V$ enables us to probe this landscape.
Fig.~\ref{fig:potential} shows 2-dimensional projections of the 3-dimensional
free energy $V(Z_1,Z_2,Z_3 )$. We identify the critical points of $V$ by solving
$\nabla V(Z) = 0$ using the Broyden-Fletcher-Goldfarb-Shanno (BFGS) method.
We found two stable fixed points and two saddle points of interest marked in
Fig.~\ref{fig:potential}. Using a simultaneously trained PCA-ResNet decoder,
we can reconstruct the macroscopic spatial coordinates of the
polymer chain at the critical points to identify their physical origin. Up to
reflection symmetry in the elongational axis, the stable points correspond to
fully stretched states, whereas the saddle points refer to completely folded
states. The origin of the heterogeneity in unfolding times is now clear: a
rapidly stretching polymer is the one that avoids the saddle point and goes
directly to a stable minimum, whereas a slowly stretching chain gets trapped
around the stable manifold (attractive part) of the saddle point for a long
time, before finally escaping through the unstable manifold (repulsive part) of
the saddle towards the stable minimum (see Supplementary Video~1).
We confirm this by
overlaying a fast and a slow unfolding trajectory with the potential landscape
in Fig.~\ref{fig:potential}(C,F). Despite the similarity between the initial
chain configurations, as demonstrated by the proximity between the initial
points on the potential energy landscape, the chains exhibit different
stretching behaviors that can be rationalized by the constructed potential.

Moreover, a Taylor expansion via automatic differentiation of the learned $V(Z)$
captures the leading-order fluctuations near a stable stretched state.  We
denote by $\delta V$ a typical energy fluctuation around the stretched state
(proportional to temperature). Then, ignoring small terms we find
\begin{equation}\label{Eq.NUM}
   \delta V
   \approx
   a_1\,(\delta Z_1 - a_4 \, \delta  Z_2)^2
   +
   a_2\,\delta Z_2^2
   +
   a_3\,\delta Z_3^2,
\end{equation}
where $a_1=153.1$, $a_2=205.5$, $a_3=36.96$, $a_4=1.54$ and $\delta
Z_i=Z_i-[Z_{\text{stable},1}]_i$ is the fluctuations in the thermodynamic
variables.  Note that the coefficients $a_j$ implicitly depend on the strength
of the flow.  Eq.~\eqref{Eq.NUM} is an effective equation of state, from which
we observe the positive correlation of $Z_1$ and $Z_2$. The physical
interpretation is that near the stretched state, the chain extension and the
end-to-end distance tend to change simultaneously
(see Fig.~\ref{fig:potential}(G-I) and polymer dynamics analysis in Methods).

\paragraph{Controlling polymer stretching dynamics}

Further, understanding the laws of the custom thermodynamics of polymer chain
folding allows us to interact with the dynamics by designing controls over the
polymer environment to initiate desired changes in its behavior. To this end, we
perform another Taylor expansion of $V$ -- this time near a saddle point
$Z_\text{saddle,2}$ corresponding to a folded state to give another local
equation of state
\begin{equation}\label{Eq.NUM2}
   \delta V
   \approx
   b_1 \, \delta Z_1^2
   -
   b_2 \, (\delta Z_2 - b_4\,\delta  Z_3)^2
   +
   b_3 \, \delta Z_3^2,
\end{equation}
where $b_1=102.96$, $b_2=31.13$, $b_3=24.16$
and $b_4=0.255$.
Eq.~\eqref{Eq.NUM2} suggests that to escape this saddle point leading to polymer
unfolding, it is most effective to increase end-to-end distance ($Z_2$) while
decreasing foldedness ($Z_3$) in a proportional way. This leads to a data-driven
control protocol in Fig.~\ref{fig:design}(A,B). We choose the external
elongational flow as the only control parameter (in real experiments it
corresponds to switching on and off the flow of fluid or the electric
field~\cite{Smith1998}). We start with a polymer configuration near the saddle
point of the energy landscape, corresponding to a folded state. Without any
intervention, the subsequent unfolding is expected to follow a slow trajectory,
staying near the folded state for a long time. From our landscape analysis
above, the most effective escape from the saddle point is along its unstable
manifold – approximately corresponding to increasing $\lvert  \delta Z_2-
b_4\,\delta Z_3 \rvert $. Thus, to speed up unfolding we can design the
following control strategy: we turn off the external elongational flow, so the
polymer drifts randomly under Brownian forces around the saddle point. We track
its reduced coordinates, and once we observe sufficient alignment with the
unstable manifold, we turn on the externally applied elongational flow. We
observe in Fig.~\ref{fig:design}(C) that this simple control system speeds up
the unfolding dynamics substantially. We can also increase the
unfolding time by reversing this protocol (Fig.~\ref{fig:design}(D,E)).
These control strategies based on the learned thermodynamic description
have notable advantages over classical model-free control regimes (e.g.
reinforcement learning), which may require large exploration times or
small, finite state spaces~\cite{sutton2018reinforcement}.

\paragraph{Experimental validation}

Remarkably, some qualitative predictions of the constructed thermodynamic
description are confirmed by physical experiments.
Not only can one show that the constructed dynamical landscape
allows for fine-grained classification of stretching behavior on simulation data
(see polymer dynamics analysis in Methods),
we demonstrate in
Fig.~\ref{fig:experimental} that this applies directly to physical experiments.
Here, we perform single-molecule experiments to observe the stretching
trajectories of DNA molecules in a planar elongational field
(see data preparation in Methods).
We select two samples that initially appear similar
(Fig.~\ref{fig:experimental}(E,G)), making it impossible to visually
distinguish them in terms of stretching behavior.  We then cast them into the
learned thermodynamic coordinates $Z$,
which when super-imposed on the free energy landscape reveals that the $Z_2$
coordinates differ subtly, leading to different predicted stretching statistics
(Fig.~\ref{fig:experimental}(I)).  This substantially improves upon human-level
labelling, which can only occur much later in the dynamical evolution
(Fig.~\ref{fig:experimental}(H,J)).  Furthermore, we show in
Fig.~\ref{fig:experimental}(J) that the effective equation of state
Eq.~\eqref{Eq.NUM} that captures the correlations of $Z_1$ and $Z_2$ around the
stretched state also applies to experimental data from two sources, including
the current experiments and previously available data~\cite{soh2018knots}.
These results demonstrate the promise of the current approach in enabling
physical understanding and control of real polymer dynamics.

\paragraph*{Modelling spatial epidemics}

To further demonstrate general applicability,
we employ our method to derive macroscopic dynamics of spatial epidemics.
The classical spatial SIR model~\cite{murray2001mathematical}
governs microscopic evolution of infective and susceptible individuals
on a spatial domain (See Extended Data Figure 1(A-C)).
Using these microscopic trajectories,
we construct a thermodynamic model that accurately models
the evolution of the spatial averages of infective and susceptible individuals
(Extended Data Figure 1(D,E)) with an
additional learned closure coordinate.
Moreover, following the same approach before,
we can interpret this coordinate as the spatial overlap of
infective and susceptible individuals (Extended Data Figure 2),
thereby rationalizing the dynamical landscape (Extended Data Figure 3),
where this overlap determines the onset and outcome of disease spread.
Details are found in Methods under spatial epidemics analysis.

\section*{Discussion}

The potential applicability of our method goes beyond polymer and epidemic dynamics,
and includes general complex dissipative processes such as
protein folding~\cite{Dobson2003},
self-assembly~\cite{whitesides2002self,capito2008self} and
glassy systems~\cite{cipelletti2005slow,stillinger2013glass}.
Despite the importance
of potential energy landscapes for functional properties material systems, the
challenge in constructing them has limited current approaches to systems with
small degrees of freedom, and/or requiring judicious selection of system
descriptors based on expert knowledge~\cite{krivov2004hidden}.
The method described in this work offers the potential of
automating this process, creating pathways towards a multitude of opportunities for understanding and
control over various complex systems and their scientific applications.

There are many worthwhile future research directions
to further improve the robustness and generality of the proposed method.
Here, we inherently
constrain the search space to macroscopic dynamics
that conform to the generalized stochastic Onsager principle.
Thus, it naturally has limited ability to model systems
that may not readily admit such a description, such as chaotic systems.
Moreover, the present model reduction and stability theory
require the thermal noise to be small
compared with the dissipative and conservative forces.
While this is the case for the polymer dynamics studied here,
our theory needs to be expanded in order to handle
highly stochastic cases.
In terms of training methodology,
the current trial-and-error selection of the dimension
of the closure variables can be made more systematic,
e.g. by building on manifold learning
approaches~\cite{chen2022.AutomatedDiscoveryFundamental}.
Another potential improvement is the data sampling process.
We observed in our numerical experiments that
accurate construction of the dynamical landscape
requires the trajectory data to sufficiently sample
the regions of interest (stable and transition states).
An adaptive or active learning algorithm~\cite{settles2012.ActiveLearningVolume,zhao2022.AdaptiveSamplingMethods}
that couples data sampling and S-OnsagerNet training
can be developed to improve upon the current
random sampling strategy.
On the scientific problem of polymer dynamics,
we have only considered motion under a single stretching force.
It is worthwhile to extend our study to varying stretching
conditions in order to build a more comprehensive
picture of polymer stretching.
More broadly, one may apply our approach
to learn macroscopic thermodynamics of other systems
of scientific interest.

\section*{Methods}

\paragraph{Theoretical results}

In this section, we collect a number of theoretical results concerning the
\textit{stochastic OnsagerNet} (S-OnsagerNet) approach. We first show that if a
high dimensional stochastic dynamical system satisfies the \textit{generalized
stochastic Onsager principle} (GSOP), then, any well-behaved reduction into a
lower dimensional system will result in one that obeys approximately the GSOP
introduced in Eq.~\eqref{eq:sto_onsa} (Theorem~\ref{thm:appen_reduction}). An immediate
consequence is that our model reduction approach is theoretically justified for
a wide variety of dissipative and conservative systems, including molecular
dynamics \cite{nadler1987molecular}, stochastic Hamiltonian systems
\cite{milstein2002symplectic}, and stochastic Lotka-Volterra model
\cite{goel1971volterra}. Next, we prove that dynamics described by the GSOP
satisfy an energy dissipation law (see Theorem~\ref{thm:appen_stability}) and
thus our machine learning approach produces stable dynamics at sufficiently low
temperatures.

For convenience,
we show that there are two equivalent forms of the GSOP.
The formulation of the GSOP we use to construct our neural networks is
\begin{align}\label{eq:sto_onsa1}
  dZ(t) &= -(M(Z(t))+W(Z(t))) \nabla V(Z(t))dt+\sigma(Z(t))dB(t).
\end{align}
Now, assuming $M(\cdot)+W(\cdot)$ is invertible, we define
\begin{align}
 \tilde{M}(Z)&=\frac{(M(Z)+W(Z))^{-1}+((M(Z)+W(Z))^{-1})^T}{2},\nonumber\\
 \tilde{W}(Z)&=\frac{(M(Z)+W(Z))^{-1}-((M(Z)+W(Z))^{-1})^T}{2},\nonumber\\
 \tilde{\sigma}(Z)&=(M(Z)+W(Z))^{-1}\sigma(Z).\nonumber
\end{align}
Observe that
Eq.~\eqref{eq:sto_onsa1} can be rewritten in the form
\begin{align} \label{eq:sto_onsa2}
(\tilde{M}(Z(t))+\tilde{W}(Z(t)))  dZ(t) &= - \nabla V(Z(t))dt+\tilde{\sigma}(Z(t))dB(t).
\end{align}
A similar construction shows that we can also rewrite Eq.~\eqref{eq:sto_onsa2}
in the form of Eq.~\eqref{eq:sto_onsa1} for any $\tilde{M},\tilde{W},\tilde{\sigma}$
assuming similar invertibility conditions.
Thus, they are in fact equivalent.
While the numerical implementation is based on Eq.~\eqref{eq:sto_onsa1},
the form Eq.~\eqref{eq:sto_onsa2} is also useful, and in the following
we refer to both as GSOP.

Now, we demonstrate the general applicability of the GSOP
in the context of model reduction.
We consider a microscopic (high dimensional) dynamics
satisfying a GSOP of the form
\begin{align}\label{eq:high_onsager}
    dX(t) \!&= \!-\! \left(M_1(X(t))+W_1(X(t))\right)\nabla V_1(X(t))dt \!+\!\sqrt{\epsilon_1}\Sigma_1(X(t)) dB_1(t),
\end{align}
where $X(t) \in \mathbb{R}^D$, $M_1(\cdot), W_1(\cdot)$ are symmetric positive
semi-definite and anti-symmetric matrix valued functions respectively,
$\Sigma_1(\cdot)$ is the $D\times p_1$ dimensional diffusion matrix,
$B_1$ is a $p_1$-dimensional Brownian motion
and $\epsilon_1$ is a positive
parameter related to temperature (e.g. $\epsilon_1 \propto k_B T$
with $k_B$ being the Boltzmann's constant and $T$ the temperature).

We now show that many high dimensional systems of physical interest
indeed satisfy a version of GSOP.
We first consider the well-known Langevin dynamics, which has been
used to model many stochastic dynamical systems  e.g. molecular dynamics \cite{tuckerman2010statistical}.

\begin{example}\label{exa:appen_Langevin}
The Langevin equation
\begin{equation}
\begin{split}\label{Langevin}
m \ddot{x}&=-\nabla  U(x)-m\gamma_1  \dot{x}+\sqrt{2m\gamma_1 k_B T} R(t),\\
\end{split}
\end{equation}
can be written in the form of Eq.~\eqref{eq:sto_onsa2}, where the dot denotes a time
derivative, $\dot {x}$ is the velocity, $\ddot{x}$ is the acceleration,
$U(x)$ is the particle interaction potential, and so $-\nabla U(x)$ is the potential force;
$\gamma_1$  is the damping constant (units of reciprocal time),
$R(t)=\dot{B}(t)$ is a delta-correlated stationary Gaussian process with
zero-mean, satisfying $$\left\langle R(t)   \right\rangle =0
,~~~\left\langle R(t)R(t')\right\rangle =\delta (t-t').$$ If we set
$\dot{x}=v$, the Langevin equation can be written as
\begin{align*}
\left(    \begin{array}{ccc}
      m \gamma_1&m\\
       -m &0 \\
    \end{array} \right) \left(    \begin{array}{cc}
       dx\\
       dv\\
    \end{array} \right)  =-\left(    \begin{array}{cc}
       \nabla U(x)\\
       mv \\
    \end{array} \right) dt+\left(    \begin{array}{cc}
       \sqrt{2m\gamma_1 k_B T}\\
      0 \\
    \end{array} \right) dB(t) .\nonumber\
 \end{align*}
Denoting
$X=\left(    \begin{array}{cc}
       x\\
       v\\
    \end{array} \right)$, $\tilde{M}= \left(    \begin{array}{ccc}
       m\gamma_1&0\\
       0 &0 \\
    \end{array} \right)$, $\tilde{W}= \left(    \begin{array}{ccc}
       0&m\\
        -m &0 \\
    \end{array} \right)$, $\Sigma=\left(    \begin{array}{cc}
       \sqrt{2m\gamma_1 k_B T}\\
      0 \\
    \end{array} \right)$ and $V(x,v)=U(x)+\frac{m}{2}v^2$, the Langevin equation can be written in the form of the GSOP as follows:
\begin{align*}
    \left(  \tilde{M}+\tilde{W}\right)     dX    =-    \nabla V   dt+\Sigma dB.
\end{align*}
\end{example}

Another important class of dynamical systems are those described by Poisson
brackets~\cite{beris1994thermodynamics}, whose stochastic extension encompasses
many applications, including the stochastic Lotka-Volterra models and
variants~\cite{goel1971volterra}. In the following, we show that these
dynamical systems can also be written in the form of the GSOP.

\begin{example}\label{exa:appen_poisson}
The stochastic dynamics with generalized coordinates $(q_1, \cdots,q_n,p_1,
\cdots,p_n)$ described by generalized Poisson brackets
\begin{align}\label{Poisson}
    dF =(\{F,H\}-[F,H])dt+\sigma(F) dB,
\end{align}
can be written in the form of Eq.~\eqref{eq:high_onsager},
where $H(q_1, \cdots,q_n;p_1, \cdots,p_n)$ is the Hamiltonian of the
system, $F$ is an arbitrary function depending on the system variables. The
Poisson bracket $\{\cdot,\cdot\}$ and the dissipation bracket
$[\cdot,\cdot]$ are defined as \begin{align} \{F,H\}&=\sum_{i=1}^n\left(
    \frac{\partial F}{\partial q_i}\frac{\partial H}{\partial
    p_i}-\frac{\partial F}{\partial p_i}\frac{\partial H}{\partial q_i}
    \right)\nonumber\\ [F,H]&=J_F M J_H^T,~~~J_F=\left[\frac{\partial
F}{\partial q_1},\cdots, \frac{\partial F}{\partial q_n},\frac{\partial
F}{\partial p_1},\cdots,\frac{\partial F}{\partial p_n}\right] ,\nonumber\
\end{align} where $M$ is symmetric positive semi-definite.

Denote $(h_1,h_2,\cdots,h_{2n})=(q_1,\cdots,q_n,p_1,\cdots,p_n)$.
By the definition of $\{F,H\}$ and taking
$W=
\left(
    \begin{array}{ccc}
        0 &-I_n\\
        I_n &0
    \end{array}
\right)$, we have
\begin{align*}
    \{F,H\}&=\left(\frac{\partial F}{\partial q},\frac{\partial F}{\partial p}\right) \left(\begin{array}{ccc}
          0 &I_n\\
           -I_n &0 \\
        \end{array} \right) \left(    \begin{array}{cc}
          \frac{\partial H}{\partial q}\\
           \frac{\partial H}{\partial p} \\
        \end{array} \right) =\nabla_h F \left(    \begin{array}{ccc}
          0 &I_n\\
           -I_n &0 \\
        \end{array} \right)(\nabla_h H)^T=-J_F W J_H^T.
\end{align*}
Hence, Eq.~\eqref{Poisson} can be written as
\begin{align*}
    dF =(\nabla_h F)^T dh=
    -J_F(W+M)J_H^T dt+\sigma(F)  dB.
\end{align*}
Taking $F=(h_1,\cdots,h_{2n})$ and $\nabla_h F=I_{2n}$, we obtain
\begin{align*}
    dh=
    -J_F(W+M)J_H^T dt+\sigma(F)  dB=
    -(W+M)  \nabla_hH dt+\sigma(h)  dB.
\end{align*}
\end{example}

Next, let us consider the reduction of a microscopic dynamical system
satisfying a GSOP ($X(t)$) into a macroscopic dynamical system ($Z(t)$).
This is achieved by a differentiable reduction function
$\phi : \mathbb{R}^D \rightarrow \mathbb{R}^d$
such that $Z(t) \approx \phi(X(t))$.
Moreover, we consider a differentiable reconstruction function
$\psi : \mathbb{R}^d \rightarrow \mathbb{R}^D$
such that $X(t) \approx \psi(Z(t))$.
Our main result is that $Z(t)$ also satisfies an approximate GSOP.
In other words, the GSOP family is approximately invariant
under dimensionality reduction, or coordinate transformation in general.

In the following, we adopt the notation
\begin{align*}
    \nabla \phi_i(X(t)) &:=
    \left(\frac{\partial \phi_i(X(t))}{\partial x_1},\frac{\partial
    \phi_i(X(t))}{\partial x_2},\cdots,\frac{\partial \phi_i(X(t))}{\partial
    x_D}\right)^T,\\
    \nabla \phi(X(t)) &:=
    \left(\nabla \phi_1(X(t)),\cdots,\nabla \phi_d(X(t))\right).
\end{align*}
We will also adopt the following technical assumptions.
\begin{Assumption}\label{app:assumption1}
The functions $M_1,W_1, \nabla V_1: \mathbb{R}^D  \rightarrow \mathbb{R}^D$,
$\Sigma_1:  \mathbb{R}^D  \rightarrow \mathbb{R}^{D\times p_1}$,
$\phi:\mathbb{R}^D \rightarrow \mathbb{R}^d$,
and $\psi:\mathbb{R}^d \rightarrow \mathbb{R}^D$
satisfy:
\begin{enumerate}
    \item Growth condition: there exists $L>0$ such that, for all $x\in \mathbb{R}^D$ and $z\in \mathbb{R}^d$,
    \begin{align}
    &\lvert (M_1(x)+W_1(x))\nabla V_1(x)\rvert+\lvert \Sigma_1(x)\rvert+\sum_{i=1}^d\lvert \nabla \phi_i(x)\rvert \leq L^2(1+\lvert  x\rvert),\nonumber\\
    &\sum_{i=1}^D\lvert \nabla \psi_i(z)\rvert +\lvert \psi(z)\rvert\leq L^2(1+\lvert  z\rvert),\nonumber\
    \end{align}
    \item Lipschitz condition: there exists $L>0$ such that, for all $x\in \mathbb{R}^D$, and $z\in \mathbb{R}^d$, the function $(M_1(x)+W_1(x))\nabla V_1(x) $, $\Sigma_1(x)$, $\{\nabla  \phi_i(x)\}_{i=1}^{d}$, $\{\nabla  \psi_i(z)\}_{i=1}^{D} $ and  $\psi(z)$ satisfy the Lipchitz condition with constant $L$.
    \item Approximate reconstruction: there exists $\epsilon_0>0$ such that,
    $\sup_{x\in\Omega} \lvert  x-\psi(\phi(x))\rvert<\epsilon_0$ where $\Omega \subset \mathbb{R}^D$ is a domain such that $X(t) \in \Omega$ for all $t\in[0,T]$ almost surely.
\end{enumerate}
Here, $\lvert \cdot\rvert$ is the Euclidean norm for a vector and Frobenius norm for a matrix.
\end{Assumption}

\begin{theorem}\label{thm:appen_reduction}
    Let $X(t)$ satisfy Eq.~\eqref{eq:high_onsager} and $Z(t)$ satisfy Eq.~\eqref{eq:sto_onsa1} with
    \begin{align}
        M(Z) &=  \nabla\phi[\psi(Z)]^T
          M_1(\psi(Z))
        \nabla \phi[\psi(Z)],
        \nonumber\\
        W(Z) &= \nabla\phi[\psi(Z)]^T
           W_1(\psi(Z))
        \nabla \phi[\psi(Z)],
        \nonumber\\
        V(Z)&=V_1[\psi(Z)],\nonumber\\
        \sigma_1(Z) &=
        \sqrt{\epsilon_1} \nabla\phi[\psi(Z)]^T
        \Sigma_1[\psi(Z)] ,\nonumber\\
        \sigma(Z) &= [\sigma_1(Z) \sigma_1(Z)^T]^{\frac{1}{2}} ,\nonumber
    \end{align}
    Then, for each $u\in \mathcal{C}^{2}(\mathbb{R}^d)$ (twice continuously differentiable function)
    there exists a constant $C>0$, independent of $\epsilon_0$ and $\epsilon_1$, such that
    \begin{align}
    \lvert \mathbb{E}   u(\phi[X(t)])- \mathbb{E}u(Z(t))\rvert\leq C(\epsilon_0+\epsilon_1 ). \nonumber\
    \end{align}
\end{theorem}
\begin{proof}
    Suppose $Y(t)$ satisfy the following equation:
    \begin{align}\label{eq:sto_y}
      dY(t) &= -(M(Y(t)+W(Y(t))) \nabla V(Y(t))dt+\sigma_1(Y(t))dB_1(t),
    \end{align}
    we will prove
        \begin{align}
        \mathbb{E} \lvert  u(\phi[X(t)])- u(Y(t))\rvert\leq C(\epsilon_0+\epsilon_1 ). \nonumber\
        \end{align}
    By It\^{o}'s formula, Eq.~\eqref{eq:high_onsager}
    and Assump.~\ref{app:assumption1}, we obtain
    \begin{align}\label{app:eq_X}
        &d\phi_i(X(t))\nonumber\\
        =&
        [\nabla\phi_i(X(t)) ]^TdX(t)
        +
        \frac{1}{2}
        [dX(t)]^T \nabla^2\phi_i(X(t))  dX(t)
        \nonumber\\
        =&
        [\nabla \phi_i(X(t))]^T
        \left[
            - \left(
                M_1(X(t))+W_1(X(t))
            \right)
            \nabla V_1(X(t))dt
            + \sqrt{\epsilon_1}\Sigma_1(X(t)) dB_1(t)
        \right]
        \nonumber\\
        &+
        \frac{\epsilon_1}{2}
        \mathrm{Tr}
        [
            \Sigma_1(X(t))
            \Sigma_1^T(X(t))
            \nabla^2 \phi_i(X(t))
        ] dt
        ,\
    \end{align}
    where $\nabla^2 \phi_i$ is the Hessian matrix of $\phi_i$ and $\mathrm{Tr}$ is the trace of a square matrix.

    Now, by definition, $Y(t)$ satisfies the following SDE
    \begin{align}\label{app:eq_Z}
        &dY(t)\nonumber\\
        =&-\left(M(Y(t))+W(Y(t))\right)\nabla V(Y(t))dt+\sigma_1(Y(t))dB_1(t)\nonumber\\
        =&-\nabla \phi[\psi(Y(t))]^T[M_1(\psi(Y(t)))+W_1(\psi(Y(t)))]
        \nabla\phi[\psi(Y(t))] \nabla V(Y(t))dt \nonumber\\
        &+\sqrt{\epsilon_1} \nabla \phi[\psi(Y(t))]^T \Sigma_1(\psi(Y(t))) dB_1(t)\nonumber\\
        =&-\nabla \phi[\psi(Y(t))]^T[M_1(\psi(Y(t)))+W_1(\psi(Y(t)))  ]    \nabla V_1(\psi(Y(t)))dt \nonumber\\
        &+\sqrt{\epsilon_1}  \nabla \phi[\psi(Y(t))]^T\Sigma_1(\psi( Y(t))) dB_1(t).\
    \end{align}
    Subtracting equations~\eqref{app:eq_X} and~\eqref{app:eq_Z},
    and integrating on $[0,t]$, we get
    \begin{align}
       & \phi_i(X(t))- Y_i(t)\nonumber\\
        =&
        \int_0^t
        \left[
            \nabla  \phi_i(X(r))^T f_1(X(r))-\nabla  \phi_i[\psi(Y(r))]^T f_1(\psi(Y(r)))
        \right] dr +
        \epsilon_1\int_0^t g_{i}(X(r)) dr    \nonumber\\
        &+
        \sqrt{\epsilon_1}\int_0^t
        [ \nabla  \phi_i(X(r))^T \Sigma_1(X(r))    -\nabla  \phi_i[\psi(Y(r))]^T\Sigma_1(\psi(Y(r)))] dB_1(r),\nonumber
    \end{align}
    where $f_1(\cdot)=-\left(M_1(\cdot)+W_1(\cdot)\right)\nabla V_1(\cdot)$ and $g_{i}(\cdot)=\frac{1}{2}\mathrm{Tr}[\Sigma_1\Sigma_1^T \nabla^2\phi_i](\cdot)$ with $i=1,2,\cdots,d$.

    By Cauchy-Schwarz inequality, and It\^{o} isometry, we  have
    \begin{align}\label{estimate_square}
        &\mathbb{E}\lvert \phi_i(X(t))- Y_i(t)\rvert^2\nonumber\\
        \leq &3\mathbb{E}\left\lvert  \int_0^t    [\nabla  \phi_i(X(r))^T f_1(X(r))-\nabla  \phi_i[\psi(Y(r))]^T f_1(\psi(Y(r)))]dr\right\rvert^2+3\mathbb{E}\left\lvert \int_0^t  \epsilon_1 g_{i}(X(r)) dr  \right\rvert^2\nonumber\\
        &+3\epsilon_1 \mathbb{E}\left\lvert \int_0^t [  \nabla  \phi_i(X(r))\Sigma_1(X(r))\!\! -\!\
        \nabla  \phi_i[\psi(Y(r))]^T\Sigma_1(\psi(Y(r)))] dB_1(r)\right\rvert^2\nonumber\\
        \leq &3t\mathbb{E} \int_0^t    \left\lvert \nabla  \phi_i(X(r))^T f_1(X(r))\!\!-\!\!\nabla  \phi_i[\psi(Y(r))]^T f_1(\psi(Y(r)))\right\rvert^2dr\!\!+\!3\epsilon_1^2t\mathbb{E}\int_0^t  \lvert  g_{i}(X(r))\rvert^2 dr \nonumber\\
        &+3\epsilon_1 \mathbb{E}\int_0^t \left\lvert   \nabla  \phi_i(X(r))\Sigma_1(X(r))\! -\!
        \nabla  \phi_i[\psi(Y(r))]^T\Sigma_1(\psi(Y(r))) \right\rvert^2dr.\
    \end{align}
    By Assump.~\ref{app:assumption1}, there exists positive constants $C_1$ and $C_2$ such that
    \begin{align*}
        &\mathbb{E}  \int_0^t  \lvert
         \nabla\phi_i(X(r))^T f_1(X(r))
        -\nabla  \phi_i[\psi(Y(r))]^T   f_1(\psi(Y(r)))\rvert^2dr \nonumber\\
        \leq&\mathbb{E} \int_0^t \lvert \nabla  \phi_i(X(r))^Tf_1(X(r))
        \!-\!\nabla  \phi_i[\psi(\phi(X(r)))]^T f_1(X(r))
        \!+\!\nabla  \phi_i[\psi(\phi(X(r)))]^T f_1(X(r)) \nonumber\\
        & \!-\!\nabla  \phi_i[\psi(Y(r))]^T f_1(X(r))
        \!+\!\nabla  \phi_i[\psi(Y(r))]^T f_1(X(r))
        \!-\!\nabla  \phi_i[\psi(Y(r))]^T f_1(  \psi(\phi(X(r)))    ) \nonumber\\
        &\!+\!\nabla  \phi_i[\psi(Y(r))]^T f_1(  \psi(\phi(X(r)))  )
        \!-\!\nabla  \phi_i[\psi(Y(r))]^T f_1(  \psi(Y(r))   )
        \rvert^2dr\nonumber\\
        &\leq C_1 \mathbb{E}  \int_0^t\lvert  X(r)-\psi(\phi(X(r))) \rvert^2dr+ C_2 \mathbb{E}  \int_0^t\lvert \phi(X(r))-Y(r)\rvert^2dr\nonumber\\
        &\leq C_1 t \epsilon_0^2+ C_2 \mathbb{E}  \int_0^t\lvert \phi(X(r))-Y(r)\rvert^2dr.\
    \end{align*}

    We employ a similar above argument of the third term of Eq.~\eqref{estimate_square} and get
    \begin{align*}
        \mathbb{E}\lvert \phi(X(t))- Y(t)\rvert^2
        \leq & C_3  \int_0^t\mathbb{E}\lvert \phi(X(r))-Y(r) \rvert^2dr+ C_4(\epsilon_0 +\epsilon_1 )^2 .\
    \end{align*}
     Here, we have used the Lipschitz conditions and the boundness
    of the first moment of $f_1(X)$ and $g_i(X)$, which is implied by the growth condition in Assump.~\ref{app:assumption1}.
    This shows that $\phi(X(t))$ and $Y(t)$ are close in the
    mean-square sense.
    By Gronwall's inequality, we get
    \begin{align}\label{eq:low_esti}
        \mathbb{E}\lvert \phi(X(t))- Y(t)\rvert^2
        \leq & C_5(\epsilon_0 +\epsilon_1 )^2 .\
    \end{align}

    Now, we employ a similar argument to show that
    $u(\phi(X(t)))$ and $u(Y(t))$ are close for any
    sufficiently smooth $u$.
    We apply It\^{o} formula to $u(\phi(X(t)))$ and
    $u(Y(t))$ to get
    \begin{align}
        du(\phi(X(t)))=&\nabla u\cdot d\phi(X(t))+\frac{1}{2}d \phi(X(t))^T \nabla^2 u(\phi(X(t)) d\phi(X(t))\nonumber\\
        =&\nabla u \cdot
        [(\nabla \phi(X(t))^Tf_1(X(t))+\epsilon_1 g(X(t)) ) dt+  \nabla  \phi(X(t))\Sigma_1(X(t))dB_1(t) ]\nonumber\\
        &+\frac{1}{2}\epsilon_1 \mathrm{Tr}[\nabla  \phi(X(t))\Sigma_1(X(t))(\nabla  \phi(X(t))\Sigma_1(X(t))^T \nabla^2 u ]  ) dt,\nonumber\\
        du( Y(t))
        =&\nabla u\cdot dY(t)+\frac{1}{2}d Y(t)^T \nabla^2 u(Y(t)) dY(t)\nonumber\\
        =&\nabla u \cdot
        \nabla \phi[\psi(Y(t))]^T[f_1(\psi(Y(t)))  dt+  \sqrt{\epsilon_1} \Sigma_1(\psi(Y(t)))dB_1(t) ]\nonumber\\
        &+\frac{\epsilon_1}{2} \mathrm{Tr}[\nabla  \phi[\psi(Y(t))]\Sigma_1(\psi(Y(t)))(\nabla  \phi[\psi(Y(t))]\Sigma_1(\psi(Y(t))))^T \nabla^2 u ]  dt.\nonumber\
    \end{align}
    As before, by Assump.~\ref{app:assumption1}, there exist positive constants $C_6$ and $C_7$ such that
    \begin{align}\label{eq:estimate_u}
        \mathbb{E}\lvert  u(\phi(X(t)))- u(Y(t))\rvert^2\leq &C_6\mathbb{E}\int_{0}^t\lvert \phi(X(r)))- Y(r)\rvert^2dr+C_7(\epsilon_0+ \epsilon_1 )^2.
    \end{align}
    Combining equations~\eqref{eq:low_esti} and~\eqref{eq:estimate_u}, we obtain
    \begin{align}
        \mathbb{E} \lvert u(\phi(X(t)))- u(Y(t))\rvert^2\leq C^2(\epsilon_0+\epsilon_1)^2. \nonumber\
    \end{align}
    By Jensen's inequality, we can get
    \begin{align}
        \mathbb{E} \lvert u(\phi(X(t)))- u(Y(t))\rvert\leq C(\epsilon_0+\epsilon_1). \nonumber\
    \end{align}
    According to the Eq. \eqref{eq:sto_onsa1} and \eqref{eq:sto_y}, we can get $Z(t)$ and $Y(t)$ has the some distribution, i.e.
    \begin{align}
         \mathbb{E} u(Y(t))-\mathbb{E} u(Z(t))=0. \nonumber\
    \end{align}
    Finally, by triangle inequality, we have
    \begin{align}
         &\lvert \mathbb{E} u(\phi(X(t)))- \mathbb{E} u(Y(t))\rvert\nonumber\\
         =&
         \lvert \mathbb{E} u(\phi(X(t)))-\mathbb{E} u(Y(t))+\mathbb{E} u(Y(t))- \mathbb{E} u(Z(t))\rvert \nonumber\\
         \leq&  \mathbb{E}\lvert u(\phi(X(t)))-  u(Y(t))\rvert
         +\lvert\mathbb{E} u(Y(t)) -\mathbb{E} u(Z(t))\rvert \nonumber\\
         \leq& C(\epsilon_0+\epsilon_1).\nonumber\
    \end{align}
    This completes the proof.
\end{proof}

This results demonstrate the validity of the GSOP as a dimensionality reduction
method. In short, it says that if the microscopic dynamics satisfies a GSOP,
then the macroscopic dynamics will also satisfy a GSOP approximately. Since a
large amount of conservative and dissipative microscopic physical systems are
shown to satisfy the GSOP, the S-OnsagerNet approach based on the GSOP is a
principled model reduction ansatz for physical processes.

Next, we show the stability of a solution of the GSOP.
More precisely, we prove in Theorem~\ref{thm:appen_stability} below that the mean of the potential
is non-increasing in $t$ for sufficiently low temperatures
(small $\lvert  \sigma \rvert$).
Consequently, the S-OnsagerNet produces dissipative dynamical systems that enjoy long term stability.
\begin{theorem}\label{thm:appen_stability}
The solution of Eq.~\eqref{eq:sto_onsa1} satisfies the dissipation law
\begin{align}
    \mathbb{E}V(Z(t))-\mathbb{E}V(Z(0))
    =&
    - \int_0^t \mathbb{E}\|\nabla V(Z(r))\|_{M}^2dr \nonumber\\
    &+\frac{1}{2}\int_{0}^t  \mathbb{E} \mathrm{Tr}[\sigma(Z(r))\sigma(Z(r))^T   \nabla^2 V(Z(r))] dr.\nonumber\
\end{align}
Here, $\|\cdot\|_M^2$ denotes $\lvert  M^{1/2} \cdot\rvert^2$ where
$M^{1/2}$ is the non-negative square-root of $M$.
If we assume further that there exists a positive constant $\alpha$ such that $Z \cdot M(Z) Z\geq \alpha \lvert  Z\rvert^2$ and $\frac{1}{2}  \mathrm{Tr}[\sigma(Z)\sigma(Z)^T \nabla^2V(Z)]   \leq \alpha  \lvert \nabla V(Z)\rvert^2 $ for all $Z$,
then $\mathbb{E}[V(Z(t))]$ is non-increasing in $t$.
\end{theorem}
\begin{proof}
By It\^{o} formula, we obtain
\begin{align}
dV(Z)=&\nabla V\cdot dZ+\frac{1}{2}(dZ)^T\nabla^2 V dZ\nonumber\\
=&-\nabla V\cdot [(M+W)\nabla V]dt+\nabla V\cdot \sigma dB_1+\frac{1}{2}\mathrm{Tr}[\sigma\sigma^T \nabla^2 V  ]dt\nonumber\\
=& -  \|\nabla V\|_{M}^2 dt+\frac{1}{2}\mathrm{Tr}[\sigma\sigma^T \nabla^2 V] dt+\nabla V\cdot \sigma dB_1,\label{potential_equ}
 \end{align}
where $ \nabla V\cdot M\nabla V=\|\nabla V\|_{M}^2$ and $ \nabla V\cdot W\nabla V=0$ are used.

Integrating Eq.~\eqref{potential_equ} from $0$ to $t$ and taking expectation, we obtain
\begin{align}
    \mathbb{E}V(Z(t))-\mathbb{E}V(Z(0))&= -  \int_0^t \mathbb{E}\|\nabla V\|_{M}^2dr+\frac{1}{2}\int_{0}^t  \mathbb{E} \mathrm{Tr}[\sigma\sigma^T \nabla^2 V] dr.\nonumber\
\end{align}
Finally, according to the proposed condition, it is easy to arrive at
\begin{align}
    \mathbb{E}V(Z(t))-\mathbb{E}V(Z(0))& \leq - \alpha \int_0^t \mathbb{E}\lvert \nabla V\rvert^2dr+\frac{1}{2}\int_{0}^t  \mathbb{E}\mathrm{Tr}[\sigma\sigma^T \nabla^2 V] dr \leq0, \nonumber\
\end{align}
thus the mean of the energy potential is non-increasing in $t$. This completes the proof.
\end{proof}

Theorem~\ref{thm:appen_stability} shows that energy is being dissipated
if the temperature of the ambient reservoir is sufficiently
low.
Accordingly, if the free energy $V$ has compact level sets,
then the dynamics at low temperatures will be confined
on average to these compact sets and is thus stable.
This contrasts with unstructured methods
that may learn dynamics that
are accurate for short times but induce instability
at long times.
We numerically verify this energy dissipation law
in the learned polymer dynamics in Supplementary Figure~1.

\paragraph{Model implementation}

We provide in this section the detailed implementation of
S-OnsagerNet.

\textbf{Model architecture.}
We begin with discussing the architecture design of
the neural network approximators.
Following the acquisition of the closure coordinates  $Z(t)=(Z^*(t),\hat{Z}(t))$,
the S-OnsagerNet architecture implements equation~\eqref{eq:sto_onsa1}.
To ensure the symmetric positive definiteness of $M(Z)$ and the anti-symmetry of $W(Z)$, we use a neural network to approximate $A(\cdot) : \mathbb{R}^d \rightarrow \mathbb{R}^{d^2 }$ with dimension $d^2$. Then, we take the lower-triangular part as $L_1(Z)$ and the upper-triangular part as $L_2(Z)$. $M(Z)$ and $W(Z)$ are represented by
\begin{equation*}
    \begin{aligned}
        M(Z)&=L_1(Z)L_1(Z)^T+\alpha I,\\
        W(Z)&=L_2(Z)-L_2(Z)^T,
    \end{aligned}
\end{equation*}
where $\alpha$ is a positive constant and $I$ is an identity matrix.

The energy function $V(\cdot)$ is lower bounded, so we use the following structure
\begin{align*}
    V(Z)=\frac{1}{2}\sum_{i=1}^{m}\left(U_i(Z)+\sum_{j=1}^{d}\gamma_{ij}Z_{j} \right)^2+\beta \lvert  Z\rvert^2,
\end{align*}
where $U(Z)$ is a neural network with $d$-dimensional input and $m$-dimensional output,
$\{\gamma_{ij}\}$ is a trainable matrix, and $\beta$ is a positive parameter.

In the architecture used for the polymer dynamics application in this paper, we
set $\alpha=0.1$ and utilize a neural network with 2 hidden layers with 20
neurons each and the $\operatorname{tanh}$ activation function to approximate
$M(z)$ and $W(z)$. To parameterize the potential $V(z)$, we decompose it into a
sum of squares of the output layer (size $m=50$) of 1 hidden layer neural
network with 128 hidden neurons and the $\operatorname{ReQUr}$ activation
function \cite{li2019better,yu2020onsagernet}. This is to ensure that the
potential satisfies the correct growth conditions as outlined in
Assump.~\ref{app:assumption1}.

For the diffusion matrix $\sigma(z)$,
since it has no symmetry constraints other than a growth condition,
we use a fully connected neural network to approximate it.
In our polymer dynamics application, we found empirically that a diagonal, $z$-independent diffusion matrix
(corresponding to a linear network with zero weight and trainable diagonal bias) performed the best,
but our algorithm can handle general architectures for $\sigma(z)$.

\textbf{Closure coordinate construction.}
We now provide details of the procedure to construct closure coordinates
$\hat{Z}(t)$ using the time series observation data of chain configuration
coordinates at $\{t_k\}_{k=1}^{N_t}$, with $0=t_0 < t_1 ... <t_{N_t} = T$. The
available data are $\{\mathcal{X}_j\}_{j=1}^M$ with $ \mathcal{X}_j =
\{X(t_i)^{(j)}\}_{i=1}^{N_t} \in \mathbb{M}^{D \times N_t}$, where $M$ is the
number of trajectories and $X(t_i)^{(j)}$ is the $j^\text{th}$ observation
trajectory at $t=t_i$. We obtained 610 observational trajectories, and for each
trajectory, the number of time snapshots is 1001, i.e. $M=610$ and $N_t=1001$.
We reshape the observation data as
$X=[\mathcal{X}_1,\mathcal{X}_2,...,\mathcal{X}_M]$, where $X \in \mathbb{M}^{
    D \times N_t M  }$. We re-center the data, such that the mean of the
training data is 0 for each time snapshot, and set it as $X$.
%\QL{This is confusing - just introduce $r$, center it, and set X?}
The covariance matrix of $X$ is $S=Cov(X)$. Denote its eigenvalues by
$\Lambda=diag(\lambda_1,\lambda_2,...,\lambda_D)$ (arranged in non-increasing
order) and corresponding eigenvectors as $V=(V_1,V_2,...,V_D)$.

We use the following PCA-ResNet encoder to find the closure coordinates
\begin{align}\label{eq:PCA_resNet_encoder}
    \hat{Z}(t)= Z_{2:d}(t)=\hat{\phi}(X(t))=P_dX(t)+\operatorname{NN}_{e}(X(t)),
\end{align}
where $P_d=\frac{1}{\sqrt{\Lambda_{1:d-1}}}V_{1:d-1}$ and $\operatorname{NN}_{e}(\cdot)$ is a fully-connected
neural network with input dimension $D$ and output dimension $d-1$.
We can reconstruct the high-dimensional coordinates via the decoder
\begin{align}\label{eq:PCA_resNet_decoder}
    \tilde{X}(t)= \psi(\hat{Z}(t))= P_d^{\dagger}\hat{Z}(t)+\operatorname{NN}_{d}(\hat{Z}(t)),
\end{align}
where $P_d^{\dagger}=V_{1:d-1}^T\sqrt{\Lambda_{1:d-1}}$ and
$\operatorname{NN}_{d}(\hat{Z}(t))$ is another neural network with input
dimension $d-1$ and output dimension $D$. Note that without
$\operatorname{NN}_e,\operatorname{NN}_d$, this amounts to a principal
component analysis (PCA) based coordinate reduction. The combination of PCA and
neural networks combines approximate feature orthogonality and approximation
flexibility.

We construct the reconstruction loss function as $\lvert  X-\tilde{X}\rvert^2$.
We set the reconstruction error of PCA as $E_{pca}=\lvert  X-P_d^{\dagger} P_d
X\rvert^2$. To make the reconstruction error near but less than the
reconstruction error of PCA alone, we add the regularization term
$\mathrm{\mathrm{ReLU}}(\log\lvert  X-\tilde{X}\rvert^2 - \log(E_{pca}))$ in
the loss function, where $E_{pca} = \lvert  X-P_d^{\dagger} P_d X\rvert^2$ is
the reconstruction error of PCA and $\mathrm{ReLU}$ is the rectified linear
unit, i.e. $\mathrm{ReLU}(u) = \max(0, u)$. Thus, the combined reconstruction
loss function is
\begin{align*}
\mathrm{loss}_\mathrm{Rec}= \lvert  X-\tilde{X}\rvert^2+\rho_1  \mathrm{ReLU}(\log\lvert  X-\tilde{X}\rvert^2 - \log\lvert  X-P_d^{\dagger} P_d X\rvert^2),\nonumber\
\end{align*}
where %$\|\cdot\|$ is the Frobenius norm,
$\rho_1$ is a regularization parameter.
The regularization term penalizes
the loss if we observe a reconstruction error
that is higher than PCA.

%  \begin{algorithm}[htb]
% \caption{ Framework of S-OnsagerNet to learn SDE.}
% \label{alg:Framwork}
% \begin{algorithmic}[1]
% % \textbf{Input:} $Z\in \mathbb{R}^d$, parameter $\alpha>0$, $\beta>0$, activation function for $\h_A$ and $\h_U$, number of hidden layer $l_A$ and $l_U$ and number of each nodes in each hidden layer $n_A$ and $n_U$
% % \textbf{Output:}: OnsagerNet(\alpha,\beta,l_A,L_U,n_A,n_U;Z) \in \mathbb{R}^d
%     \STATE Construct the neural network for $U_i$ and then calculate $V$ according to \eqref{}
%     \STATE Construct the neural network to $A\in \mathbb{R}^{d^2}$ and reshape $A$ as a $d\times d$ matrix, take the lower-triangular part as $L$, the upper-triangular part to form the $W$;
%     \STATE Calculate the output of the drift term of $f$;
%     \STATE Compare with $\sigma$ and use maximal likelihood to get the loss function.
% %\RETURN $E_n$;
% \end{algorithmic}
% \end{algorithm}

\textbf{Training algorithm based on maximum likelihood estimation.}
After constructing the structure of the drift term $f(Z(t))=-(M(Z(t)+W(Z(t)))
\nabla V(Z(t))$ and diffusion term $\sigma(Z(t))$, we consider how to construct
the loss function to learn the stochastic dynamics. In deterministic dynamical systems, we
can use the mean square error to learn $f$ given the trajectory observation data.
However, to deal with stochastic dynamics (in particular, learning the
diffusion matrix $\sigma$), we have to devise more general methods based on
maximum likelihood estimation.

We discretize Eq.~\eqref{eq:sto_onsa1} by the Euler-Maruyama scheme, giving
\begin{align*}
    Z(t_{i+1})&=Z(t_i) +f(Z(t_i))\Delta t + \sigma(Z(t_i)) \sqrt{\Delta t}\xi_i,\nonumber
\end{align*}
where $\Delta t$ is the time step, $t_i=i\Delta t$, $i=0,1,...N_t-1$ and $T=N_t\Delta t$.
Here, $\{\xi_i\}_{i=0}^{N_t-1} $ are independent random vectors following the standard normal distribution.

In the training data set, we have access to the microscopic coordinates $X(t)$,
from which we construct the reduced coordinates $\{(Z(t_i)^{(j)},Z(t_{i+1}^{(j)}),\Delta t)\}_{i,j=0,1}^{N_t-1,M}$
via the (to be trained) reduction function $\phi$.
Given the neural networks $f_\theta$ and $\sigma_{\theta}$
(we use the subscript $\theta$ to denote all trainable parameters)
constructed previously, the conditional probability is given by
\begin{align}
    &p(  Z(t_{i+1})^{(j)} \vert   Z(t_i)^{(j)} )\nonumber\\
    =&\mathcal{N}(Z(t_{i+1} )^{(j)}\rvert Z(t_i)^{(j)}  +f_{\theta}(Z(t_i)^{(j)}) \Delta t, \sigma_{\theta}(Z(t_i)^{(j)}) \sigma_{\theta}^T(Z(t_i)^{(j)})\Delta t)\nonumber\\
    =&\frac{1}{ (2\pi\Delta t)^{d/2}\sqrt{\det(\sigma_{\theta} \sigma_{\theta}^T)}}\exp\{  -\frac{1}{2\Delta t}[ Z(t_{i+1})^{(j)} -Z(t_i)^{(j)} -f_{\theta}(Z(t_i)^{(j)} )\Delta t   ]^T\nonumber\\
    &(\sigma_{\theta} \sigma_{\theta}^T)^{-1} [ Z(t_{i+1})^{(j)} -Z(t_i)^{(j)}  -f_{\theta}(Z(t_i)^{(j)})\Delta t   ] \},\nonumber\
\end{align}
where we use the short form $\sigma_{\theta}=\sigma_{\theta}(Z(t_i)^{(j)})$
and $\det$ denotes the determinant of a matrix.

Taking the logarithm of the above equation, we obtain
\begin{align}%\label{MLE}
    &\log p(  Z(t_{i+1})^{(j)}  \vert   Z(t_i)^{(j)} )\nonumber\\
    =&-\frac{1}{2}\log \det(\sigma_{\theta} \sigma_{\theta}^T) -\frac{\Delta t}{2}\left( \frac{Z(t_{i+1})^{(j)} -Z(t_i)^{(j)}}{\Delta t} -f_{\theta}(Z(t_i)^{(j)} )   \right)^T\nonumber\\
    & (\sigma_{\theta} \sigma_{\theta}^T)^{-1}
    \left( \frac{ Z(t_{i+1} )^{(j)}-Z(t_i)^{(j)}}{\Delta t}  -f_{\theta}(Z(t_i)^{(j)})   \right) +
    \text{constant}.\nonumber
\end{align}
As a result, we may obtain the loss function by maximizing the log-likelihood of the training data
\begin{align*}
    &\mathrm{loss}_{\mathrm{MLE}}\nonumber\\
    =&\frac{1}{N_t M}\sum_{i=1}^{N_t}\sum_{j=1}^{M}\bigg(\frac{1}{2}\log \det(\sigma_{\theta} \sigma_{\theta}^T)+ \frac{\Delta t}{2}\left( \frac{Z(t_{i+1})^{(j)}-Z(t_i)^{(j)}}{\Delta t} -f_{\theta}(Z(t_i)^{(j)} )    \right)^T \nonumber\\
    &(\sigma_{\theta} \sigma_{\theta}^T)^{-1}
    \left( \frac{ Z(t_{i+1})^{(j)} -Z(t_i)^{(j)}}{\Delta t}  -f_{\theta}(Z(t_i)^{(j)})    \right) \bigg).
\end{align*}
The total loss is then
\begin{align}\label{loss_total}
    \mathrm{loss} =\mathrm{loss}_{\mathrm{MLE}}+\rho~\mathrm{loss}_{\mathrm{Rec}},
\end{align}
where $\rho$ is a parameter to balance the accuracy of the
learned dynamics and the error from reconstruction. In our computation, we
first train the loss function~\eqref{loss_total} with $\rho=0.01$. After some
training steps, we fix the encoder part \eqref{eq:PCA_resNet_encoder} and the
decoder part \eqref{eq:PCA_resNet_decoder} of the neural network, and train
$\mathrm{loss}_{\mathrm{MLE}}$ ($\rho=0$) to fine-tune accuracy of the stochastic dynamics.
We use the Adam optimizer for training.
The overall implementation, including the network architectures and loss computation,
is shown in Supplementary Figure~2.

\paragraph{Data preparation}

\textbf{Simulation data.}
We used a Brownian dynamics approach to simulate linear, touching-bead chains as polymer chains in a planar elongational flow (left of Fig.~\ref{fig:unfolding_times}). Each polymer chain consisted of $N = 300$ ($D=3N$) beads with diameter $r$ at positions of bead $i$ $\mathbf{r}_i$, connected by $N-1$ rigid rods of length $b = r = 10$ nm. The governing stochastic differential equation was obtained by considering the following forces acting upon the system: excluded volume, constraint, Brownian and hydrodynamic.

The excluded volume potential characterizes the short-range repulsions between beads and can be described by
$$E^{ev} = -\sum_{i,j}^N \mu r_{ij}\ \ \ \ \   \text{if}\ \  r_{ij} < r \ \text{nm}$$
where $\mu = 35$ pN has been demonstrated to result in a low frequency of chain crossings \cite{Vologodskii2006}. The constraint force is given by
$$\mathbf{F}^c_i = T_i\mathbf{b}_i-T_{i-1}\mathbf{b}_{i-1}$$
where $\mathbf{b}_i$ is the unit vector of bond $i$ and $T_i$ is the tension in rod $i$ that imposes constant bond length. The Brownian forces are random forces that satisfy the fluctuation-dissipation theorem, represented as
$$\langle \mathbf{F}_i^{br}(t) \rangle = 0 \ \ \  \text{and}\ \ \ \langle \mathbf{F}_i^{br}(t) \mathbf{F}_j^{br}(t) \rangle = \frac{2 k_BT \zeta \mathbf{I} \delta_{ij}}{\Delta t}$$ where $\delta_{ij}$ is the Kronecker delta, $\mathbf{I}$ is the identity matrix and $\Delta t$ is the simulation time step. By neglecting hydrodynamic interactions between the beads, the drag force on the $i$th bead is
$$\mathbf{F}_i^d = - \zeta\Big(\mathbf{u(r_i)}-\frac{\text{d}\mathbf{r_i}}{\text{d}t}\Big)$$
where $\zeta \approx 3 \pi \eta r$ is the drag coefficient of a single bead, $\eta$ is the solvent viscosity and $\mathbf{u(r_i)}$ is the unperturbed solvent velocity.

Due to the small mass of the beads, it is common to neglect inertial effects and set the sum of forces on the beads to be zero. Hence, the Langevin equation that describes the motion of each bead along the chain is
$$\frac{\text{d}\mathbf{r_i}}{\text{d}t} = \mathbf{u(r_i)}+\frac{1}{\zeta}\Big(\mathbf{F}^{ev}+\mathbf{F}^c+\mathbf{F}^{br}\Big)$$
We employed a predictor-corrector scheme~\cite{Liu1989} to determine the position of each bead at every time step. The enforcement of rigid rod constraints leads to a system of nonlinear equations for  the rod tensions $T_i$, which we solved for using Newton's method.

For each simulation run, the polymer chain was allowed to equilibrate for $10^4 \  \tau_d$, with $\tau_d = b^2\zeta/k_BT$ being the characteristic rod diffusion time. During equilibration, the chain would adopt random configurations as governed by thermal fluctuations. At $t = 0$, the chain was subjected to a planar elongational flow of the form $\mathbf{u}(\mathbf{r_i}) = \dot{\epsilon}(\mathbf{ \hat{x}-\hat{y}}) \cdot \mathbf{r_i}$, where $\dot{\epsilon}$ is the strain rate and $\mathbf{\hat{x}}$ and $\mathbf{\hat{y}}$ are unit vectors parallel to the $x$ and $y$ axes respectively. The simulations were run until $t = 10^4\ \tau_d$, using a time step of $\Delta t = 5 \times 10^{-4}\ \tau_d$. To generate training data, we simulated 610 stretching trajectories.
To test the predictions, we simulated 500 stretching trajectories each for three different initial chain configurations, which were deliberately selected from three vastly different trajectories. For each trajectory, we obtain the $(x,y,z)$ positional coordinates $\mathbf{r_i}$ of $N$ beads every 10 $\ \tau_d$.
Given the observation data, we can get the chain extension (Fig.~\ref{fig:unfolding_times}).
We note that time reported hereon is in units of $\tau_d$. Although the data set in this work is generated based on known equations, it should be highlighted that our machine learning approach for constructing the reduced dynamical model and all resulting consequences are independent of the microscopic model used in the simulations, with only the positional coordinates as inputs into the S-OnsagerNet. In other words, the approach used is purely data-driven and can therefore be generally applied to other non-equilibrium problems.

\textbf{Experimental data.}
We provide the details of the experimental validation of
the S-OnsagerNet results.

\textit{Experiments on electrokinetic stretching of DNA.}
To collect the experimental data leading to Fig.~\ref{fig:experimental}, we
needed to create an automated single molecule stretching trap. We describe the
essential features of the experiments, while the reader can find additional
details about the trap and the material preparation methodology in Soh et al.~\cite{Soh2023}.
The polymer samples used for this study were $T4$ phage
double-stranded DNA ($165.6$ kbp, Nippon Gene), chosen for high monodispersity
and ready availability. The DNA was diluted in buffer solution, and
fluorescently labelled (YOYO-1) to aid visualization.

The trap was based on a microfluidic cross-slot channel device with a wide
central chamber, splitting into 40 $\mu m$ wide channels in each of the four
cardinal directions (Supplementary Figure~3). Each of the four
channels was terminated by a macroscopic reservoir, in which we pipetted the DNA
sample and inserted platinum wire electrodes. These electrodes were connected to
a computer-controlled analog voltage source, such that the North and South
electrodes were grounded, but East and West were positively biased $V_0=+30V$.
This electric quadrupole arrangement set up a potential well in the North-South
direction and a potential hill in the East-West direction. The saddle point was
located in the central chamber.

In aqueous solution, DNA is naturally negatively charged, and in the
microfluidic device, molecules drifted electrokinetically from the North and
South reservoirs towards East and West. At the saddle point, the East-West
potential hill could be exploited to stretch a DNA molecule, but being an
unstable equilibrium point, molecules would approach its location and slow down,
but could not remain there for long observation times without external
intervention.

To actively trap and observe single DNA molecules at the saddle point, the
microfluidic electrokinetic device was placed on an inverted fluorescence
microscope (Nikon Eclipse $\mathrm{Ti}2\mathrm{U}$) with a $60\mathrm{X}$
oil-immersion objective ($1.4\mathrm{NA}$). A live fluorescence image was
captured by a sCMOS camera (Teledyne Photometrics Prime $95\mathrm{b}$) at 50ms
intervals and sent to a desktop computer that analyzed the incoming images in
real time. The image analysis included a fast cleanup step (detailed in the next
section) that removed background noise and stray `passer-by'
molecules/fragments, followed by a step to calculate the intensity centroid of
the target molecule and its projected length. The displacement $x$ of the
molecule’s centroid from the saddle point was input to a proportional feedback
loop that output a voltage tilt $\Delta \mathrm{V} \propto x$, which was then
superimposed on the East-West electrode biases, such that $\Delta \mathrm{V}
=\mathrm{V}_0 \pm Gx$. Setting $G=2.2 \ \mathrm{V}/\mu \text{m}$ confined the
DNA centroid to within $1 \ \mu \mathrm{m}$ of the saddle point even while a
molecule stretched. This feedback process for tracking and trapping DNA
molecules was automated through a custom LabVIEW program.

In a single stretch experiment, the platform was programmed to actively search
for a molecule to trap as they flowed through the central chamber. Once a
molecule was trapped, $\mathrm{V}_0$ was set to zero temporarily to allow the
molecule to relax into an un-stretched, equilibrium state over $10\ \mathrm{s}$
(chosen based on experience from preliminary experiments). After this relaxation
period, recording began with images being streamed to the computer's solid-state
drive. With every image, the associated centroid coordinates, projected length,
voltages, and other parameters were logged in real time. $\mathrm{V}_0$ was then
reset to $+30 \mathrm{V}$, which re-established the East-West potential hill and
stretched the DNA molecule. By monitoring the projected length history, the
platform could recognize when the molecule was fully stretched, and in response
stop the recording and release the molecule to escape naturally towards East or
West. Using this protocol, we were able to capture a diverse set of stretching
trajectories, including the dumbbell and folded conformations that were selected
for analysis in Fig.~\ref{fig:experimental}.

\textit{Real-time image processing.}
A DNA molecule deforms over time, and the purpose of this section is to
elucidate our method used to extract sequential image snaps that can be used to
capture the DNA unfolding. In order to capture the exact location of the DNA
molecule, we devised an algorithm to classify pixels according to the density of
a pixel's immediate surrounding. The algorithm exploited the fact that the
targeted molecule would very likely be centered in the tracking
Region-Of-Interest (ROI), and this was used to discriminate against stray
particles during tracking. To make this algorithm fast enough for real time, the
algorithm operated only on the binarized version of the raw image, after
application of a threshold. The result was a binary image containing the
targeted DNA pixels at the exclusion of pixels belonging to noise speckles and
stray particles. See Supplementary Figure~4 as an illustrated example of the real
time image processing pipeline.

One would imagine that calculating the intensity centroid is a natural method
for defining the molecule location and tracking it across image frames in time.
However, doing so with the raw DNA images is problematic for the two following
reasons:

First: Background noise from the camera biases the centroid towards the center
of the frame. This occurred even though we subtracted a pre-recorded background
frame. The problematic noise in this case arises from pixel shot noise
(electronic, photon) and pixel readout, and manifests as low-intensity sparkles
in the image. One possible solution is to increase the excitation intensity, but
this caused the DNA to photocleave much more readily, too often breaking into
fragments before stretching out fully.

Second: While a trapped DNA molecule was being stretched, stray molecules
continued to drift into the camera's field of view, which biased the centroid
towards these stray molecules. This problem was partially solved by calculating
the centroid from a smaller ROI that was just large enough to cover the
fluctuating motion of a single molecule. This ROI was software-based and
centered on the molecule centroid for tracking, i.e. implemented as a subset of
the image array. The reduced ROI size avoided most of the stray particles, but
it was still common for some to enter the tracking ROI during a trap-stretch
cycle.

\textit{Data selection.}
Among the stretching trajectories collected, a few were selected for further
analysis.  Specifically, examples of molecules adopting the ``dumbbell'' and
``folded'' conformations during the stretching process with similar initial
extension lengths were chosen in order to benchmark to human labeling.  The
trajectories of the conformations as the molecules were stretched were plotted
in the reduced coordinates space (Fig.~\ref{fig:experimental}).  When in the
fully stretched stable state, the DNA molecules still undergo conformational
fluctuations due to Brownian motion, which result in fluctuations in the reduced
coordinate space. Images of molecules in this stable state were analyzed to
obtain the fluctuations in $Z_1$, $Z_2$ and $Z_3$, as plotted in Fig.~\ref{fig:experimental}.
The method to extract the learned thermodynamic
coordinates $Z_2,Z_3$ from experimental images are described next.

\textit{Extracting reduced coordinates from filtered images.}
Different from the simulated data described earlier, given a filtered image, it
is more challenging to obtain in a robust way a sequence of ordered coordinates
representing the location of each part of the DNA molecule. However, identifying
the two end points is easier.  Thus, we first find the coordinates of the two
end points and also the centre of mass (weighted by intensity) of all
illuminated pixels (see Supplementary Figure~5).  Then, we
estimate $Z_2,Z_3$ by computing the end-to-end distance and foldedness as
defined in Fig.~\ref{fig:Z2_Z3_meaning}, which are shown to strongly correlate
with $Z_2$ and $Z_3$ respectively.  These computations only require the relative
position vectors of the two end points with respect to the centre of mass
($r_{1}=(r_{1,x},r_{1,y},0)$ and $r_{N}=(r_{N,x},r_{N,y},0)$).  The third
coordinate is set to $0$ because the DNA molecule is confined to have limited
motion in this direction.  The center of mass $r_{\text{cm}}$ is computed by
\begin{align}
    r_{\text{cm}}=\frac{1}{\bar{I}}\sum_{m,n}u_{m,n} I_{m,n},\nonumber
\end{align}
where $\bar{I}=\sum_{m,n}I_{m,n}$, $u_{m,n}$ is the spatial position of the
$(m,n)$-pixel, $I_{n,m}$ is the intensity of the image at pixel $(m,n)$.  This
allows for the computation of $r_1$ and $r_N$ (see Fig.~\ref{fig:filter low}).
According to the definition of end-to-end distance and foldedness and their
linear correlation with $Z_2, Z_3$ (Fig.~\ref{fig:Z2_Z3_meaning}), we may
compute the reduced coordinates as
\begin{align}
    Z_1&=C_1 L,\nonumber\\
    Z_2&=C_1C_2 (r_{1,x}-r_{N,x}),\nonumber\\
    Z_3&=C_1[C_3(r_{1,x}+r_{N,x})+C_4],\nonumber
\end{align}
where $L$ is the extension length scale
of the experimental data,
and the parameters $C_1$, $C_2$, $C_3$ and $C_4$ are scaling parameters
to account for the change in molecule configuration properties (e.g. length scales)
from simulation to experimental data.
In particular
\begin{itemize}
    \item
    $C_1$ is obtained by dividing the average extension of unfolding simulated data by that of experimental data.
    \item $C_2$ is the relationship between $Z_2$ and end-to-end distance (see Fig.~\ref{fig:Z2_Z3_meaning} (B)),
    and is obtained by dividing
    the average value of $Z_2$ in simulated data by that of the end-to-end distance.
    \item $C_3$ and $C_4$ is obtained by the relationship between the $Z_3$ and foldedness.
    We use least squares method to get $C_3$ and $C_4$ (see Fig.\ref{fig:Z2_Z3_meaning} (E)).
\end{itemize}
The method of extraction differs in the plot of Fig.~\ref{fig:experimental} (J).
Here, we only consider the stretched state, for which
we can extract 300 coordinates (assuming the third coordinate $z=0$)
equally spaced along the stretched polymer.
A scaling along the x-axis is performed so that
the average full extension of the experimental polymer data
matches that of the simulation data.
These coordinates are then fed into the trained
PCA-ResNet to extract the $Z_2$ and $Z_3$ values.

\paragraph{Polymer dynamics analysis}

In this section, we detail the from the modelling of the polymer
unfolding problem using the S-OnsagerNet.

\textbf{Accurate Prediction of the statistics of unfolding.}
In our training process, we have 610 training trajectories and 110
testing trajectories. Note that although the dataset is
generated by known yet complex microscopic equations, our approach does not
require, nor rely on, the knowledge of these equations.
The chain extension evolution of the training data
with different initial chain are shown in Supplementary Figure~6(A)
(black) and the test data are shown in Supplementary Figure~6(E) (black).
After training the model, we predict the extension of the polymer (red).
We can see the extension of the polymer can be predicted well.
We also compute the mean (Supplementary Figure~6(B,F)),
standard derivation (Supplementary Figure~6(C,G)) and the distribution of unfolding time
(Supplementary Figure~6(D,H)) of the training and test results. We also compute
the error of the mean (relative $L^2$ error), standard derivation, and the
probability distribution of unfolding time of the training and test results
in Supplementary Table~1.
We observe that our model successfully captures the
statistical behavior of a polymer stretching
with only a 3-variable dynamical system.

\textbf{Interpreting the learned closure coordinates.}
Supplementary Figure~7 shows the evolution of chain extension
and the two learned closure coordinates with time for the training data. The
trajectories are colored by unfolding time. Based on the unfolding times, we
observe that chains with similar initial extensions (determined by the
y-intercept, i.e., $Z_1$ at $t=0$) can take vastly different times to
stretch, hence it is not sufficient to consider only chain extension for the
purposes of predicting dynamics.
However, the successful prediction of the
statistics of unfolding dynamics implies that $Z_2$ and $Z_3$ capture crucial
information of the system. Thus, we seek to gain some physical understanding of
the learned closure coordinates $Z_2$ and $Z_3$.

Here, we provide details on the analysis of the closure coordinates ($Z_2,Z_3$)
that characterize the stochastic evolution of the extension length ($Z_1$),
Recall that $Z_2, Z_3$ are deterministic functions of the microscopic
configuration $X$, and the functions are approximated by a trained PCA-ResNet.
A useful property of we can exploit is differentiability of the neural network.
We can ask: given a certain configuration $X$, what kind of small perturbations
to $X$ will most drastically increase or decrease the value of $Z_2$ and $Z_3$?
In other words, we can consider perturbations in the directions of $\pm
\partial Z_2 / \partial X$ and $\pm \partial Z_3 / \partial X$ respectively.
We first analyze the closure coordinate $Z_2$.  From
Fig.~\ref{fig:Z2_Z3_meaning}(A), we observe that perturbations
in the direction of $Z_2$ tend to change the end-to-end distance in the
elongational axis (distance between the first and the last bead in the polymer
chain along the elongational direction, i.e. $\vert r_{N,x}-r_{1,x}\vert$), where
$r_{j,i}$ is the $i^\text{th}$ coordinate of the $j^\text{th}$ bead in the
chain (cyan: $Z_2=0.453$, blue: $Z_2=0.194$, and black: $Z_2=0.323$).  We
confirm this hypothesis by visualizing the correlation of the end-to-end
distance and the magnitude of $Z_2$ in Fig.~\ref{fig:Z2_Z3_meaning}(B,C).
Generally, we observe that as $\vert Z_2\vert $ decreases, the
distance between the chain ends (marked by red points in the figure) decreases.
Thus, we can interpret the first learned closure coordinate as an indicator of
end-to-end distance.

We proceed with a similar analysis for the other closure coordinate $Z_3$.
Fig.~\ref{fig:Z2_Z3_meaning}(D) shows a given chain
configuration and perturbations in the positive and negative directions of
$\frac{\partial Z_3}{\partial X}$ (cyan: $Z_3=7.903$, blue: $Z_3=1.595$, and
black: $Z_3=4.749$).  Here, we observe that the end-to-end distance is largely
unchanged, but the degree of foldedness of the chain in the elongational axis
of the flow ($x$ direction) appears to change.  This leads us to hypothesize
that the second learned coordinate represents a degree of foldedness with
respect to the elongational flow. As a measure of the degree of foldedness of a
chain, we compute $\vert r_{1,x}+r_{N,x}\vert $.  During data pre-processing,
the chain is centered such that its center-of-mass is 0. Hence, if $\vert
r_{1,x}+r_{N,x}\vert $ is small, the polymer is symmetric around 0 in the
elongational $x$ direction and tends to be in the elongated state. If $\vert
r_{1,x}+r_{N,x}\vert $ is large, the polymer is likely to be in the folded
state. We plot $\lvert  r_{1,x}+r_{N,x}\rvert $ as a function of $\vert
Z_3\vert $ for all configurations in the training data set in
Fig.~\ref{fig:Z2_Z3_meaning}(E). The strong correlation
between $\vert r_{1,x}+r_{N,x}\vert $ and $\lvert  Z_3\rvert$ supports our
interpretation that the second learned closure coordinate is an indicator of
the degree of foldedness in the elongational direction. To demonstrate this, we
plot in Fig.~\ref{fig:Z2_Z3_meaning}(F) visualizations of
different chains with a range of $\lvert  Z_3\rvert $ values, with the chain
ends marked by red points. As $\lvert  Z_3\rvert$ decreases, we observe the
chains generally shift from the folded to the elongated state. We note that the
degree of foldedness is sufficiently described solely by considering the
projection of chain coordinates onto the elongational axis of the flow, as the
flow is stable in the compressional axis and thus the degree of foldedness is
primarily relevant to the unstable elongational axis that drives the unfolding
process.

\textbf{Advancing classification methods for polymer stretching.}
With the new understanding of the closure variables, we now consider how our
results improve the current understanding of polymer chain dynamics. In the
landmark experimental study of dilute polymer chains stretching under
elongational flow, the molecules were categorized into seven different
conformations and the dynamics of dominant conformations were
analyzed~\cite{Perkins1997}. Specifically, it was found that chains in the
``folded'' conformation (which is one of the seven categories) took the longest
time to stretch, while chains in the dumbbell conformation stretched relatively
quickly (Supplementary Figure~8).

Our analysis shows that instead of a categorical labelling, it is perhaps more
useful to characterize the stretching dynamics of a polymer by three numbers
representing the generalized coordinates: $Z_1$ (extension length), $Z_2$
(related to end-to-end distance), and $Z_3$ (related to foldedness in the flow
direction). Our results show that these are sufficient to predict the dynamics
of the chain extension. We show that this characterization is largely
consistent with previous categorical ones, but improves upon them in some
intermediate cases. In Supplementary Figure~9, we plot the values of the
low dimensional coordinates $\lvert  Z_2\rvert $ and $\lvert  Z_3\rvert$ of
different chains with initial dumbbell and folded configurations at selected
chain extension $Z_1$ values, colored by the predicted unfolding times. We
observe segregation between the folded and dumbbell configurations in the
$Z_2-Z_3$ space, indicating that the qualitative differences between different
conformations can be captured by our characterization. Generally, the folded
chains take a longer time to stretch compared to the dumbbell chains. This is
consistent with experimental and computational observations reported in the
literature~\cite{Perkins1997,Smith1998,Larson1999}. However, we highlight that
the region with high $\lvert  Z_2\rvert $ and low $\lvert  Z_3\rvert $ values
encompasses a mix of folded and dumbbell chains with similar unfolding times.
Therefore, while the broad categorization scheme allows for coarse
discrimination of the stretching dynamics, the qualitative classification does
not allow for finer predictions. Instead of classifying the stretching
trajectories based on qualitative, human judgement of chain conformation during
the process, we present a robust, quantitative approach to interpreting the
stretching dynamics that involves consideration of the initial chain
configuration in reduced dimensions.

\textbf{Free energy landscape analysis.}
We now provide details on the analysis of the free energy
landscape. We begin with an important technical note. Our learned GSOP
following Eq.~\eqref{eq:sto_onsa1} is in general not guaranteed to be a
gradient system, unless $W(Z)=0$ and $M(Z)=I$. However, since the drift term
$f(Z)=(f_1(Z),f_2(Z),f_3(Z))^T=-(M(Z)+W(Z)) \nabla V(Z) $, the stationary
points of $V$ are also critical points of the dynamics ($f(Z) = 0$). The saddle
points of $V$ are the \emph{saddle foci} of the non-gradient dynamics $\dot{Z}
= f(Z)$. For simplicity, we refer to a saddle point of $V$ and saddle focus of
$f$ interchangeably.

We compute the critical points by numerically solving $\nabla V(Z) = 0$ with
the BFGS method from different initial conditions.
We obtain four critical points:
$(247.15 , 1.701, 0.193)$ (yellow point),
$(247.29 , -1.700, -0.166)$ (blue point),
$(87.858, -0.041, \\-4.269)$ (cyan triangle) and
$(85.887 , -0.050,  4.126)$ (magenta triangle),
which are shown in Fig.~\ref{fig:potential}. The Jacobian matrix of the drift
term is then computed
\[J(f)= \left(    \begin {array}{cccc}
    \partial_{Z_1}f_1 &\partial_{Z_2}f_1&\partial_{Z_3}f_1\\
    \partial_{Z_1}f_2 &\partial_{Z_2}f_2&\partial_{Z_3}f_2\\
    \partial_{Z_1}f_3 &\partial_{Z_2}f_3&\partial_{Z_3}f_3\\
    \end{array} \right).  \]
We calculate the eigenvalues and eigenvectors of $J$ at the
four critical points. It has three eigenvalues with negative
real parts at the yellow and blue points, so these are the stable
points (stable nodes). We name them $Z_{\text{stable},1}$ and
$Z_{\text{stable},2}$. For the cyan triangle and magenta diamond points, $J$
has one eigenvalue with positive real part, and two with negative
real parts. These points are saddle points (saddle focus) of index 1. We call
these two points $Z_{\text{saddle},1}$ and $Z_{\text{saddle},2}$. The direction
of the eigenvector corresponding to the eigenvalue with positive real part
is the unstable manifold, along which the trajectories escape from the saddle.
On the other hand, the span of
the real and imaginary parts of the other eigenvector (one must be a complex
conjugate of the other) constructs the stable manifold, along which
trajectories are attracted to the saddle.

The local behavior of the learned potential $V(Z)$ around its critical points
characterize the typical fluctuations around them. There lies important
physical meaning that we can probe.
For example, let us exploit automatic differentiation
of neural networks and expand $V(Z)$ in a Taylor series
around $Z_{\text{stable},1}$ (fully stretched state),
corresponding to $(247.15 , 1.701, 0.193)$.
Neglecting small terms, we obtain the approximate
formula
\begin{equation}\label{SIeq:stable_state}
    \delta V \approx 153.1(\delta Z_1-1.54\delta  Z_2)^2+205.5\delta Z_2^2+36.96\delta Z_3^2,
\end{equation}
where $\delta Z_i=Z_i-[Z_{\text{stable},1} ]_i$ is the fluctuations in the thermodynamic variables.
Now, assuming that the distribution of states
around $Z_\text{stable}$ follows a Boltzmann distribution
$Z\sim \exp[- V(Z) / {k_B T}]$ (here we assume that $M\sim I$)
where $\vert\sigma\vert^2 \sim k_B T$,
the typical small fluctuations $\delta V$ is proportional to $k_B T$.
In other words, Eq.~\eqref{SIeq:stable_state} is approximate ``isotherms'' that captures
the form of typical fluctuations.
For example, we can infer from the formula that
typical fluctuations of the extension length ($Z_1$) and
the end-to-end distance ($Z_2$) are highly correlated.
This is sensible, since in the fully stretched state,
these two quantities are expected to change simultaneously.
In Fig.~\ref{fig:potential}(G,H,I), we confirm these correlations.

Similarly, one can also expand $V(Z)$ around a saddle point
(fully folded state)
$Z_{\text{saddle},2}$.
We obtain the formula
\begin{equation*}
    \delta V \approx 102.96\delta Z_1^2-31.13(\delta Z_2-0.255\delta  Z_3)^2+24.16\delta Z_3^2,
\end{equation*}
Here, we immediately observe that to escape the saddle point
(lowering the energy),
one should increase end-to-end distance ($Z_2$) and decrease foldedness ($Z_3$).
This approximately aligns with the unstable manifold described above,
and forms the basis of our control protocols described
in the main text.

\textbf{Selection of the dimension of the reduced coordinates.}
In this part, we describe how we arrived at the selection of a 3-variable
reduced coordinate space. We tested various different number of reduction
dimensions ($d=2,3,4$) and the relative errors of predictions are summarized in
Supplementary Table~1. We observed that going to a higher dimension
$(d=4)$ did not result in noticeable gains. In fact, increasing dimension may
cause increasing optimization error and hence worse results in standard
deviation. Going to a lower dimension $(d=2)$ resulted in increased prediction
error in general.

Interestingly, we can formulate a physical argument that suggests that a
2-dimensional system (with only 1 additional closure coordinate) is not a
suitable reduced model for the polymer dynamics we study. The argument is based
on index theory for 2-dimensional dynamical
systems~\cite{strogatz2018nonlinear}.

Let us assume for the sake of contradiction that
the projected free energy landscape of our learned 3-dimensional system
into $Z_1-Z_2$ plane is the free energy of a system in 2-dimensional
(Supplementary Figure~10).
Since $Z_1$ is the polymer extension, we expect that
\begin{enumerate}
    \item
    There exists 2 stable states at large $Z_1$
    corresponding to the fully extended state.
    There are 2 of them due to reflection symmetry
    in the flow direction.
    Around this $Z_1$ value, all
    trajectories in the reduced space
    should converge to one of the stable steady states.
    \item
    There cannot be saddle points with the same $Z_1$
    value, since it is close to the maximal extended chain length.
\end{enumerate}
These conditions are enough to imply a contradiction in the following way.  The
simple concept we use from index theory is the definition of the index of a
closed curve in the phase space of a 2-dimensional dynamical system (See
Strogatz~\cite{strogatz2018nonlinear}, Chapter 6.8 for details).  Let $\Gamma$
be a closed curve in $\mathbb{R}^2$, and consider a dynamical system $\dot{z} =
f(z)=(f_1(z),f_2(z))^T$, $z=(z_1,z_2)^T\in\mathbb{R}^2$.  The \emph{index} of
$\Gamma$ with respect to the dynamics $f$ is defined as
\begin{align}\label{eq:index}
    I_\Gamma(f) := \frac{1}{2\pi}\oint_\Gamma \frac{f_1df_2-f_2df_1}{f_1^2+f_2^2}.
\end{align}
Intuitively, this is the sum of the angles of the force vectors across
the curve $\Gamma$.

From index theory, we know that the index of a closed curve must equal to the
sum of indices of critical points it encloses.  Moreover, the index of a stable
critical point is $+1$, and the index of a saddle is $-1$.  Now, we draw a curve
enclosing the two stable points as shown in Supplementary Figure~10.
By condition 1 above, the vector fields should point inwards towards the interior of the
curve, thus we can show (from Eq.~\eqref{eq:index}) that the index of this curve
is $+1$.
However, it can only enclose stable critical points due to condition
2, and consequently its index is at least $+2$.  Thus we arrive at a
contradiction, showing that the dynamical landscape in
Supplementary Figure~10 is not possible in 2 dimensions.

The only remaining possibility is that $Z_2$ does not distinguish the two
symmetric stable extended states.  However, in this case to allow a saddle point
(which is the key to inducing heterogeneity in unfolding) we must have another
stable steady state at a different $Z_1$ value.  This is unlikely from a
physical viewpoint for non-self-interacting polymers, since we expect the only
asymptotically stable state to be the stretched, fully extended state.

\textbf{The impact of dataset size on predictive accuracy of S-OnsagerNet.}
In order to study the impact of the number of training data on the
computational accuracy of the model, the original training dataset containing
610 trajectories was split into two subsets containing 25\% and 75\% of the
data. A third subset was produced by selecting at random 50\% of the data from
the original dataset. The three datasets and the learned potential landscapes
are presented in Supplementary Figure~11. The 25\% and 75\% datasets
have no common trajectories, whilst the 50\% dataset contains some trajectories
from each of the other two datasets. The quantitative results are reported in
Supplementary Table~2, where the trained models are used to predict 500
unseen trajectories with fast, medium and slow unfolding times, suggest that
with a smaller amount of training data, the S-OnsagerNet model loses some
predictive accuracy. However, the more important factor is the diversity of the
data contained in the training datasets. The full dataset has a high proportion
of trajectories with fast, followed by middle (slower), and then slow unfolding
times, however it is the most balanced of all the training datasets (25\%,
50\%, 75\%, 100\%), in addition to containing the most data. Training with the
25\% dataset, which contains the largest proportion of trajectories with fast
unfolding times results in the lowest $L^2$ error, but the model over-fits to
that type of trajectories, and has the largest $L^2$ error for the medium and
slow unfolding times. On the other hand, the 50\% dataset is the most balanced
of the reduced datasets, which is reflected in the relative $L^2$ errors. It
can also be observed in Supplementary Figure~11 (bottom) that the
potential landscapes resulting from training with smaller datasets do not
contain a clear stagnation (saddle) point. Overall, for training datasets with
different number of trajectories and their respective proportions of fast,
medium and slow trajectories, the results, both in terms of prediction error
and potential landscape characteristics suggest that the model is relatively
robust, and the amount of training data cannot be reduced much without
affecting predictive performance.

\paragraph{Spatial epidemics analysis}

In this section, we provide details of our analysis method for an alternative
application of S-OnsagerNet -- modelling the macroscopic dynamics of the spread
of epidemics. This highlights the general applicability of our method.

We focus on the most well-known model for disease spread in a spatial domain -
the spatial SIR model~\cite{murray2001mathematical}. Let us consider a
two-dimensional square domain (representing a city, say) discretized into $n
\times n$ sectors. We use $I_{i,j}$ and $S_{i,j}$ represent the number
(density) of infective and susceptible individuals at spatial location $(i,j)$.
The basic mechanism of the model is as follows: each infective individual may
infect a susceptible individual in the same spatial location. At the same time,
each infective recover (or is removed) at a rate, after which they are no
longer infective. Finally, both infective and susceptible individuals move on
the spatial domain randomly. Mathematically, the temporal evolution of $I,S$
(understood as $n \times n$ matrices, or length $n^2$ vectors) are governed by
the following dynamics
\begin{align}\label{eq:2d_spatial_sir}
     \dot{S}_{i,j}=&-\beta I_{i,j}S_{i,j}
     +\frac{\delta}{\delta_x^2} (S_{i-1,j}-2S_{i,j}+S_{i+1,j}) \nonumber\\
     &+\frac{\delta}{\delta_y^2} (S_{i,j-1}-2S_{i,j}+S_{i,j+1})+\sigma \dot{B}_1(t),  \nonumber\\
     \dot{I}_{i,j}=&\beta I_{i,j}S_{i,j}-\gamma I_{i,j}+\frac{\delta}{\delta_x^2} (I_{i-1,j}-2I_{i,j}+I_{i+1,j})\nonumber\\
     &+\frac{\delta}{\delta_y^2} (I_{i,j-1}-2I_{i,j}+I_{i,j+1})+\sigma \dot{B}_2(t).
\end{align}
As usual, the dot denotes time derivative. The parameter $\beta$ is a measure
of the transmission efficiency of the disease from infectives to susceptibles,
and $1/\gamma$ is the life expectancy (or expected recovery time) of an
infective. The constant $\delta$ is the diffusion coefficient, and this term in
the equation models the spatial movement of individuals as a diffusion process
over the domain. The last terms of the equation models the stochastic
fluctuations of the number densities, with  $\sigma$ as the noise intensity.
The parameters $\delta_x$ and $\delta_y$ are the spatial discretization sizes in the two
spatial directions. In our simulations, we take $n=16$, $\beta = 0.3$, $\gamma
= 0.13$, $\delta = 0.5$, $\sigma=0.03$ and $\delta_x=\delta_y=\frac{2}{3}$.
Eq.~\eqref{eq:2d_spatial_sir} governs the microscopic dynamics of disease
spread, and non-trivial outcomes can result from different initial spatial
configurations and parameters (e.g. infection rate, recovery rate). See
Extended Data Figure 1.

While this microscopic model and its variants has been subject to intense study
(see Murray~\cite{murray2001mathematical} and references therein), a
macroscopic understanding of the dynamics of disease spread is challenging due
to the complex spatial interactions. For example, one may be interested to
model the dynamics of average (or total) number of infective and susceptible
individuals over the spatial domain. Observe in Extended Data Figure 1 that
configurations with identical initial spatial averages of infective and
susceptible individuals can have drastically different subsequent evolution.
More precisely, one spatial configuration Extended Data Figure 1(A) can lead to
initial disease spread (epidemic), where the mean number of infected
individuals initially increases sharply, whereas another spatial configuration
Extended Data Figure 1(B) leads to the disease dying out monotonically. Thus,
it is of interest to develop a thermodynamic description to capture and
elucidate the driving factors of such variations.

\textbf{Modelling the thermodynamics of the spatial SIR model using S-OnsagerNet.}
Following our framework, we now set the macroscopic variables of interest as
the respective spatial averages $Z_1 =\delta_x \delta_y \sum_{i,j}^n I_{i,j}$ and $Z_2 =\delta_x
\delta_y \sum_{i,j}^n S_{i,j}$. Recall from Extended Data Figure 1 that these
variables alone are insufficient to determine their subsequent evolution. Our
goal is to learn closure variable(s), and a stochastic dynamics that describes
the evolution of these variables.

Training data is generated through integrating Eq.~\eqref{eq:2d_spatial_sir}
using the Euler-Maruyama method with time step size $dt=0.03$. The initial
conditions are selected as $I_{i,j}(0)=5 e^{-(x_i-x_0)^2-(y_j-y_0)^2}$ and
$S_{i,j}(0)=5 e^{-(x_i-1)^2-(y_j-1)^2}+5 e^{-(x-1)^2-(y+2)^2}$, where $x_0$ and
$y_0$ are randomly generated from a uniform distribution in the unit square,
and $x_i=-5+i \delta_x $ and $y_j=-5+j \delta_y$. This initial condition
corresponds to the scenario where the initial susceptible population is fixed,
but the initial infective population is a cluster that is uniformly and
randomly distributed in the domain. Then, we carry out the S-OnsagerNet
workflow as shown in Fig.~\ref{fig:workflow}, with $\beta=0.01$,
$\alpha=0.001$, $m=50$ and $d=3$ (i.e. one closure coordinate). We employ a
PCA-encoder network to obtain the closure coordinate. The training process
involves initially training the PCA-encoder and the S-OnsagerNet
simultaneously. Subsequently, we fix the former and continue training the
latter to a desired error tolerance. Note that in the SIR model case, we do not
utilize the decoder network as there is no need to obtain a reconstruction of
$I$ and $S$.

\textbf{Capturing the stochastic dynamics.}
First, we show that with just one additional learned closure coordinate $Z_3$,
we can capture the statistics of the macroscopic dynamics of the spatial
averages of infected and susceptible individuals. The results are shown in
Extended Data Figure 1(D,E), where true mean and standard deviation of $Z_1$ and
$Z_2$ are obtained from Eq.~\eqref{eq:2d_spatial_sir}, while the predicted
results are derived from the S-OnsagerNet. Four representative test initial
conditions are shown: two with disease spread and the other two without. We
observe that we can successfully capture the macroscopic dynamics of disease
spread with just one additional closure coordinate.

\textbf{Interpreting the closure coordinate.}
Next, since we have shown that only one closure coordinate is required for a
thermodynamic description, it is natural to probe its physical meaning. We use
the same technique described in the polymer stretching case, where we
investigate the effect of perturbation of a microscopic state in the directions
that cause the sharpest changes in $Z_3$. The microscopic state we probe is
chosen as a pair of partially overlapping clusters of infective and susceptible
individuals (Extended Data Figure 2(A)). We observe that the perturbations
induced by $dZ_3/dX$ (where $X=(I,S)$) correspond to increasing/decreasing the
overlap of clusters of susceptible and infected individuals. Hence, this
suggests that $Z_3$ is a macroscopic descriptor that correlates with such
effective spatial overlap. We confirm this hypothesis via a scatter plot in
Extended Data Figure 2(B), where the overlap is defined by $IS_{mean}=\delta_x
\delta_y\sum_{i,j}^n I_{i,j}S_{i,j}$. Thus, we can interpret the closure
coordinate $Z_3$ as an indicator of the overlap of clusters of susceptible and
infected individuals. This is physically sensible, since the measure of overlap
of spatial clusters should determine the outcome of an epidemic. Nevertheless,
we must emphasize that the learned $Z_3$ is a quantitative measure and can be
applied to more complex configurations than a pair of clusters, for which one
may not be able to easily define a notion of effective overlap by empirical
observation.

\textbf{Free energy landscape.}
Finally, we study the dynamical landscape of the learned S-OnsagerNet model.
Using the same projection technique in the polymer case, we plot 2D projections
of the learned 3D free energy landscape in Extended Data Figure 3, overlaid
with two representative trajectories (with disease spread in blue and
disease dying out in red).

A number of interesting features can be gleaned from the landscape. First, we
can clearly see the origins of the divergence of the two different types of
trajectories. While they have identical initial $Z_1$ (average infective) and
$Z_2$ (average susceptible) values, their initial $Z_3$ (infective/susceptible
overlap) values differ (Extended Data Figure 3(B,C)). In particular, the
initial disease spread seen in the red and blue trajectories is approximately
in accordance to the steepest descent of the energy landscape. That is, the
initial $Z_3$ value is a determining factor for the onset of epidemics. Second,
we compute using the steady states of the dynamics corresponding by solving
$\nabla V(Z) = 0$. Instead of isolated steady states as in the case of polymer
dynamics, we find a 1D manifold of stable steady states in the $Z_2-Z_3$ plane
(see Extended Data Figure 3(C,F)). This implies in particular that the
terminal state (the remaining number of susceptible individuals) is not unique,
but rather depends on the initial configuration, and in particular on the
initial overlap described by the learned coordinate $Z_3$. This rationalizes
the observed heterogeneity in the terminal configurations as shown in
Extended Data Figure 1.

\paragraph{Statistics \& reproducibility}

The train-test splits in this paper are performed by random uniform
sub-sampling. Repeats of numerical experiments are performed by running the
same code with different random seeds. No data were excluded from the analyses,
and the investigators were not blinded to allocation during experiments and
outcome assessment.

\section*{Data availability}

The simulation and experimental datasets used are publicly available in the
Harvard Dataverse public repository~\cite{soh_2023_data}.
The simulation data was generated according to the methods
introduced in Methods under data preparation.
Source Data for Figures
\ref{fig:unfolding_times}, \ref{fig:Z2_Z3_meaning}, \ref{fig:potential},
\ref{fig:design}, Extended Data Figures 1, 2 and 3 is available with this manuscript.

\section*{Code availability}

Code to reproduce the analysis generated within the study is
provided at \url{https://github.com/MLDS-NUS/DeepLearningCustomThermodynamics}~\cite{chen_2023_10212239}.

\section*{Acknowledgements}

This research is supported by the Ministry of Education, Singapore, under its Research Centre of Excellence award to the Institute for Functional Intelligent Materials (I-FIM, project No. EDUNC-33-18-279-V12 KSN). KSN is grateful to the Royal Society (UK, grant number RSRP \textbackslash R \textbackslash 190000 KSN) for support.
QL acknowledges support from the National Research Foundation, Singapore, under the NRF fellowship (project No. NRF-NRFF13-2021-0005 QL).
HY acknowledges support from the
National Natural Science Foundation of China under Grant No. 12171467 HY and 12161141017 HY.
KH, BWS and Z-EO acknowledge support from the Accelerated Materials Development for Manufacturing Program at A*STAR via the AME Programmatic Fund by the Agency for Science, Technology and Research under Grant No. A1898b0043 KH. The computational work for this article was partially performed on resources of the National Supercomputing Centre, Singapore.
The funders had no role in study design, data collection and analysis, decision to publish or preparation of the manuscript.
The authors thank Chang Jie Leong for help with experimental preparation.

\section*{Author contributions statement}
X.C., B.W.S., K.S.N, K.H. and Q.L. conceptualized the project and developed the
methodologies. X.C. and B.W.S. conducted the numerical experiments. B.W.S.,
E.V-.G. and Z.-E. O. set up and conducted the physical experiments.
X.C., H.Y. and Q.L. conducted the theoretical investigations.
All authors co-wrote the paper.

\section*{Competing interests statement}
Authors declare that they have no competing interests.

\section*{Figures}

\begin{figure}[H]
    \centering
    \includegraphics[width=\textwidth]{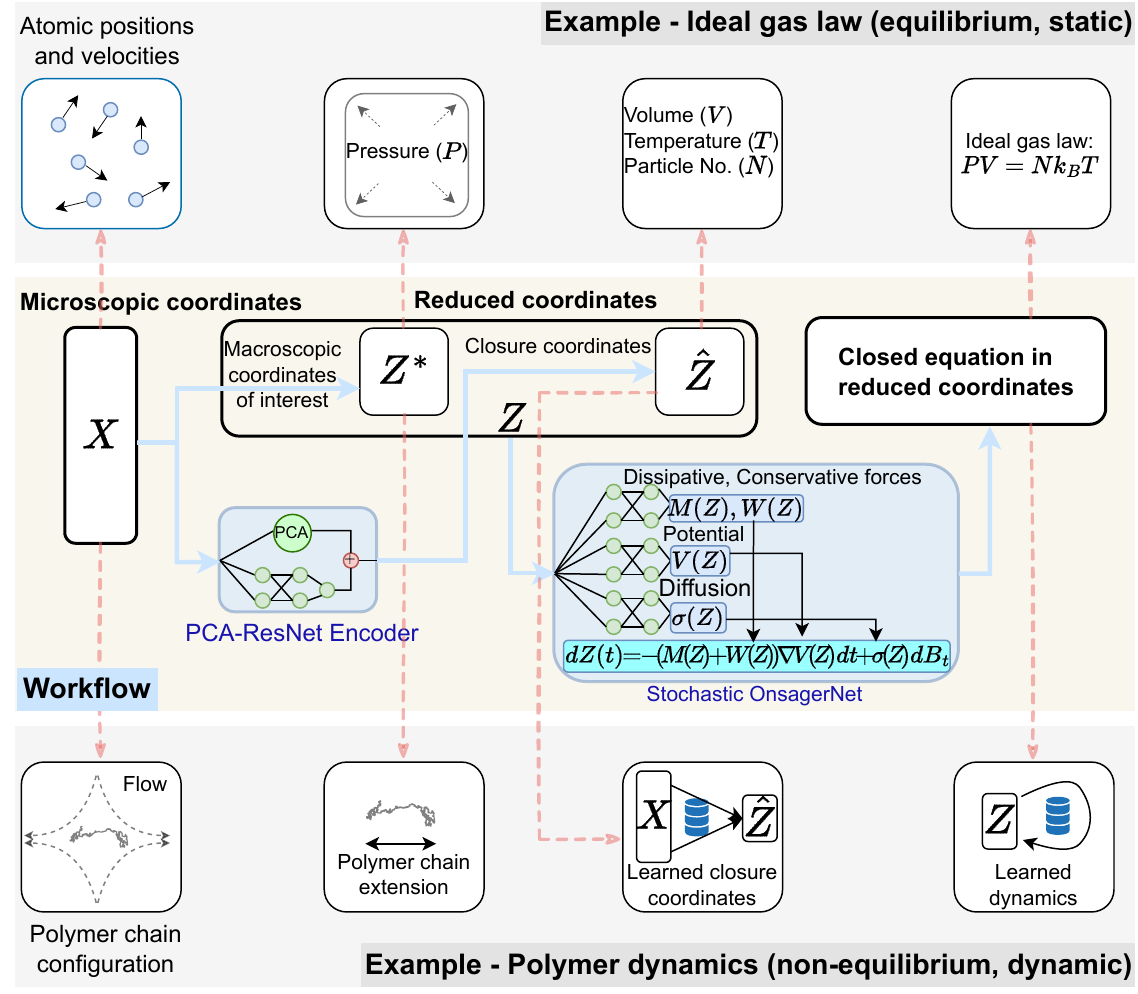}
    \caption{
        \textbf{Overall workflow of the proposed approach.}
        Given a complex system described by $X$, the goal is to model the behavior of macroscopic coordinates of interest  $Z^*$. We construct closure coordinates $\hat{Z}$ and closed (dynamical) equation on the combined reduced coordinates $Z=(Z^*,\hat{Z})$.  The classical ideal-gas law is an illustration of this process (top panel), where $k_B$ is the Boltzmann constant. For general non-equilibrium, dynamic systems (bottom panel), carrying out this workflow from theoretical analysis is challenging. Our machine learning method (mid panel) addresses this by simultaneously constructing the closure coordinates using PCA-ResNet (see Methods), and governing equations on reduced coordinates using the S-OnsagerNet with drift term $-(M(Z)+W(Z))\nabla V(Z)$ and noise term $\sigma(Z)$ (see Eq.~\eqref{eq:sto_onsa}).
    }
    \label{fig:workflow}
\end{figure}

\begin{figure}[H]
    \centerline{\includegraphics[width=\textwidth]{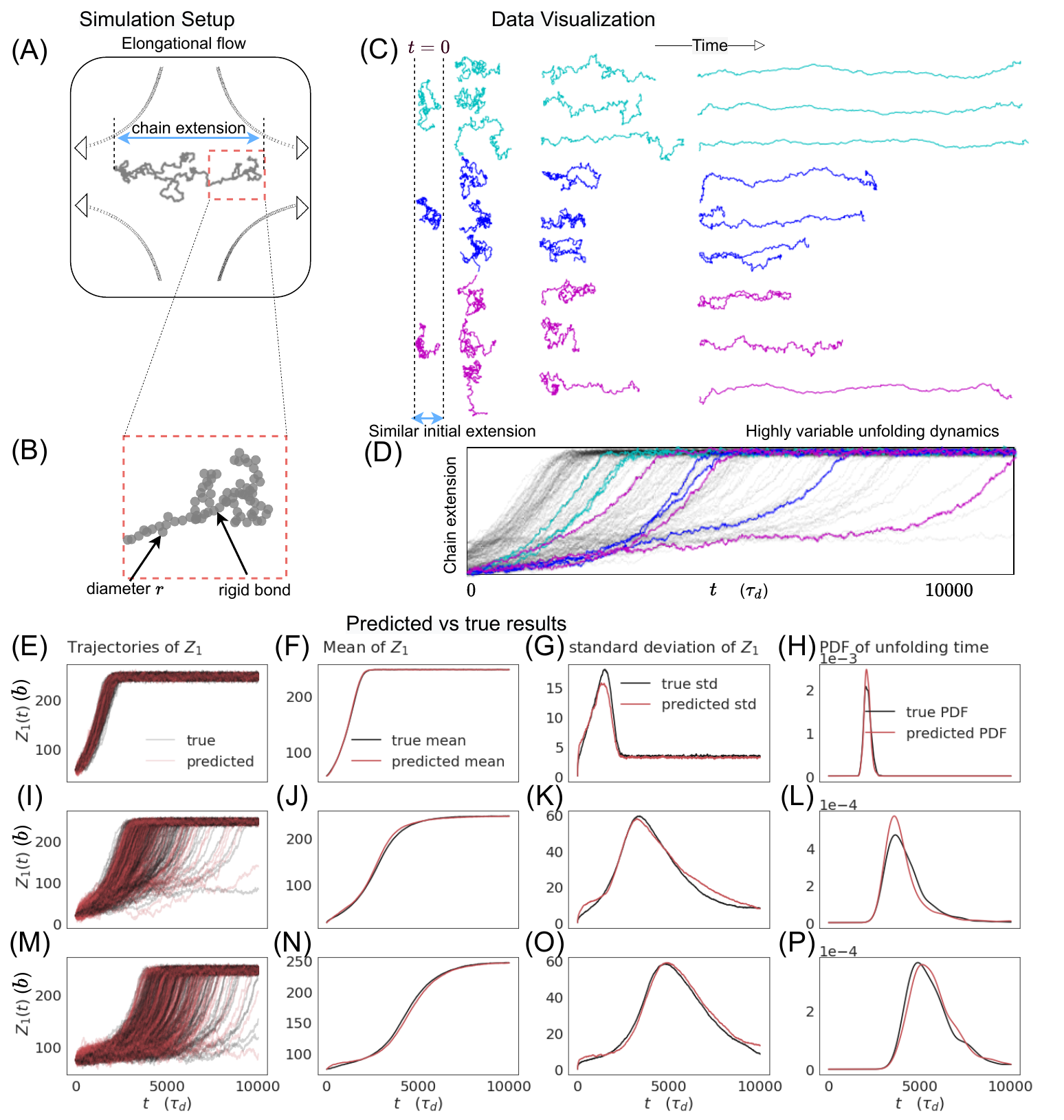}}
    \caption{
        \textbf{Simulation setup, data visualization, and predicted vs true stretching trajectories and their statistics.}
        (A,B) The polymer chain is represented by a bead-rod model with bead diameter $r$ (in units of $b=10\text{nm}$) and rigid bonds, subjected to hydrodynamic and Brownian forces. (C,D) The statistics of chain extension projected along the elongational axis are recorded. Times are reported in units of the characteristic rod diffusion time $\tau_d$ (see data preparation in Methods).
        Different initial conditions (colors) are chosen to have similar initial extension but varying unfolding times.
        Identical initial configurations also have different unfolding dynamics due to thermal fluctuations.
        In (E-P), we show that S-OnsagerNet can capture this
        heterogeneity. We plot in (E) 500 trajectories of polymer extensions ($Z_1$) from the same initial condition, together with their mean (F) and standard deviation (G). The probability density functions (PDF) of unfolding times is shown in (H). (I-L) and (M-P) show results for two other unseen initial chain configurations.
    }
    \label{fig:unfolding_times}
\end{figure}

\begin{figure}[H]
    \centerline{\includegraphics[width=\textwidth]{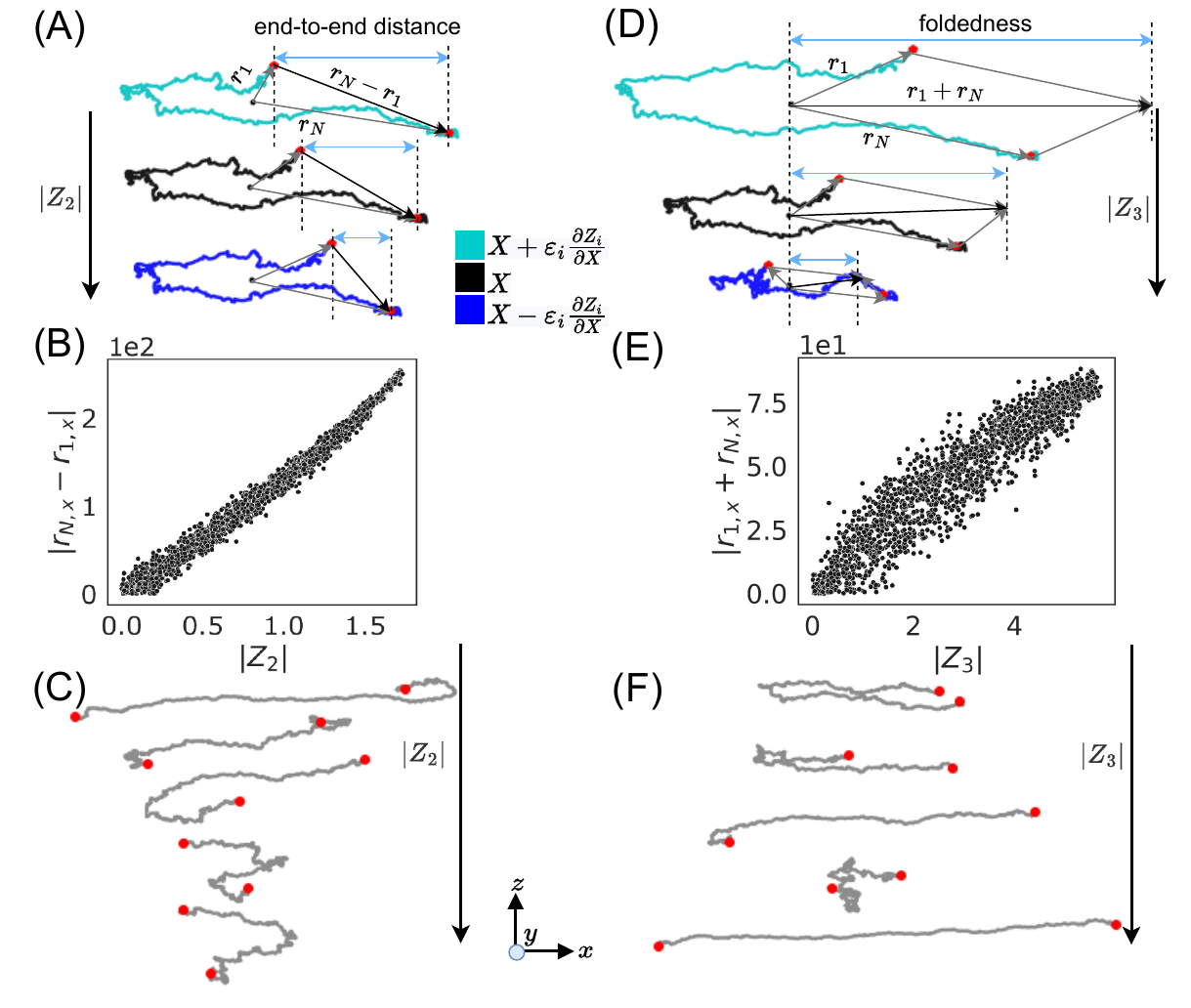}}
    \caption{
        \textbf{Physical interpretation of learned closure coordinates.}
        (A) Perturbation of the polymer chain $X\pm \varepsilon_2 \frac{\partial Z_2}{\partial X}$ from a given configuration, $\varepsilon_2=100/\|\frac{\partial Z_2}{\partial X}\|_{2}$. (B) Plot of projected end-to-end distance $\lvert  r_{N,x}-r_{1,x}\rvert $ as a function of $\lvert  Z_2\rvert $ for the training data. (C) Configurations of different polymer chains with decreasing $\lvert  Z_2\rvert $ values. As $\lvert  Z_2\rvert $ decreases, the projected distance between the chain ends decreases. (D) Perturbation of the polymer chain $X\pm \varepsilon_3 \frac{\partial Z_3}{\partial X}$ from a given configuration, $\varepsilon_3=260/\|\frac{\partial Z_3}{\partial X}\|_{2}$. (E) Plot of degree of foldedness $ \lvert  r_{1,x}+r_{N,x}\rvert $ as a function of $\lvert  Z_3\rvert $ for the training data. (F) Configurations of different polymer chains with decreasing $\lvert  Z_3\rvert $ values. As $\lvert  Z_3\rvert $ decreases, the chains tend towards a more stretched state.
    }
    \label{fig:Z2_Z3_meaning}
\end{figure}

\begin{figure}[H]
    \centering
    \includegraphics[width=\textwidth]{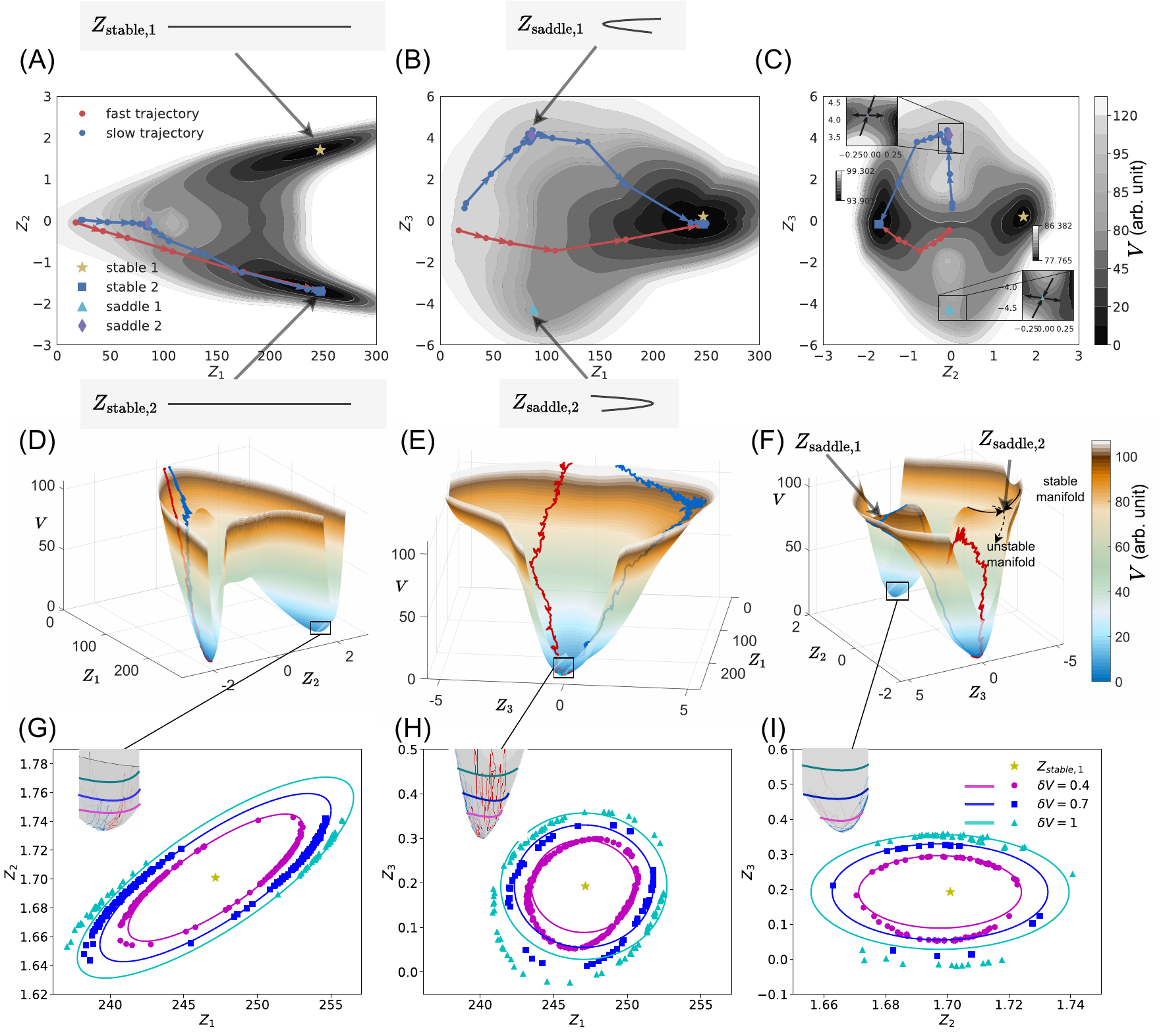}
    \caption{
    \textbf{ Learned potential energy landscape.}
        We plot $V$ projected onto $Z_1-Z_2$ (A,D), $Z_1-Z_3$ (B,E) and $Z_2-Z_3$ (C,F) planes (inset: stable and unstable directions of the saddle points). Projection is computed via minimization (e.g. $V(Z_1,Z_2) = \min_{Z_3} V(Z_1, Z_2, Z_3)$), which at low temperatures closely approximates marginalizing with respect to the Boltzmann distribution. The stable and saddle points are marked on the energy landscape, and their corresponding reconstructed fully extended and folded states are shown. A pair of each exists due to reflection symmetry in the flow direction. Example “fast” (red) and “slow” (blue) trajectories from the training data set are overlaid on the landscape. The fast trajectory avoids the saddle points and goes directly towards a stable minimum, whereas the slow trajectory gets trapped for long times near saddle 2, before finally escaping through its unstable manifold. For (B,E), the stable manifolds of the saddles closely align with $Z_2$, and hence are not visible due to minimization (marginalization). (G-I) Scatter plot together with predicted isotherms (solid lines) capturing typical fluctuations around a fully stretched state $Z_\text{stable,1}$. The insets are magnified views of the fluctuating trajectories around the stretched state over the energy landscape.
    }
    \label{fig:potential}
\end{figure}

\begin{figure}[H]
    \centerline{\includegraphics[width=\textwidth]{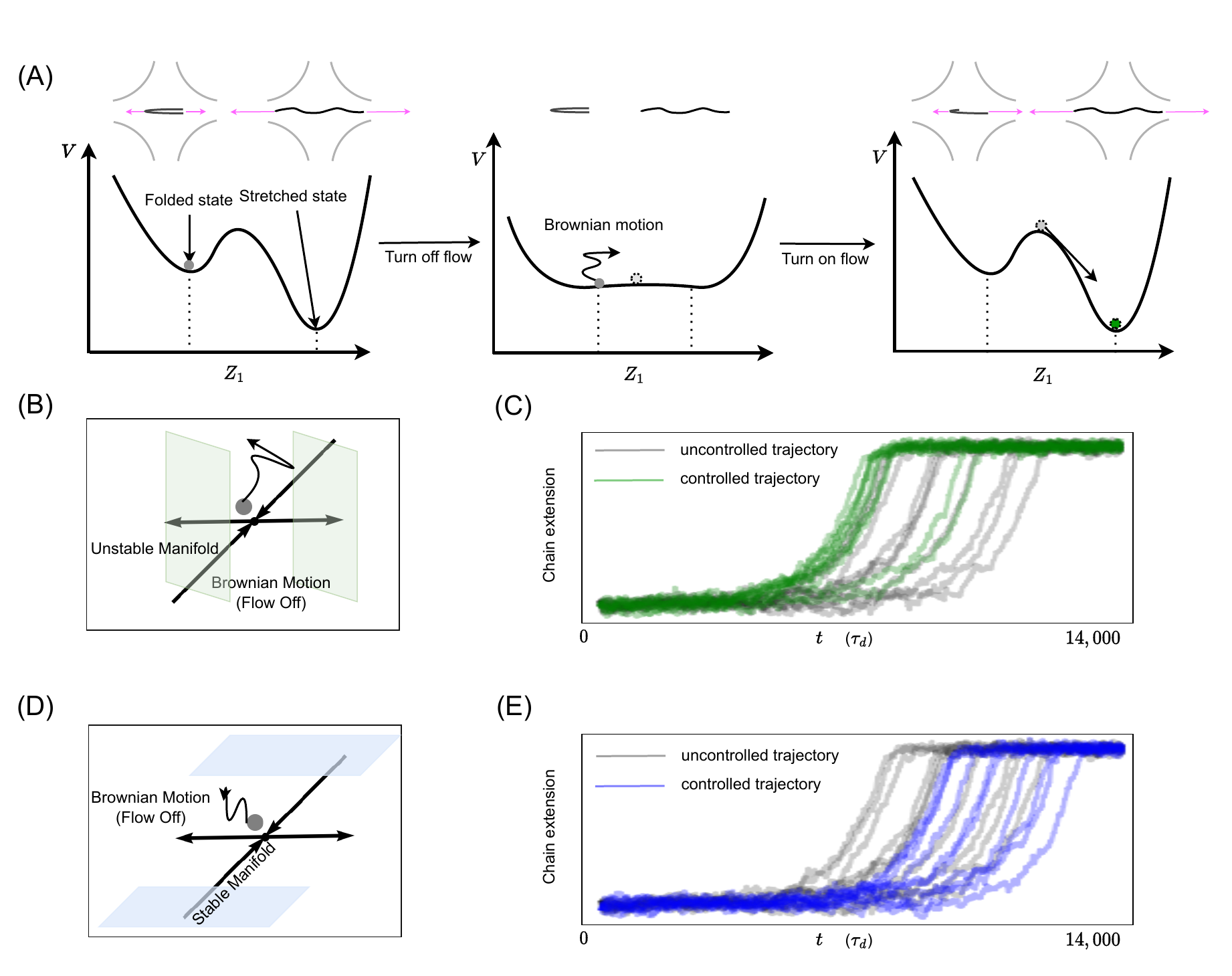}}
    \caption{
        \textbf{Data-driven control of the stretching dynamics.}
        (A) Control protocol to speed up unfolding. Projected along the extension direction, the polymer must overcome energy barriers to transition from the folded to the stretched state (left). In the control protocol, the flow is turned off if the reduced coordinate of the polymer is near the saddle point corresponding to the folded state, and the polymer drifts under the Brownian motion (middle). Then the flow is turned back on if the reduced coordinates of the polymer becomes sufficiently aligned (green shaded region in (B)) with the unstable manifold of the saddle or when the the equilibration time reaches $100\tau_d$ (right). Without any control, a folded state near the saddle point will unfold slowly (grey lines in (C)); with control, the chains unfold more rapidly (green lines in (C)). For the slowest 10 trajectories shown, their mean unfolding time was reduced by $14.14\%$. (D,E) Reversed control to slow down unfolding by turning on the flow when the reduced coordinates become aligned with the stable manifold instead (blue shaded region in (D)). The mean unfolding time increased by $14.96\%$ (blue lines in (E)).
    }
    \label{fig:design}
\end{figure}

\begin{figure}[H]
    \centerline{\includegraphics[width=\textwidth]{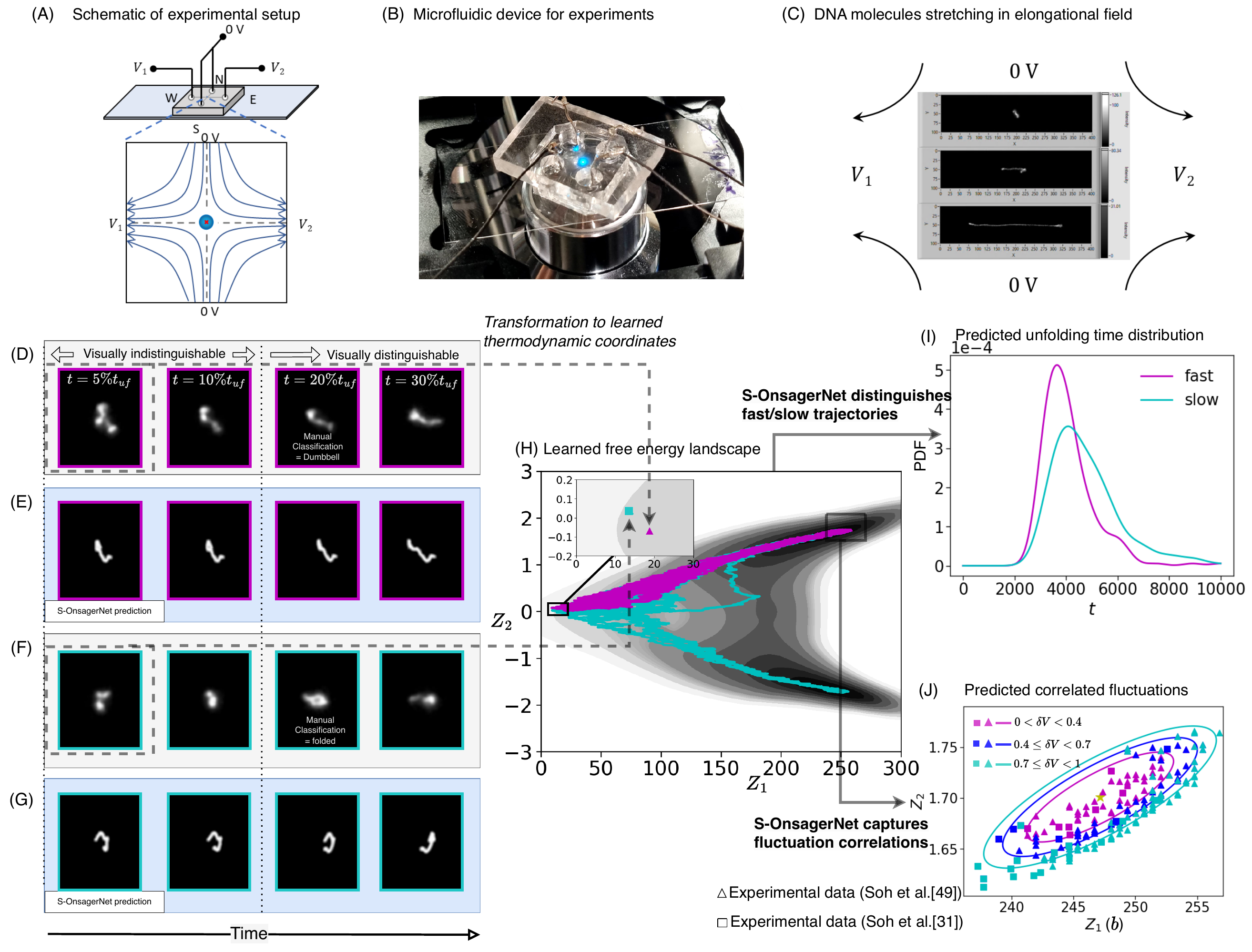}}
    \caption{
        \textbf{Analysis of experimental data.}
        (A-B) Schematic and photo of experimental setup, consisting of a microfluidic cross-slot device and platinum electrodes in the reservoirs. Via the electrodes, positive voltage levels $V_1$ and $V_2$ are applied to the West and East reservoirs (W and E, respectively), whilst the North and South reservoirs (N and S, respectively) are kept at 0V. During trapping, the negatively charged DNA will thus flow from N/S to W/E, until a molecule is trapped at the center of the channel - blue dot in the schematic. (C) Snapshots of a DNA molecule stretching. (D,F) Processed experimental images at various percentages of the unfolding time ($t_{\text{uf}}$). The selected trajectories have similar initial configurations and are visually indistinguishable in terms of unfolding dynamics. (H) Learned potential landscape and predicted slow and fast trajectories using S-OnsagerNet with the initial configurations at $t=5\% t_{uf}$ from (D,F). We note slight differences in the initial $Z_2$ values only (inset). (E,G) Reconstructed high dimensional configurations of selected simulated trajectories with similar low dimensional coordinates as the experimental configurations in (D,F). The S-OnsagerNet is capable of distinguishing between the manually classified dumbbell and folded states. (I) Predicted probability density functions of unfolding time with initial condition of fast and slow experimental trajectories at $t=5\% t_{uf}$ using S-OnsagerNet. (J) Fluctuations in $Z_1$ and $Z_2$ around the stable (stretched) state from experimental images. Data was obtained from ``Soh et al.~\cite{Soh2023}'' (triangle)
        and ``Soh et al.~\cite{soh2018knots}'' (square).
        Reduced coordinates are constructed according to the procedure outlined in data preparation in Methods.
        Colors indicate predicted energy
        levels according to the learned
        potential.
        We observe that
        fluctuations in $Z_1$ and $Z_2$ are highly correlated and agree well with that predicted by the effective equation of state.
    }
    \label{fig:experimental}
\end{figure}

\setcounter{figure}{0}
\makeatletter
\renewcommand{\figurename}{Extended Data Figure}
\makeatother

\begin{figure}[H]
\centerline{\includegraphics[width=0.9\textwidth]{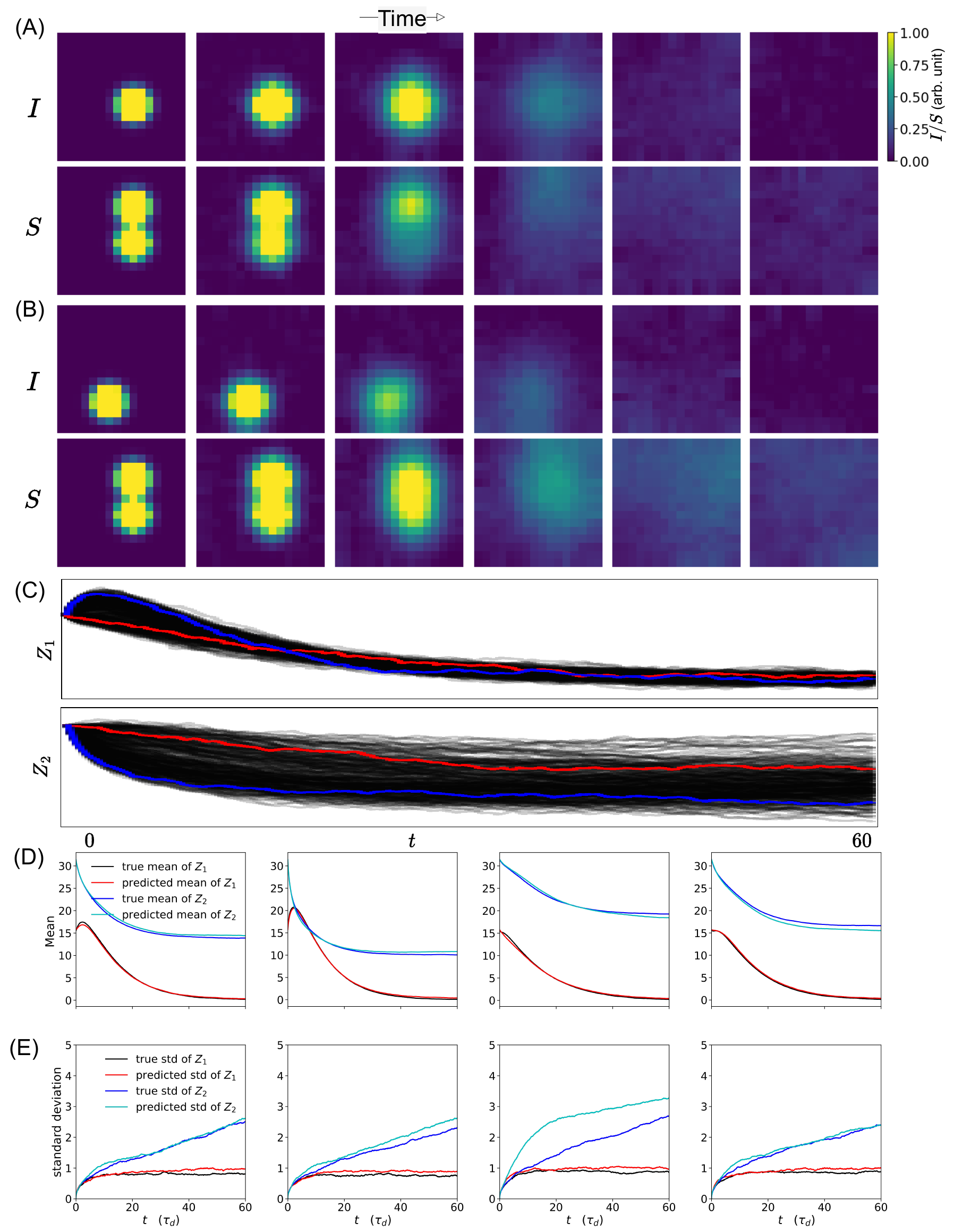}}
\caption{Data visualization and prediction results.
(A) and (B): Two trajectories depicting the spatial evolution of $I$
(infective) and $S$ (susceptible) with distinct initial conditions are plotted.
They have identical spatial averages initially but differing subsequent evolution.
In particular, in (A) the disease spreads ($Z_1$, the spatial average of $I$
increases initially) but in (B) the disease dies out monotonically.
(C) Scatter of $Z_1$ and $Z_2$ (spatial average of $S$) trajectories,
showing a high degree of variability despite identical initial values.
Note that there is variability in both the presence of disease spread
($Z_1$ increasing initially)
and the terminal value of $Z_2$,
corresponding to the remaining uninfected population
after the epidemic.
The blue (resp. red) trajectory corresponds to (A) (resp. (B)).
(D,E) True vs predicted statistics using S-OnsagerNet, showing good agreement.}
\label{fig:sir_dynamics_example}
\end{figure}

% \begin{figure}[p]
% \centerline{\includegraphics[width=\textwidth]{figures_short/SIR_tra.pdf}}
% \caption{Predicted trajectories of the SIR model.}
% \label{fig:SIR_tra}
% \end{figure}

\begin{figure}[H]
\centerline{\includegraphics[width=\textwidth]{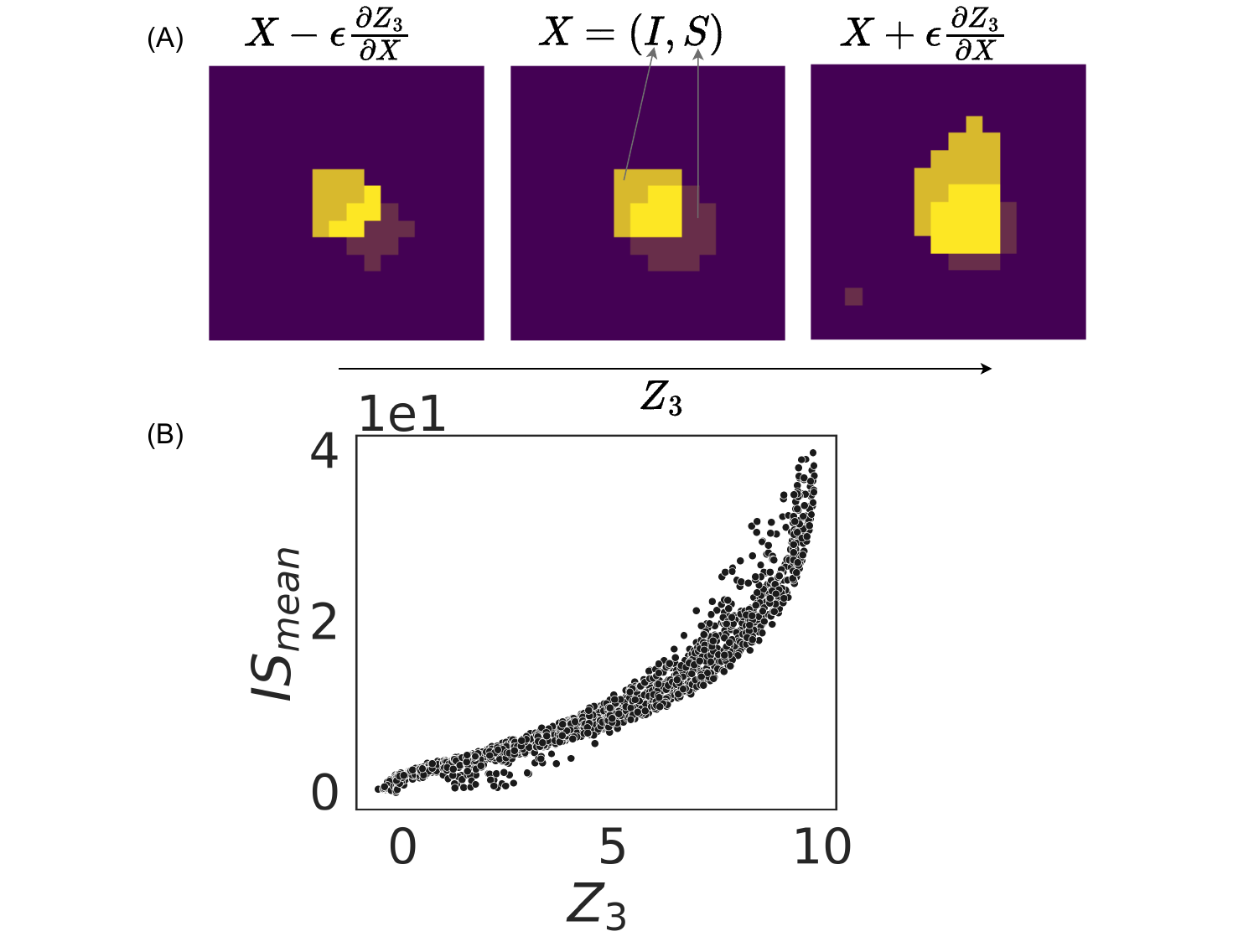}}
\caption{Physical interpretation of learned closure coordinate. (A) Perturbation of $X\pm \epsilon \partial Z_3 / \partial X$ from a given configuration $X=(I,S)$ with $\epsilon=2$.
The spatial configuration of I and S
are overlapping clusters,
in the form of Gaussians centered  at $(0,0)$ for $S$ and $(-1,1)$
for I.
To illustrate the overlap clearly,
we binarize the values as follows:
if the value of $I$ and $S$ is greater than 1,
% \QL{Why is this 1? Are we plotting densities or what?}
we truncate it to 1; otherwise,
if the value is less than 1, we truncate it to 0.
Observe from (A) that
increasing $Z_3$ corresponds to increasing
spatial overlap of the clusters, and vice versa.
We confirm this in (B), where we plot
$IS_{mean}$ (spatial overlap) vs $Z_3$,
showing a positive correlation.
}
\label{fig:SIR_meaning}
\end{figure}

\begin{figure}[H]
\centerline{\includegraphics[width=\textwidth]{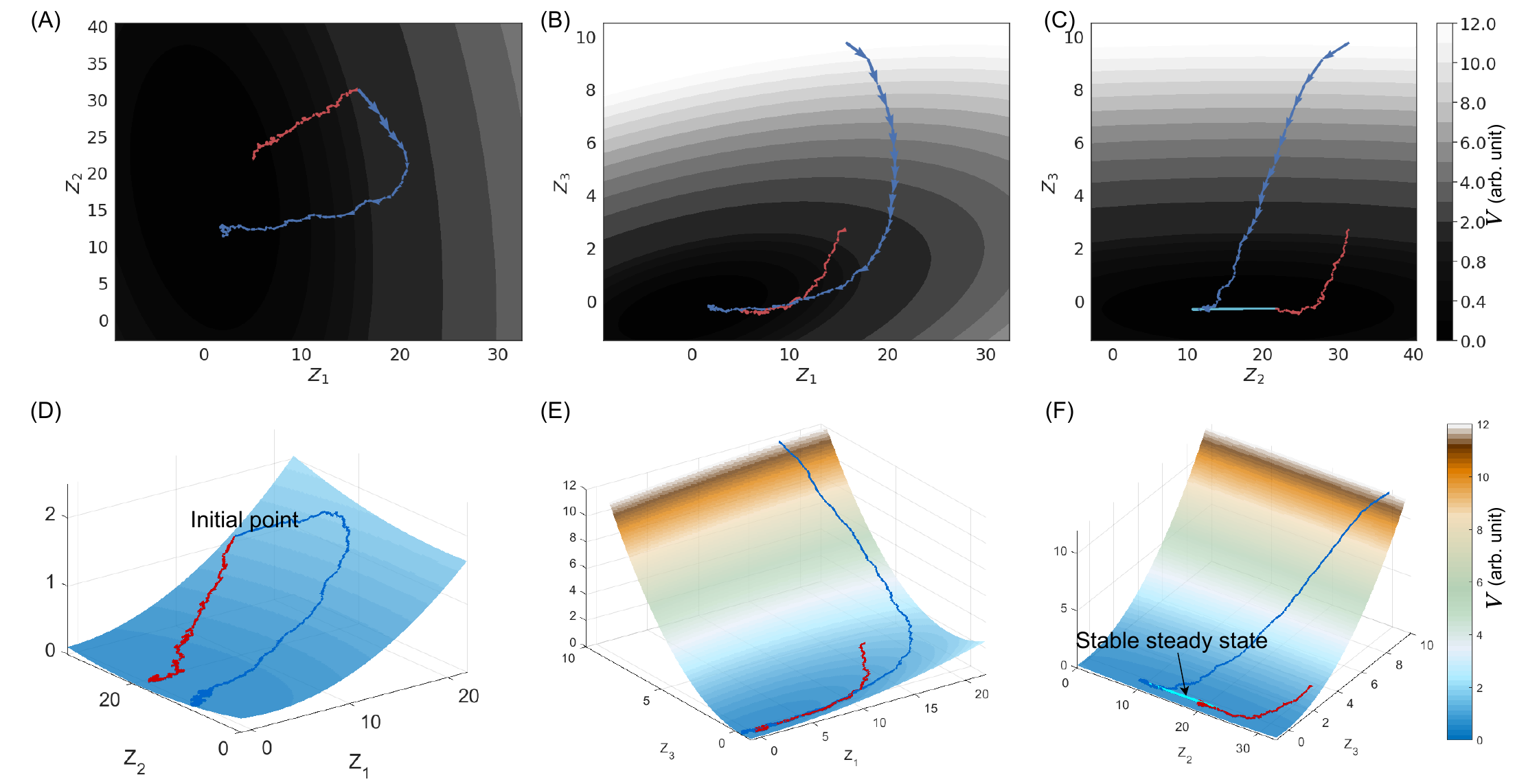}}
\caption{Potential landscape of the SIR model. Projected onto $Z_1-Z_2$ (A,D),
      $Z_1-Z_3$ (B,E) and $Z_2-Z_3$ (C,F) planes. Projection is computed via minimization (e.g. $V(Z_1,Z_2) = \min_{Z_3} V(Z_1, Z_2, Z_3)$). Example of disease spread (blue) and disease dying out (red) trajectories with same initial condition $Z_1$ and $Z_2$ from the training data set are shown.
      We observe from (B) that $Z_3$ determines the onset of disease spread
      and differentiates the two trajectories.
      Moreover, (C) shows that $Z_3$ also differentiates
      the final outcome of the epidemics,
      where the final $Z_2$ value depends on the initial $Z_3$
      value,
      and belongs to a 1D manifold of stable steady states
      as shown in (F).
      }
\label{fig:SIR_potential}
\end{figure}

% \bibliography{references}% common bib file




\appendix
\input{appendix.tex}

\end{document}

%% file: appendix.tex
\section{Supplementary figures}

\setcounter{figure}{0}
\makeatletter
\renewcommand{\figurename}{Supplementary Figure}
\makeatother

\begin{figure}[H]
\centerline{\includegraphics[width=\textwidth]{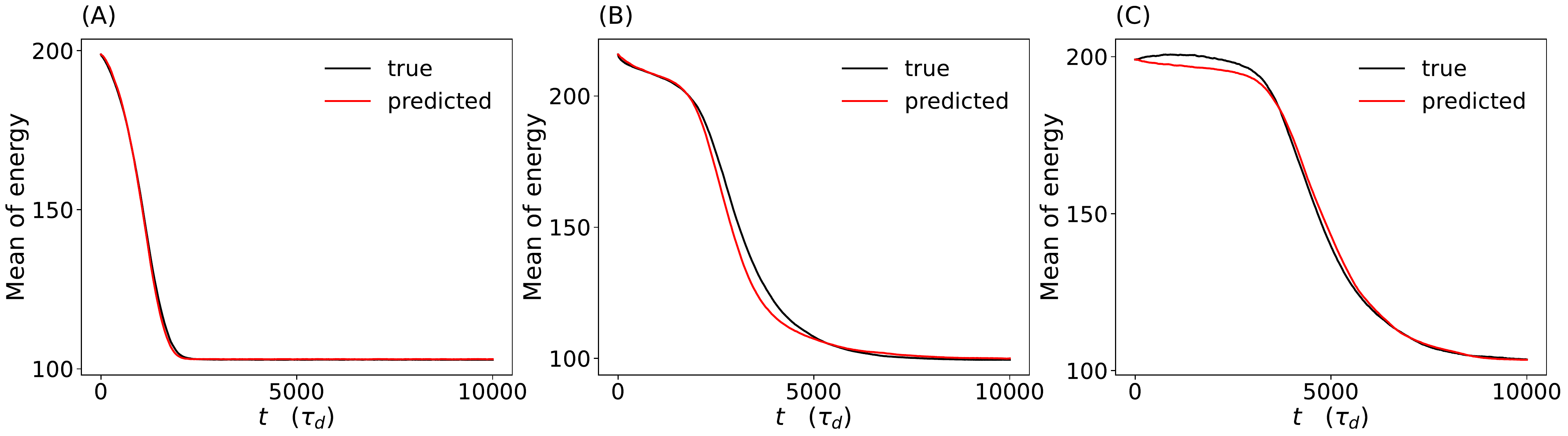}}
\caption{ \textbf{The time evolution of the mean of the energy $V$ for polymer dynamics.}
(A) fast trajectory; (B) medium trajectory; (C) slow trajectory.}
\label{mean_potential}
\end{figure}

\begin{figure}[H]
\centering
\includegraphics[width=\textwidth]{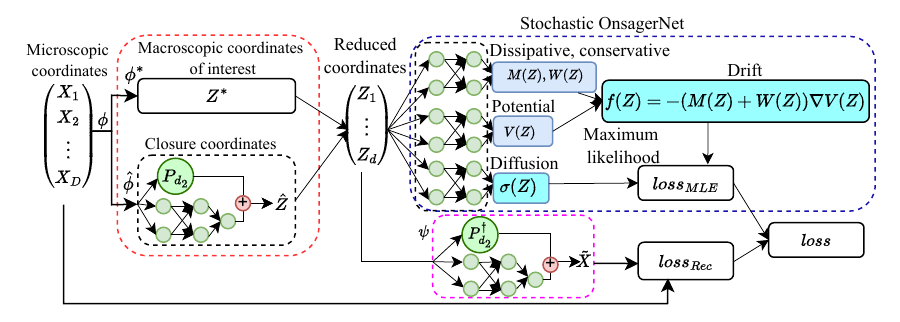}
\caption{ \textbf{Detailed S-OnsagerNet workflow.} The input data
    $X(t)=(X_1(t),\cdots,X_D(t))^T\in \mathbb{R}^D$ are the microscopic
    coordinates. The red box contains the components that discovers reduced
    coordinates $Z(t)=(Z^*(t),\hat{Z}(t))^T=(Z_1(t),\cdots,Z_d(t))^T\in
    \mathbb{R}^d$, where $\phi^*$ is known, and $\hat{\phi}$ is PCA-ResNet with
    $P_d$ the PCA projection matrix (to the first $d-1$ principal components).
    The blue box encloses the main S-OnsagerNet architecture to learn the low
    dimensional stochastic dynamical system. The function $\psi$ is a decoder
    neural network with output $\tilde{X}$, where $P^{\dagger}_d$ is the
    pseudo-inverse of $P_d$. The reconstruction error $\mathrm{loss}_{\mathrm{Rec}}$ and the
    maximum likelihood loss $\mathrm{loss}_{\mathrm{MLE}}$ are combined to obtain the total loss}
\label{NN}
\end{figure}

\begin{figure}[H]
    \centering
    \includegraphics[width=8cm]{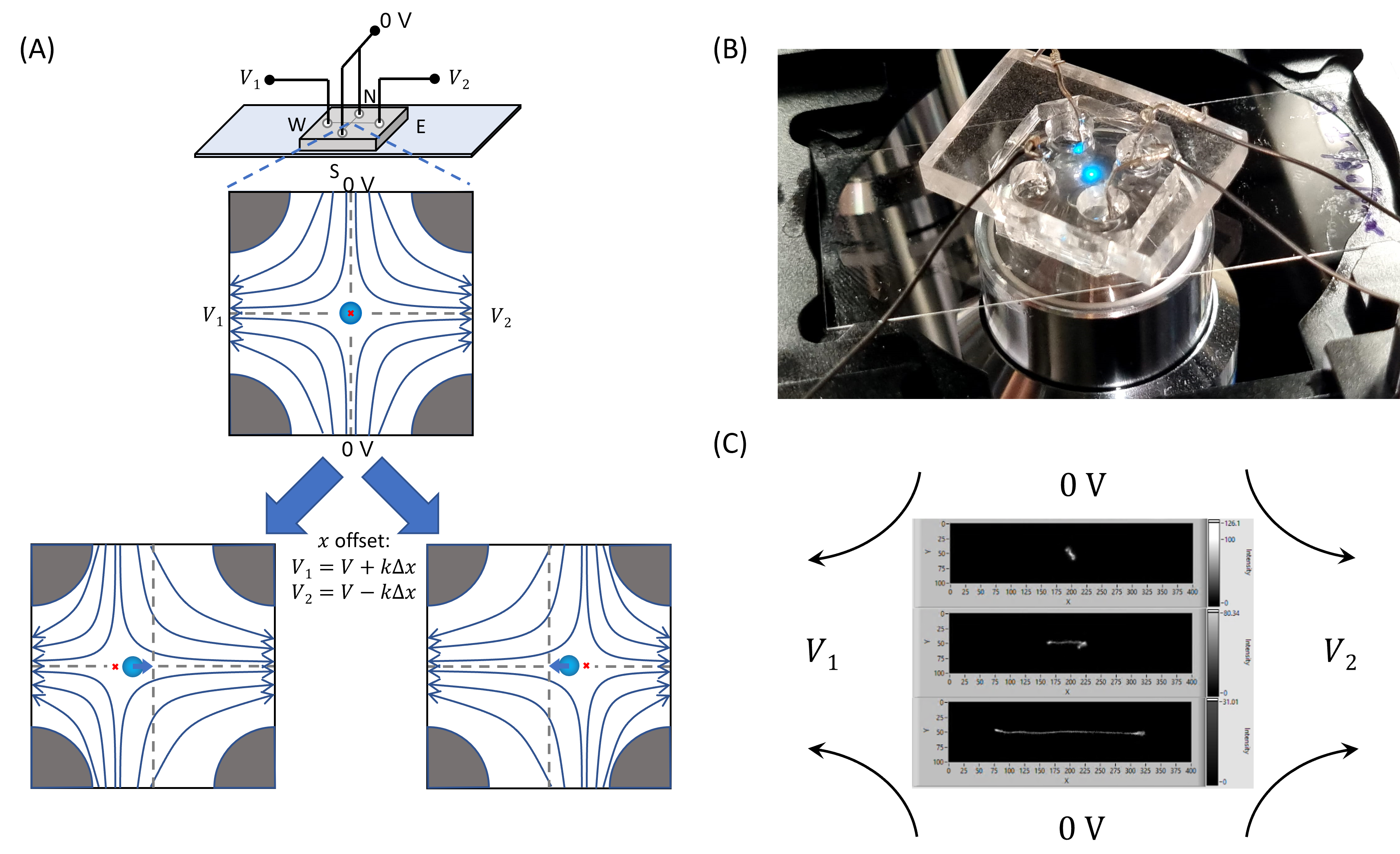}
  \caption{
      \textbf{Schematic of experimental setup.} (A) Top: Schematic of the experimental setup, consisting of a microfluidic cross-slot device and electrodes in the North, South, East and West reservoirs. $V_1$ and $V_2$ are computer-controlled voltages based on a feedback control system. Center: Negatively charged DNA molecules flow through the cross-slot channel according to the electric field lines. The blue circle represents an object at the saddle point. Bottom left, right: A proportional gain controller is used to trap and stretch a DNA molecule at the saddle point for long observation times in a planar elongational field. (B) Photo of microfluidic device sitting atop the microscope stage and arrangement of electrodes in the reservoirs. (C) Snapshots of a DNA molecule stretching under an elongational field.
  }
  \label{fig:experimentalsetup}
\end{figure}

\begin{figure}[H]
    \centering
    \includegraphics[width=8cm]{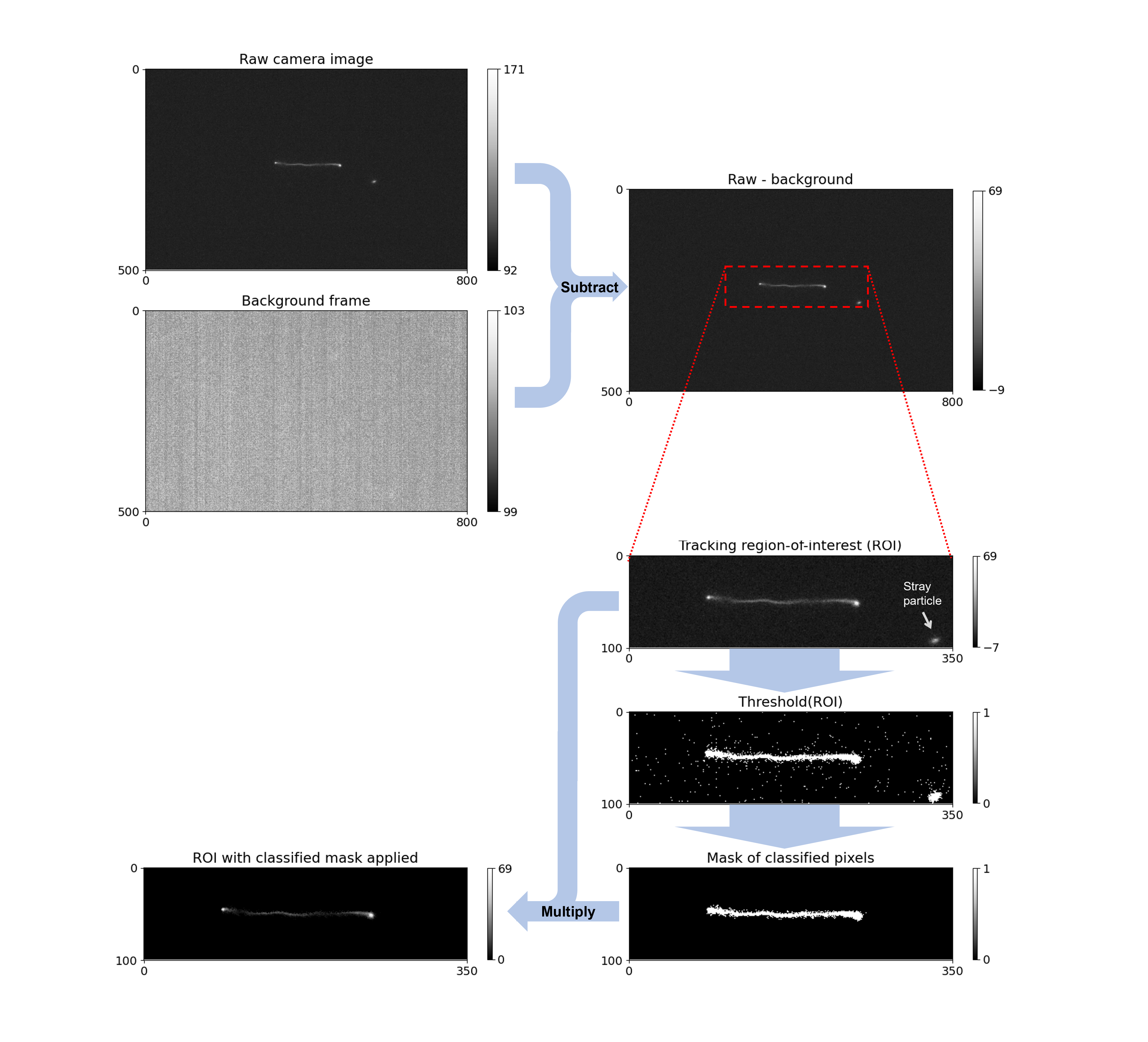}
  \caption{
      \textbf{Data filtering process for experimental images.} Image processing pipeline, showing the steps for obtaining a clean molecule image from a raw camera frame in real time. The clean image (bottom-left) is used to calculate molecule centroid coordinates and projection lengths in both axes.
  }
  \label{fig:filter}
\end{figure}

\begin{figure}[H]
    \centering
    \includegraphics[width=8cm]{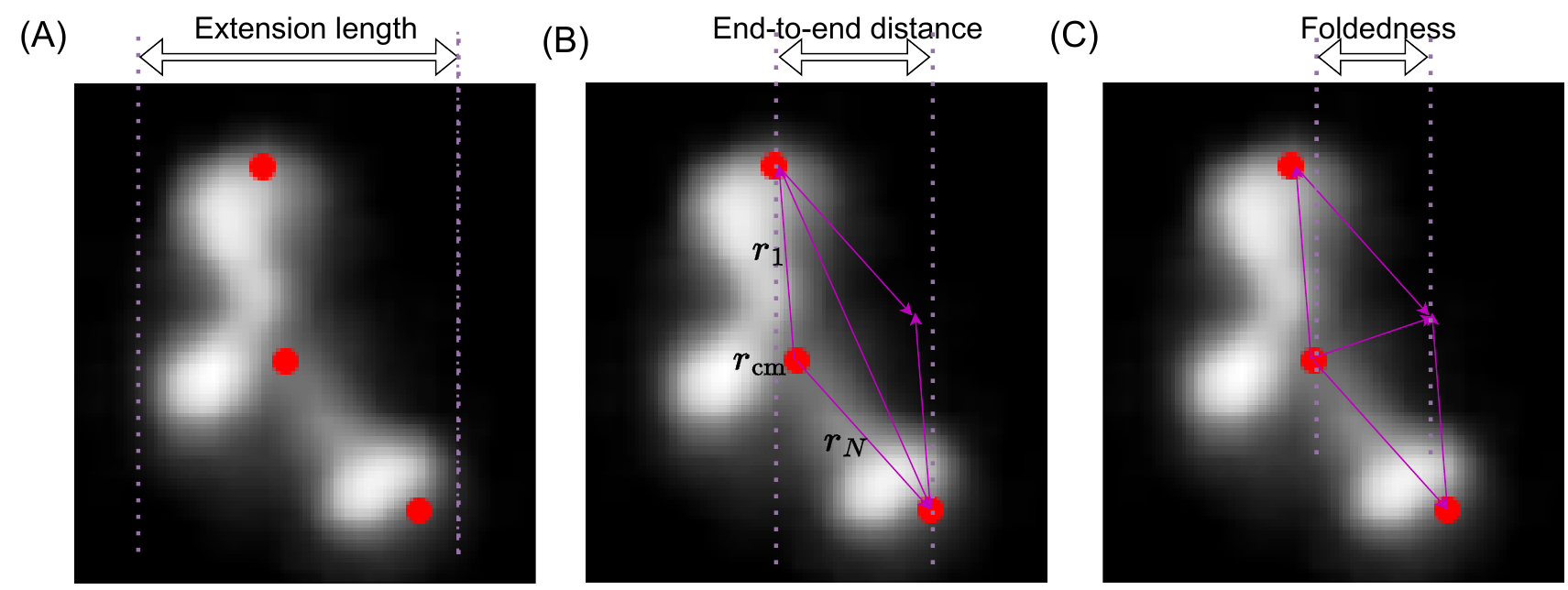}
  \caption{
      \textbf{Obtaining extension length (A), end-to-end distance (B) and
      foldedness (C) from experimental image with center of mass $r_{\text{cm}}$,
      and two end points.}
  }
  \label{fig:filter_low}
\end{figure}

\begin{figure}[H]
  \centering
 \includegraphics[width=\textwidth]{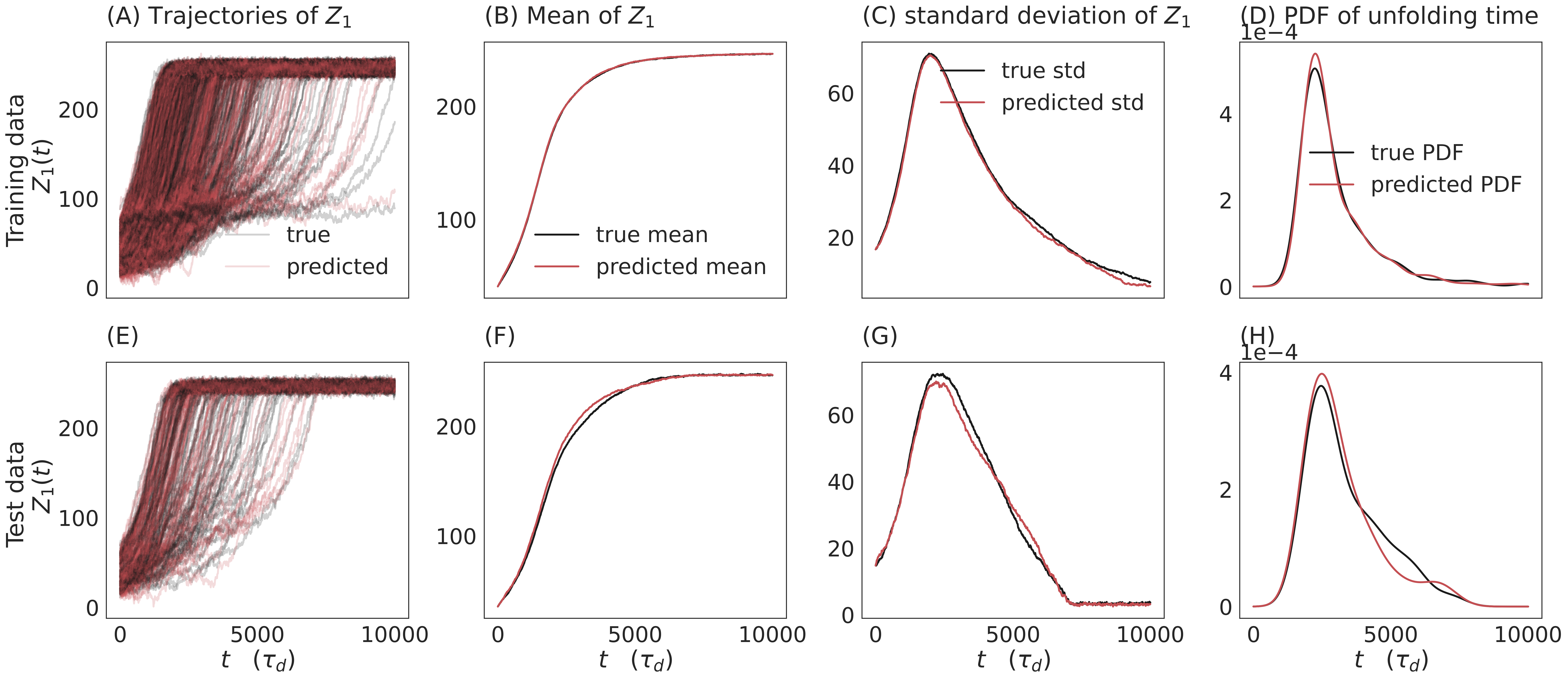}
    \caption{\textbf{Predicted stretching trajectories and statistics for 610 training trajectories (up) and 110 test trajectories (down): true data (black) and model
      prediction (red).} (A,E) Individual stretching trajectories of polymer chains from the different initial condition. (B,F) Mean and (C,G) standard
      deviations of polymer chain extensions.
      (D,H) Probability density function (PDF) of the chain unfolding times.}
    \label{fig:train_test}
\end{figure}

\begin{figure}[H]
    \centering
 \centerline{\includegraphics[width=\textwidth]{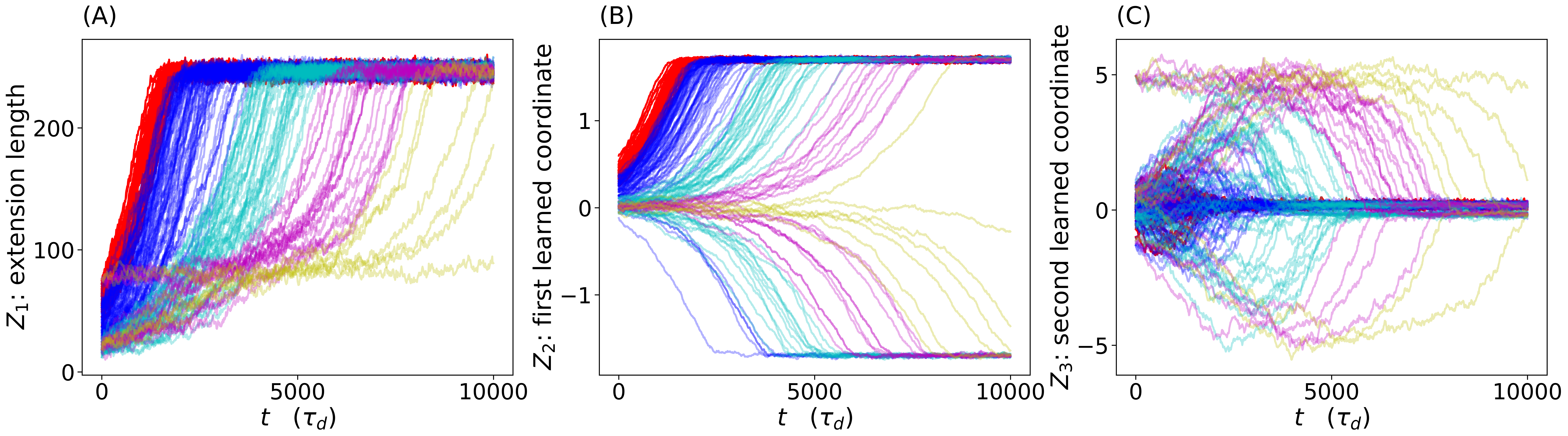}}
    \caption{\textbf{Learned reduced coordinates.} Evolution of (A) chain extension ($Z_1$), (B) first learned coordinate ($Z_2$, indicator of end-to-end distance) and (C) second learned coordinate ($Z_3$, indicator of foldedness) with time. The trajectories are colored by the chain unfolding times, red: $t_\text{unfold}<2000$; blue: $2000\leq t_\text{unfold}<4000$; cyan: $4000\leq t_\text{unfold}<6000$ ; magenta: $6000\leq t_\text{unfold}<8000$ ; yellow: $ t_\text{unfold}\geq 8000$.}
    \label{fig:learn_low_coordinates}
\end{figure}

\begin{figure}[H]
    \centerline{\includegraphics[width=\textwidth]{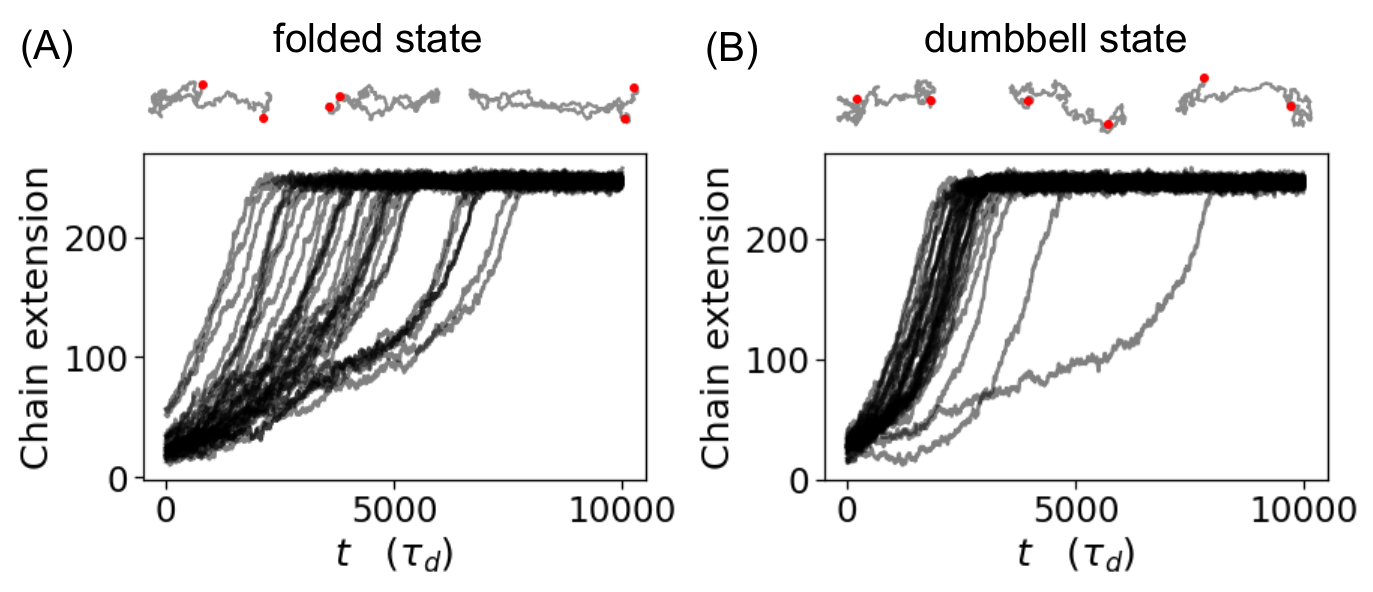}}
       \caption{\textbf{Stretching trajectories for polymer chains in the (A)
       folded and (B) dumbbell states.} The chain configurations were identified
       as described by Perkins et al. There is a large range of unfolding times
       within a given configuration type, hence the classification of chain
       configuration is insufficient for prediction purposes.}
       \label{fig:state_cate}
\end{figure}

\begin{figure}[H]
    \centerline{\includegraphics[width=\textwidth]{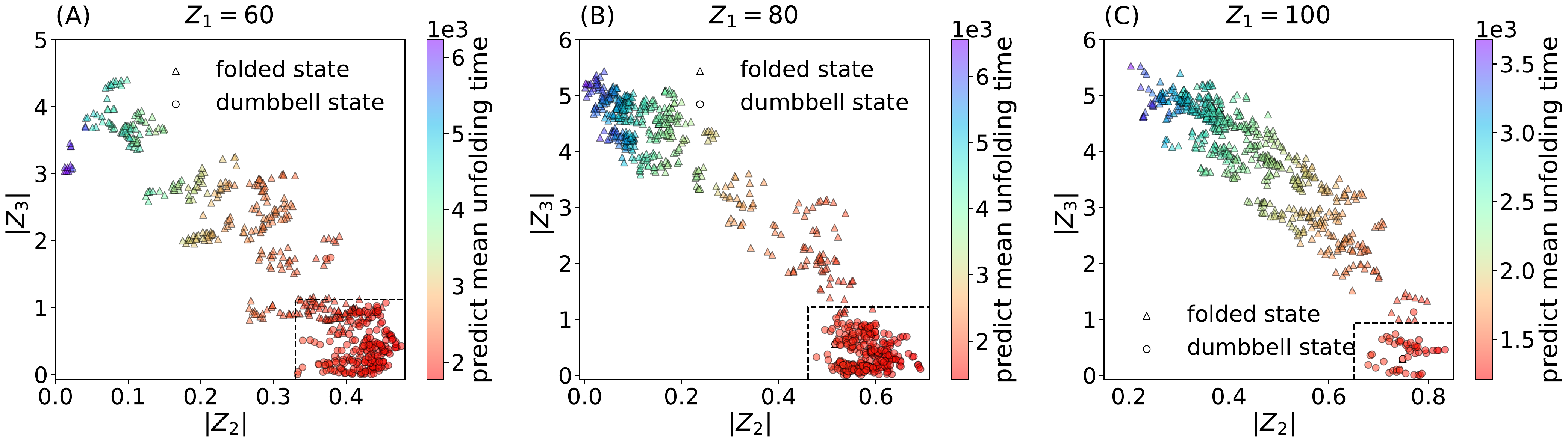}}
    \caption{
        \textbf{Consistency with configuration categorization scheme in current literature.} Plot of $\lvert  Z_3\rvert $ as a function of $\lvert  Z_2\rvert $ for folded and dumbbell configurations at (A) $Z_1=60$, (B) $Z_1=80$ and (C) $Z_1=100$. The markers are colored by the predicted chain unfolding times. The boxed region with high $\lvert  Z_2\rvert $ and low $\lvert  Z_3\rvert $ values encompasses a mix of folded and dumbbell chains with similar unfolding times, indicating that the broad categorization scheme is unable to provide accurate predictions.}
    \label{dumbbell_folded}
\end{figure}

\begin{figure}[H]
    \centering
    \includegraphics[width=8cm]{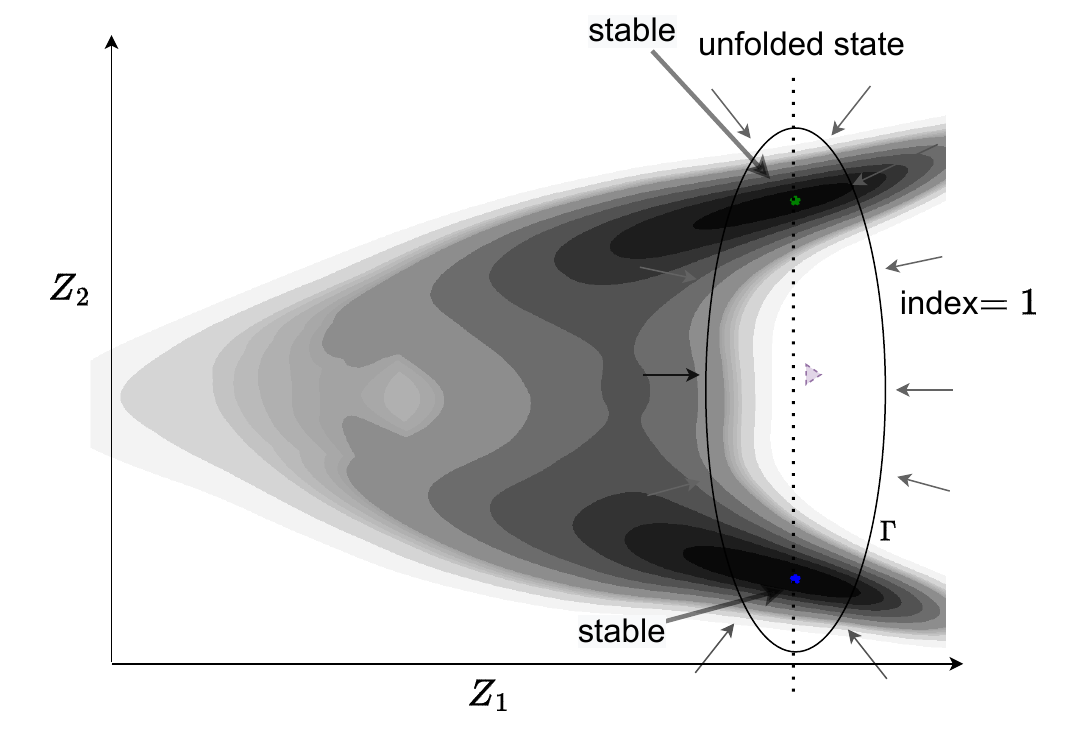}
    \caption{
        \textbf{Illustration of the limitation of a 2-dimensional
        potential landscape.}
        Due to the stability of the unfolded state,
        the vector fields around the curve $\Gamma$ point
        inwards towards the interior, so the index of $\Gamma$
        is $+1$.
        However, this contradicts the presence of two
        stable critical points inside $\Gamma$,
        which implies that its index is $+2$.
    }
    \label{fig:index}
\end{figure}

\begin{figure}[H]
    \centering
    \includegraphics[width=1\linewidth]{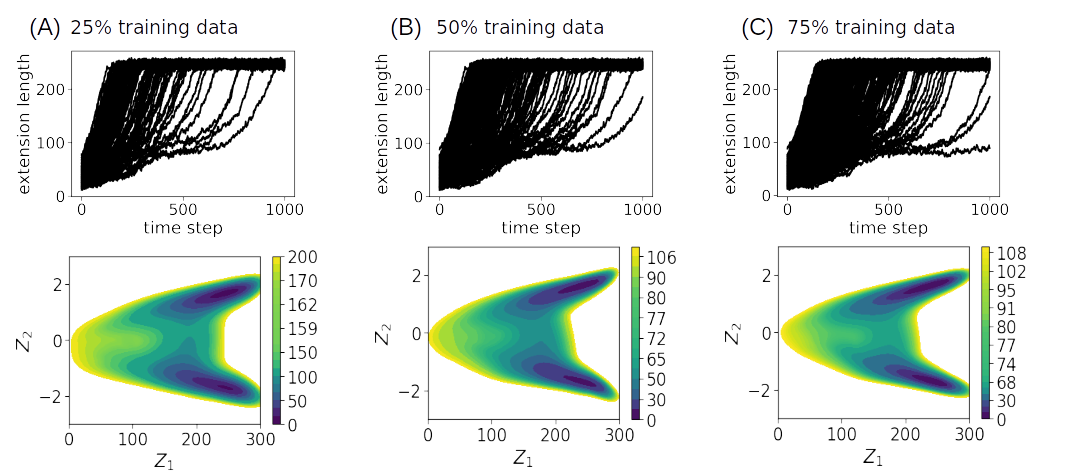}
    \caption{
        \textbf{Trajectories and potential landscapes for different percentages of training data.} The full dataset, as used in the rest of this work, contains 610 trajectories. In order to evaluate the impact of dataset size on prediction results, the S-OnsagerNet was trained using (A) 25\%, (B) 50\% and (C) 75\% of the 610 trajectories. The datasets in (A) and (C) (top) do not have shared trajectories, whilst the dataset (B) (top) contains data from both the 25\% and 75\% datasets. The potential landscapes (bottom) resulting from training with the different number of trajectories are plotted for $Z_1$ vs $Z_2$. It can be observed that each contains the characteristic features of the potential landscape discussed for the full dataset, with two areas of near-zero potential when the DNA reaches the steady-state, and the emergence of a saddle point around $Z_1 \approx 110$.
        %However, in this study, it appears that the emergence of the saddle point does not correlate with the number of trajectories, but may depend on the specific trajectories chosen instead. This suggests that a more in-depth study regarding the impact of data diversity may be needed.
    }
    \label{fig:Potential_landscape}
\end{figure}

\newpage

\section{Supplementary tables}

\makeatletter
\renewcommand{\tablename}{Supplementary Table}
\makeatother

\begin{table}[h]
\begin{center}
\begin{minipage}{\textwidth}
\caption{Variation in prediction error with number of reduced dimensions. Note that the dimension here includes the chosen macroscopic coordinate corresponding to polymer extension length.}\label{tab2}%Error of training and test data with different dimensions}\label{tab2}
\begin{tabular*}{\textwidth}{@{\extracolsep{\fill}}lcccc@{\extracolsep{\fill}}}
\toprule%{cccc}%{\textwidth}{@{\extracolsep{\fill}}lccc@{\extracolsep{\fill}}}
\multicolumn{5}{c}{relative  $L^2$ error of training/test data (\%)}\\[1ex]
\hline
Dimension & &mean&standard derivation&PDF of unfolding time\\[1ex]
\hline
2D & &0.7927/2.086&13.40/32.96&12.63/17.43\\[1ex]
\hline
3D & &0.2828/1.861&3.147/6.057&5.717/12.42\\[1ex]
\hline
4D & & 0.4363/1.363&9.329/30.62&3.041/10.52\\[1ex]
%\botrule
\bottomrule
\end{tabular*}\label{tab:dimension}
\end{minipage}
\end{center}
\end{table}

\begin{table}[h]
\begin{center}
\begin{minipage}{\textwidth}
\caption{Variation in test prediction error of 3D model with the number of trajectories in the training data. We group the test trajectories into three categories (fast, medium, slow) according to their rate of stretching.}%\label{tab3}%Error of training and test data with different number of training trajectories}\label{tab3}
\begin{tabular*}{\textwidth}{@{\extracolsep{\fill}}lcccc@{\extracolsep{\fill}}}
\toprule%{cccc}%{\textwidth}{@{\extracolsep{\fill}}lccc@{\extracolsep{\fill}}}
\multicolumn{5}{c}{relative test $L^2$ error of mean/ standard derivation /pdf of unfolding time (\%)}\\[1ex]
\hline
Number of trajectories & &fast&medium&slow\\[1ex]
\hline
153 (25\%) & &0.246/11.38/26.37& 6.749/27.51/44.80&9.666/55.33/33.81\\[1ex]
\hline
305 (50\%) & &2.504/23.52/72.49&2.967/12.03/19.63&3.404/8.438/14.06\\[1ex]
\hline
457 (75\%) & & 1.206/16.91/26.34&6.581/35.28/34.80&2.120/8.866/15.56\\[1ex]
\hline
610 (100 \%) & &0.388 /13.06/15.95& 2.101/9.227/20.18&2.027/7.121/13.42\\[1ex]
%\botrule
\bottomrule
\end{tabular*}\label{tab:percentage}
\end{minipage}
\end{center}
\end{table}

%% file: sn_main_arxiv.bbl
\begin{thebibliography}{52}
    \ifx \bisbn   \undefined \def \bisbn  #1{ISBN #1}\fi
    \ifx \binits  \undefined \def \binits#1{#1}\fi
    \ifx \bauthor  \undefined \def \bauthor#1{#1}\fi
    \ifx \batitle  \undefined \def \batitle#1{#1}\fi
    \ifx \bjtitle  \undefined \def \bjtitle#1{#1}\fi
    \ifx \bvolume  \undefined \def \bvolume#1{\textbf{#1}}\fi
    \ifx \byear  \undefined \def \byear#1{#1}\fi
    \ifx \bissue  \undefined \def \bissue#1{#1}\fi
    \ifx \bfpage  \undefined \def \bfpage#1{#1}\fi
    \ifx \blpage  \undefined \def \blpage #1{#1}\fi
    \ifx \burl  \undefined \def \burl#1{\textsf{#1}}\fi
    \ifx \doiurl  \undefined \def \doiurl#1{\url{https://doi.org/#1}}\fi
    \ifx \betal  \undefined \def \betal{\textit{et al.}}\fi
    \ifx \binstitute  \undefined \def \binstitute#1{#1}\fi
    \ifx \binstitutionaled  \undefined \def \binstitutionaled#1{#1}\fi
    \ifx \bctitle  \undefined \def \bctitle#1{#1}\fi
    \ifx \beditor  \undefined \def \beditor#1{#1}\fi
    \ifx \bpublisher  \undefined \def \bpublisher#1{#1}\fi
    \ifx \bbtitle  \undefined \def \bbtitle#1{#1}\fi
    \ifx \bedition  \undefined \def \bedition#1{#1}\fi
    \ifx \bseriesno  \undefined \def \bseriesno#1{#1}\fi
    \ifx \blocation  \undefined \def \blocation#1{#1}\fi
    \ifx \bsertitle  \undefined \def \bsertitle#1{#1}\fi
    \ifx \bsnm \undefined \def \bsnm#1{#1}\fi
    \ifx \bsuffix \undefined \def \bsuffix#1{#1}\fi
    \ifx \bparticle \undefined \def \bparticle#1{#1}\fi
    \ifx \barticle \undefined \def \barticle#1{#1}\fi
    \bibcommenthead
    \ifx \bconfdate \undefined \def \bconfdate #1{#1}\fi
    \ifx \botherref \undefined \def \botherref #1{#1}\fi
    \ifx \url \undefined \def \url#1{\textsf{#1}}\fi
    \ifx \bchapter \undefined \def \bchapter#1{#1}\fi
    \ifx \bbook \undefined \def \bbook#1{#1}\fi
    \ifx \bcomment \undefined \def \bcomment#1{#1}\fi
    \ifx \oauthor \undefined \def \oauthor#1{#1}\fi
    \ifx \citeauthoryear \undefined \def \citeauthoryear#1{#1}\fi
    \ifx \endbibitem  \undefined \def \endbibitem {}\fi
    \ifx \bconflocation  \undefined \def \bconflocation#1{#1}\fi
    \ifx \arxivurl  \undefined \def \arxivurl#1{\textsf{#1}}\fi
    \csname PreBibitemsHook\endcsname

    %%% 1
    \bibitem{noh2020machine}
    \begin{barticle}
    \bauthor{\bsnm{Noh}, \binits{J.}},
    \bauthor{\bsnm{Gu}, \binits{G.H.}},
    \bauthor{\bsnm{Kim}, \binits{S.}},
    \bauthor{\bsnm{Jung}, \binits{Y.}}:
    \batitle{Machine-enabled inverse design of inorganic solid materials: promises
      and challenges}.
    \bjtitle{Chemical Science}
    \bvolume{11}(\bissue{19}),
    \bfpage{4871}--\blpage{4881}
    (\byear{2020})
    \end{barticle}
    \endbibitem

    %%% 2
    \bibitem{hippalgaonkar2023knowledge}
    \begin{botherref}
    \oauthor{\bsnm{Hippalgaonkar}, \binits{K.}},
    \oauthor{\bsnm{Li}, \binits{Q.}},
    \oauthor{\bsnm{Wang}, \binits{X.}},
    \oauthor{\bsnm{Fisher~III}, \binits{J.W.}},
    \oauthor{\bsnm{Kirkpatrick}, \binits{J.}},
    \oauthor{\bsnm{Buonassisi}, \binits{T.}}:
    Knowledge-integrated machine learning for materials: lessons from gameplaying
      and robotics.
    Nature Reviews Materials,
    1--20
    (2023)
    \end{botherref}
    \endbibitem

    %%% 3
    \bibitem{brunton2016discovering}
    \begin{barticle}
    \bauthor{\bsnm{Brunton}, \binits{S.L.}},
    \bauthor{\bsnm{Proctor}, \binits{J.L.}},
    \bauthor{\bsnm{Kutz}, \binits{J.N.}}:
    \batitle{Discovering governing equations from data by sparse identification of
      nonlinear dynamical systems}.
    \bjtitle{Proceedings of the National Academy of Sciences}
    \bvolume{113}(\bissue{15}),
    \bfpage{3932}--\blpage{3937}
    (\byear{2016})
    \end{barticle}
    \endbibitem

    %%% 4
    \bibitem{hamzi2021learning}
    \begin{barticle}
    \bauthor{\bsnm{Hamzi}, \binits{B.}},
    \bauthor{\bsnm{Owhadi}, \binits{H.}}:
    \batitle{Learning dynamical systems from data: a simple cross-validation
      perspective, part i: parametric kernel flows}.
    \bjtitle{Physica D: Nonlinear Phenomena}
    \bvolume{421},
    \bfpage{132817}
    (\byear{2021})
    \end{barticle}
    \endbibitem

    %%% 5
    \bibitem{dietrich2021learning}
    \begin{botherref}
    \oauthor{\bsnm{Dietrich}, \binits{F.}},
    \oauthor{\bsnm{Makeev}, \binits{A.}},
    \oauthor{\bsnm{Kevrekidis}, \binits{G.}},
    \oauthor{\bsnm{Evangelou}, \binits{N.}},
    \oauthor{\bsnm{Bertalan}, \binits{T.}},
    \oauthor{\bsnm{Reich}, \binits{S.}},
    \oauthor{\bsnm{Kevrekidis}, \binits{I.G.}}:
    Learning effective stochastic differential equations from microscopic
      simulations: combining stochastic numerics and deep learning.
    arXiv preprint arXiv:2106.09004
    (2021)
    \end{botherref}
    \endbibitem

    %%% 6
    \bibitem{raissi2019physics}
    \begin{barticle}
    \bauthor{\bsnm{Raissi}, \binits{M.}},
    \bauthor{\bsnm{Perdikaris}, \binits{P.}},
    \bauthor{\bsnm{Karniadakis}, \binits{G.E.}}:
    \batitle{Physics-informed neural networks: A deep learning framework for
      solving forward and inverse problems involving nonlinear partial differential
      equations}.
    \bjtitle{Journal of Computational physics}
    \bvolume{378},
    \bfpage{686}--\blpage{707}
    (\byear{2019})
    \end{barticle}
    \endbibitem

    %%% 7
    \bibitem{yu2020onsagernet}
    \begin{barticle}
    \bauthor{\bsnm{Yu}, \binits{H.}},
    \bauthor{\bsnm{Tian}, \binits{X.}},
    \bauthor{\bsnm{E}, \binits{W.}},
    \bauthor{\bsnm{Li}, \binits{Q.}}:
    \batitle{Onsagernet: Learning stable and interpretable dynamics using a
      generalized {Onsager} principle}.
    \bjtitle{Physical Review Fluids}
    \bvolume{6}(\bissue{11}),
    \bfpage{114402}
    (\byear{2021})
    \end{barticle}
    \endbibitem

    %%% 8
    \bibitem{hesthaven2022rank}
    \begin{barticle}
    \bauthor{\bsnm{Hesthaven}, \binits{J.S.}},
    \bauthor{\bsnm{Pagliantini}, \binits{C.}},
    \bauthor{\bsnm{Ripamonti}, \binits{N.}}:
    \batitle{Rank-adaptive structure-preserving model order reduction of
      hamiltonian systems}.
    \bjtitle{ESAIM: Mathematical Modelling and Numerical Analysis}
    \bvolume{56}(\bissue{2}),
    \bfpage{617}--\blpage{650}
    (\byear{2022})
    \end{barticle}
    \endbibitem

    %%% 9
    \bibitem{valperga2022learning}
    \begin{bchapter}
    \bauthor{\bsnm{Valperga}, \binits{R.}},
    \bauthor{\bsnm{Webster}, \binits{K.}},
    \bauthor{\bsnm{Turaev}, \binits{D.}},
    \bauthor{\bsnm{Klein}, \binits{V.}},
    \bauthor{\bsnm{Lamb}, \binits{J.}}:
    \bctitle{Learning reversible symplectic dynamics}.
    In: \bbtitle{Learning for Dynamics and Control Conference},
    pp. \bfpage{906}--\blpage{916}
    (\byear{2022}).
    \bcomment{PMLR}
    \end{bchapter}
    \endbibitem

    %%% 10
    \bibitem{jin2022learning}
    \begin{botherref}
    \oauthor{\bsnm{Jin}, \binits{P.}},
    \oauthor{\bsnm{Zhang}, \binits{Z.}},
    \oauthor{\bsnm{Kevrekidis}, \binits{I.G.}},
    \oauthor{\bsnm{Karniadakis}, \binits{G.E.}}:
    Learning poisson systems and trajectories of autonomous systems via poisson
      neural networks.
    IEEE Transactions on Neural Networks and Learning Systems
    (2022)
    \end{botherref}
    \endbibitem

    %%% 11
    \bibitem{lin2022data}
    \begin{bchapter}
    \bauthor{\bsnm{Lin}, \binits{B.}},
    \bauthor{\bsnm{Li}, \binits{Q.}},
    \bauthor{\bsnm{Ren}, \binits{W.}}:
    \bctitle{A data driven method for computing quasipotentials}.
    In: \bbtitle{Mathematical and Scientific Machine Learning},
    pp. \bfpage{652}--\blpage{670}
    (\byear{2022}).
    \bcomment{PMLR}
    \end{bchapter}
    \endbibitem

    %%% 12
    \bibitem{onsager1931re1}
    \begin{barticle}
    \bauthor{\bsnm{Onsager}, \binits{L.}}:
    \batitle{Reciprocal relations in irreversible processes. i.}
    \bjtitle{Physical review}
    \bvolume{37}(\bissue{4}),
    \bfpage{405}
    (\byear{1931})
    \end{barticle}
    \endbibitem

    %%% 13
    \bibitem{onsager1931re2}
    \begin{barticle}
    \bauthor{\bsnm{Onsager}, \binits{L.}}:
    \batitle{Reciprocal relations in irreversible processes. ii.}
    \bjtitle{Physical review}
    \bvolume{38}(\bissue{12}),
    \bfpage{2265}
    (\byear{1931})
    \end{barticle}
    \endbibitem

    %%% 14
    \bibitem{chen2022.AutomatedDiscoveryFundamental}
    \begin{barticle}
    \bauthor{\bsnm{Chen}, \binits{B.}},
    \bauthor{\bsnm{Huang}, \binits{K.}},
    \bauthor{\bsnm{Raghupathi}, \binits{S.}},
    \bauthor{\bsnm{Chandratreya}, \binits{I.}},
    \bauthor{\bsnm{Du}, \binits{Q.}},
    \bauthor{\bsnm{Lipson}, \binits{H.}}:
    \batitle{Automated discovery of fundamental variables hidden in experimental
      data}.
    \bjtitle{Nature Computational Science}
    \bvolume{2}(\bissue{7}),
    \bfpage{433}--\blpage{442}
    (\byear{2022}).
    \doiurl{10.1038/s43588-022-00281-6}
    \end{barticle}
    \endbibitem

    %%% 15
    \bibitem{noe2020machine}
    \begin{barticle}
    \bauthor{\bsnm{No{\'e}}, \binits{F.}},
    \bauthor{\bsnm{Tkatchenko}, \binits{A.}},
    \bauthor{\bsnm{M{\"u}ller}, \binits{K.-R.}},
    \bauthor{\bsnm{Clementi}, \binits{C.}}:
    \batitle{Machine learning for molecular simulation}.
    \bjtitle{Annual review of physical chemistry}
    \bvolume{71},
    \bfpage{361}--\blpage{390}
    (\byear{2020})
    \end{barticle}
    \endbibitem

    %%% 16
    \bibitem{doi2011onsager}
    \begin{barticle}
    \bauthor{\bsnm{Doi}, \binits{M.}}:
    \batitle{{Onsager}’s variational principle in soft matter}.
    \bjtitle{Journal of Physics: Condensed Matter}
    \bvolume{23}(\bissue{28}),
    \bfpage{284118}
    (\byear{2011})
    \end{barticle}
    \endbibitem

    %%% 17
    \bibitem{doi2015onsager}
    \begin{barticle}
    \bauthor{\bsnm{Doi}, \binits{M.}}:
    \batitle{{Onsager} principle as a tool for approximation}.
    \bjtitle{Chinese Physics B}
    \bvolume{24}(\bissue{2}),
    \bfpage{020505}
    (\byear{2015})
    \end{barticle}
    \endbibitem

    %%% 18
    \bibitem{mielke2014.RelationGradientFlows}
    \begin{barticle}
    \bauthor{\bsnm{Mielke}, \binits{A.}},
    \bauthor{\bsnm{Peletier}, \binits{M.A.}},
    \bauthor{\bsnm{Renger}, \binits{D.R.M.}}:
    \batitle{On the {{Relation}} between {{Gradient Flows}} and the
      {{Large-Deviation Principle}}, with {{Applications}} to {{Markov Chains}} and
      {{Diffusion}}}.
    \bjtitle{Potential Analysis}
    \bvolume{41}(\bissue{4}),
    \bfpage{1293}--\blpage{1327}
    (\byear{2014}).
    \doiurl{10.1007/s11118-014-9418-5}
    \end{barticle}
    \endbibitem

    %%% 19
    \bibitem{tanaka2019multi}
    \begin{barticle}
    \bauthor{\bsnm{Tanaka}, \binits{S.}},
    \bauthor{\bsnm{Watanabe}, \binits{T.}},
    \bauthor{\bsnm{Nagata}, \binits{K.}}:
    \batitle{Multi-particle model of coarse-grained scalar dissipation rate with
      volumetric tensor in turbulence}.
    \bjtitle{Journal of Computational Physics}
    \bvolume{389},
    \bfpage{128}--\blpage{146}
    (\byear{2019})
    \end{barticle}
    \endbibitem

    %%% 20
    \bibitem{friederich2021machine}
    \begin{barticle}
    \bauthor{\bsnm{Friederich}, \binits{P.}},
    \bauthor{\bsnm{H{\"a}se}, \binits{F.}},
    \bauthor{\bsnm{Proppe}, \binits{J.}},
    \bauthor{\bsnm{Aspuru-Guzik}, \binits{A.}}:
    \batitle{Machine-learned potentials for next-generation matter simulations}.
    \bjtitle{Nature Materials}
    \bvolume{20}(\bissue{6}),
    \bfpage{750}--\blpage{761}
    (\byear{2021})
    \end{barticle}
    \endbibitem

    %%% 21
    \bibitem{Bird1987}
    \begin{bbook}
    \bauthor{\bsnm{Bird}, \binits{R.B.}},
    \bauthor{\bsnm{Curtiss}, \binits{C.F.}},
    \bauthor{\bsnm{Armstrong}, \binits{R.C.}},
    \bauthor{\bsnm{Hassager}, \binits{O.}}:
    \bbtitle{Dynamics of Polymeric Liquids, Volume 2: Kinetic Theory}.
    \bpublisher{Wiley},
    \blocation{New {Y}ork}
    (\byear{1987})
    \end{bbook}
    \endbibitem

    %%% 22
    \bibitem{Doi1988}
    \begin{bbook}
    \bauthor{\bsnm{Doi}, \binits{M.}},
    \bauthor{\bsnm{Edwards}, \binits{S.F.}}:
    \bbtitle{The Theory of Polymer Dynamics}
    vol. \bseriesno{73}.
    \bpublisher{Oxford University Press},
    \blocation{New {Y}ork}
    (\byear{1988})
    \end{bbook}
    \endbibitem

    %%% 23
    \bibitem{Larson1999structure}
    \begin{bbook}
    \bauthor{\bsnm{Larson}, \binits{R.G.}}:
    \bbtitle{The Structure and Rheology of Complex Fluids}.
    \bpublisher{Oxford University Press},
    \blocation{New {Y}ork}
    (\byear{1999})
    \end{bbook}
    \endbibitem

    %%% 24
    \bibitem{McKinley2003}
    \begin{barticle}
    \bauthor{\bsnm{McKinley}, \binits{G.H.}},
    \bauthor{\bsnm{Sridhar}, \binits{T.}}:
    \batitle{Filament-stretching rheometry of complex fluids}.
    \bjtitle{Annual Review of Fluid Mechanics}
    \bvolume{34},
    \bfpage{375}--\blpage{415}
    (\byear{2003})
    \end{barticle}
    \endbibitem

    %%% 25
    \bibitem{Perkins1997}
    \begin{barticle}
    \bauthor{\bsnm{Perkins}, \binits{T.T.}},
    \bauthor{\bsnm{Smith}, \binits{D.E.}},
    \bauthor{\bsnm{Chu}, \binits{S.}}:
    \batitle{{Single polymer dynamics in an elongational flow}}.
    \bjtitle{Science}
    \bvolume{276}(\bissue{5321}),
    \bfpage{2016}--\blpage{2021}
    (\byear{1997})
    \end{barticle}
    \endbibitem

    %%% 26
    \bibitem{Smith1998}
    \begin{barticle}
    \bauthor{\bsnm{Smith}, \binits{D.E.}},
    \bauthor{\bsnm{Chu}, \binits{S.}}:
    \batitle{{Response of flexible polymers to a sudden elongational flow}}.
    \bjtitle{Science}
    \bvolume{281}(\bissue{5381}),
    \bfpage{1335}--\blpage{1340}
    (\byear{1998})
    \end{barticle}
    \endbibitem

    %%% 27
    \bibitem{Larson1999}
    \begin{barticle}
    \bauthor{\bsnm{Larson}, \binits{R.G.}}:
    \batitle{{Brownian dynamics simulations of a DNA molecule in an extensional
      flow field}}.
    \bjtitle{Journal of Rheology}
    \bvolume{43},
    \bfpage{267}
    (\byear{1999})
    \end{barticle}
    \endbibitem

    %%% 28
    \bibitem{Jendrejack2002}
    \begin{barticle}
    \bauthor{\bsnm{Jendrejack}, \binits{R.M.}},
    \bauthor{\bsnm{{De Pablo}}, \binits{J.J.}},
    \bauthor{\bsnm{Graham}, \binits{M.D.}}:
    \batitle{{Hydrodynamic interactions in long chain polymers: Application of the
      Chebyshev polynomial approximation in stochastic simulations}}.
    \bjtitle{The Journal of Chemical Physics}
    \bvolume{116},
    \bfpage{2894}
    (\byear{2002})
    \end{barticle}
    \endbibitem

    %%% 29
    \bibitem{Hsieh2003}
    \begin{barticle}
    \bauthor{\bsnm{Hsieh}, \binits{C.C.}},
    \bauthor{\bsnm{Li}, \binits{L.}},
    \bauthor{\bsnm{Larson}, \binits{R.G.}}:
    \batitle{Modeling hydrodynamic interaction in {Brownian} dynamics: simulations
      of extensional flows of dilute solutions of {DNA} and polystyrene}.
    \bjtitle{Journal of Non-Newtonian Fluid Mechanics}
    \bvolume{113}(\bissue{2-3}),
    \bfpage{147}--\blpage{191}
    (\byear{2003})
    \end{barticle}
    \endbibitem

    %%% 30
    \bibitem{sutton2018reinforcement}
    \begin{bbook}
    \bauthor{\bsnm{Sutton}, \binits{R.S.}},
    \bauthor{\bsnm{Barto}, \binits{A.G.}}:
    \bbtitle{Reinforcement Learning: An Introduction}.
    \bpublisher{MIT press},
    \blocation{Massachusetts}
    (\byear{2018})
    \end{bbook}
    \endbibitem

    %%% 31
    \bibitem{soh2018knots}
    \begin{barticle}
    \bauthor{\bsnm{Soh}, \binits{B.W.}},
    \bauthor{\bsnm{Narsimhan}, \binits{V.}},
    \bauthor{\bsnm{Klotz}, \binits{A.R.}},
    \bauthor{\bsnm{Doyle}, \binits{P.S.}}:
    \batitle{Knots modify the coil--stretch transition in linear dna polymers}.
    \bjtitle{Soft Matter}
    \bvolume{14}(\bissue{9}),
    \bfpage{1689}--\blpage{1698}
    (\byear{2018})
    \end{barticle}
    \endbibitem

    %%% 32
    \bibitem{murray2001mathematical}
    \begin{bbook}
    \bauthor{\bsnm{Murray}, \binits{J.D.}}:
    \bbtitle{Mathematical Biology II: Spatial Models and Biomedical Applications}
    vol. \bseriesno{3}.
    \bpublisher{Springer},
    \blocation{New York}
    (\byear{2001})
    \end{bbook}
    \endbibitem

    %%% 33
    \bibitem{Dobson2003}
    \begin{barticle}
    \bauthor{\bsnm{Dobson}, \binits{C.M.}}:
    \batitle{Protein folding and misfolding}.
    \bjtitle{Nature}
    \bvolume{426}(\bissue{6968}),
    \bfpage{884}--\blpage{890}
    (\byear{2003})
    \end{barticle}
    \endbibitem

    %%% 34
    \bibitem{whitesides2002self}
    \begin{barticle}
    \bauthor{\bsnm{Whitesides}, \binits{G.M.}},
    \bauthor{\bsnm{Grzybowski}, \binits{B.}}:
    \batitle{Self-assembly at all scales}.
    \bjtitle{Science}
    \bvolume{295}(\bissue{5564}),
    \bfpage{2418}--\blpage{2421}
    (\byear{2002})
    \end{barticle}
    \endbibitem

    %%% 35
    \bibitem{capito2008self}
    \begin{barticle}
    \bauthor{\bsnm{Capito}, \binits{R.M.}},
    \bauthor{\bsnm{Azevedo}, \binits{H.S.}},
    \bauthor{\bsnm{Velichko}, \binits{Y.S.}},
    \bauthor{\bsnm{Mata}, \binits{A.}},
    \bauthor{\bsnm{Stupp}, \binits{S.I.}}:
    \batitle{Self-assembly of large and small molecules into hierarchically ordered
      sacs and membranes}.
    \bjtitle{Science}
    \bvolume{319}(\bissue{5871}),
    \bfpage{1812}--\blpage{1816}
    (\byear{2008})
    \end{barticle}
    \endbibitem

    %%% 36
    \bibitem{cipelletti2005slow}
    \begin{barticle}
    \bauthor{\bsnm{Cipelletti}, \binits{L.}},
    \bauthor{\bsnm{Ramos}, \binits{L.}}:
    \batitle{Slow dynamics in glassy soft matter}.
    \bjtitle{Journal of Physics: Condensed Matter}
    \bvolume{17}(\bissue{6}),
    \bfpage{253}
    (\byear{2005})
    \end{barticle}
    \endbibitem

    %%% 37
    \bibitem{stillinger2013glass}
    \begin{barticle}
    \bauthor{\bsnm{Stillinger}, \binits{F.H.}},
    \bauthor{\bsnm{Debenedetti}, \binits{P.G.}}:
    \batitle{Glass transition thermodynamics and kinetics}.
    \bjtitle{Annu. Rev. Condens. Matter Phys.}
    \bvolume{4}(\bissue{1}),
    \bfpage{263}--\blpage{285}
    (\byear{2013})
    \end{barticle}
    \endbibitem

    %%% 38
    \bibitem{krivov2004hidden}
    \begin{barticle}
    \bauthor{\bsnm{Krivov}, \binits{S.V.}},
    \bauthor{\bsnm{Karplus}, \binits{M.}}:
    \batitle{Hidden complexity of free energy surfaces for peptide (protein)
      folding}.
    \bjtitle{Proceedings of the National Academy of Sciences}
    \bvolume{101}(\bissue{41}),
    \bfpage{14766}--\blpage{14770}
    (\byear{2004})
    \end{barticle}
    \endbibitem

    %%% 39
    \bibitem{settles2012.ActiveLearningVolume}
    \begin{botherref}
    \oauthor{\bsnm{Settles}, \binits{B.}}:
    Active {{Learning}}, volume 6 of {{Synthesis Lectures}} on {{Artificial
      Intelligence}} and {{Machine Learning}}.
    Morgan \& Claypool
    (2012)
    \end{botherref}
    \endbibitem

    %%% 40
    \bibitem{zhao2022.AdaptiveSamplingMethods}
    \begin{bchapter}
    \bauthor{\bsnm{Zhao}, \binits{Z.}},
    \bauthor{\bsnm{Li}, \binits{Q.}}:
    \bctitle{Adaptive sampling methods for learning dynamical systems}.
    In: \bbtitle{Proceedings of Mathematical and Scientific Machine Learning},
    pp. \bfpage{335}--\blpage{350}
    (\byear{2022}).
    \bcomment{PMLR}
    \end{bchapter}
    \endbibitem

    %%% 41
    \bibitem{nadler1987molecular}
    \begin{barticle}
    \bauthor{\bsnm{Nadler}, \binits{W.}},
    \bauthor{\bsnm{Br{\"u}nger}, \binits{A.T.}},
    \bauthor{\bsnm{Schulten}, \binits{K.}},
    \bauthor{\bsnm{Karplus}, \binits{M.}}:
    \batitle{Molecular and stochastic dynamics of proteins.}
    \bjtitle{Proceedings of the National Academy of Sciences}
    \bvolume{84}(\bissue{22}),
    \bfpage{7933}--\blpage{7937}
    (\byear{1987})
    \end{barticle}
    \endbibitem

    %%% 42
    \bibitem{milstein2002symplectic}
    \begin{barticle}
    \bauthor{\bsnm{Milstein}, \binits{G.N.}},
    \bauthor{\bsnm{Repin}, \binits{Y.M.}},
    \bauthor{\bsnm{Tretyakov}, \binits{M.V.}}:
    \batitle{Symplectic integration of hamiltonian systems with additive noise}.
    \bjtitle{SIAM Journal on Numerical Analysis}
    \bvolume{39}(\bissue{6}),
    \bfpage{2066}--\blpage{2088}
    (\byear{2002})
    \end{barticle}
    \endbibitem

    %%% 43
    \bibitem{goel1971volterra}
    \begin{barticle}
    \bauthor{\bsnm{Goel}, \binits{N.S.}},
    \bauthor{\bsnm{Maitra}, \binits{S.C.}},
    \bauthor{\bsnm{Montroll}, \binits{E.W.}}:
    \batitle{On the volterra and other nonlinear models of interacting
      populations}.
    \bjtitle{Reviews of modern physics}
    \bvolume{43}(\bissue{2}),
    \bfpage{231}
    (\byear{1971})
    \end{barticle}
    \endbibitem

    %%% 44
    \bibitem{tuckerman2010statistical}
    \begin{bbook}
    \bauthor{\bsnm{Tuckerman}, \binits{M.}}:
    \bbtitle{Statistical Mechanics: Theory and Molecular Simulation}.
    \bpublisher{Oxford university press},
    \blocation{Oxford}
    (\byear{2010})
    \end{bbook}
    \endbibitem

    %%% 45
    \bibitem{beris1994thermodynamics}
    \begin{bbook}
    \bauthor{\bsnm{Beris}, \binits{A.N.}},
    \bauthor{\bsnm{Edwards}, \binits{B.J.}}:
    \bbtitle{Thermodynamics of Flowing Systems: with Internal Microstructure}
    vol. \bseriesno{36}.
    \bpublisher{Oxford University Press on Demand},
    \blocation{New York}
    (\byear{1994})
    \end{bbook}
    \endbibitem

    %%% 46
    \bibitem{li2019better}
    \begin{barticle}
    \bauthor{\bsnm{Li}, \binits{B.}},
    \bauthor{\bsnm{Tang}, \binits{S.}},
    \bauthor{\bsnm{Yu}, \binits{H.}}:
    \batitle{Better approximations of high dimensional smooth functions by deep
      neural networks with rectified power units}.
    \bjtitle{Communications in Computational Physics}
    \bvolume{27}(\bissue{2}),
    \bfpage{379}--\blpage{411}
    (\byear{2020})
    \end{barticle}
    \endbibitem

    %%% 47
    \bibitem{Vologodskii2006}
    \begin{barticle}
    \bauthor{\bsnm{Vologodskii}, \binits{A.}}:
    \batitle{Brownian dynamics simulation of knot diffusion along a stretched {DNA}
      molecule}.
    \bjtitle{Biophysical Journal}
    \bvolume{90}(\bissue{5}),
    \bfpage{1594}--\blpage{1597}
    (\byear{2006})
    \end{barticle}
    \endbibitem

    %%% 48
    \bibitem{Liu1989}
    \begin{barticle}
    \bauthor{\bsnm{Liu}, \binits{T.W.}}:
    \batitle{Flexible polymer chain dynamics and rheological properties in steady
      flows}.
    \bjtitle{The Journal of Chemical Physics}
    \bvolume{90}(\bissue{10}),
    \bfpage{5826}--\blpage{5842}
    (\byear{1989})
    \end{barticle}
    \endbibitem

    %%% 49
    \bibitem{Soh2023}
    \begin{barticle}
    \bauthor{\bsnm{Soh}, \binits{B.W.}},
    \bauthor{\bsnm{Ooi}, \binits{Z.-E.}},
    \bauthor{\bsnm{Vissol-Gaudin}, \binits{E.}},
    \bauthor{\bsnm{Leong}, \binits{C.J.}},
    \bauthor{\bsnm{Hippalgaonkar}, \binits{K.}}:
    \batitle{Automated electrokinetic stretcher for manipulating nanomaterials}.
    \bjtitle{Lab on a Chip}
    \bvolume{23},
    \bfpage{3716}--\blpage{3726}
    (\byear{2023})
    \end{barticle}
    \endbibitem

    %%% 50
    \bibitem{strogatz2018nonlinear}
    \begin{bbook}
    \bauthor{\bsnm{Strogatz}, \binits{S.H.}}:
    \bbtitle{Nonlinear Dynamics and Chaos: with Applications to Physics, Biology,
      Chemistry, and Engineering}.
    \bpublisher{CRC press},
    \blocation{Boca Raton}
    (\byear{2018})
    \end{bbook}
    \endbibitem

    %%% 51
    \bibitem{soh_2023_data}
    \begin{barticle}
    \bauthor{\bsnm{Soh}, \binits{B.}}:
    \batitle{{Data for Constructing Custom Thermodynamics Using Deep Learning}}.
    \bjtitle{Harvard Dataverse}
    (\byear{2023}).
    \doiurl{10.7910/DVN/NRIX7Y}
    \end{barticle}
    \endbibitem

    %%% 52
    \bibitem{chen_2023_10212239}
    \begin{barticle}
    \bauthor{\bsnm{Chen}, \binits{X.}},
    \bauthor{\bsnm{Soh}, \binits{B.W.}},
    \bauthor{\bsnm{Ooi}, \binits{Z.-e.}},
    \bauthor{\bsnm{Vissol-Gaudin}, \binits{E.}},
    \bauthor{\bsnm{Yu}, \binits{H.}},
    \bauthor{\bsnm{Novoselov}, \binits{K.S.}},
    \bauthor{\bsnm{Hippalgaonkar}, \binits{K.}},
    \bauthor{\bsnm{Li}, \binits{Q.}}:
    \batitle{{Constructing Custom Thermodynamics Using Deep Learning}}.
    \bjtitle{Zenodo}
    (\byear{2023}).
    \doiurl{10.5281/zenodo.10212239}
    \end{barticle}
    \endbibitem
\end{thebibliography}
